\documentclass[11pt,a4paper,reqno]{amsart}
\usepackage{amscd}
\usepackage{amssymb}
\usepackage{amsthm}
\usepackage[centering,text={15.5cm,22cm}]{geometry}
\usepackage{graphicx}
\usepackage{color}
\usepackage[all]{xy}
\usepackage{mathrsfs}
\definecolor{amaranth}{rgb}{0.9, 0.17, 0.31}
\usepackage[colorlinks=true,citecolor=amaranth,linkcolor=black]{hyperref}
\usepackage{mathtools}

\newtheorem{theorem}{Theorem}[section]
\newtheorem{lemma}[theorem]{Lemma}

\newtheoremstyle{defstyle}
  {.6em} 
  {.1em} 
  {} 
  {} 
  {\bfseries} 
  {.} 
  {.5em} 
  {} 
\theoremstyle{defstyle} \newtheorem{definition}[theorem]{Definition}

\newtheorem{construction}[theorem]{Construction}

\newtheorem{example}[theorem]{Example}

\theoremstyle{remark}
\newtheorem{remark}[theorem]{Remark}

\numberwithin{equation}{section}

\newcommand{\ZZ} {\mathbb{Z}}
\newcommand{\QQ} {\mathbb{Q}}
\newcommand{\RR} {\mathbb{R}}
\newcommand{\CC} {\mathbb{C}}
\newcommand{\PP} {\mathbb{P}}

\newcommand{\cO} {\mathcal{O}}

\newcommand {\fod}  {\mathfrak{d}}

\newcommand{\cX}{\mathcal{X}}
\newcommand{\cP}{\mathcal{P}}

\newcommand{\hooklongrightarrow}{\lhook\joinrel\longrightarrow}

\begin{document}

\title[Non-toric brane webs, Calabi--Yau $3$-folds, and 5d SCFTs]{\small Non-toric brane webs, Calabi--Yau 3-folds, and 5d SCFT\lowercase{s}}

\author[V.\,Alexeev]{Valery Alexeev}
\address{University of Georgia, Department of Mathematics, Athens, GA 30605}
\email{Valery.Alexeev@uga.edu}

\author[H.\,Arg\"uz]{H\"ulya Arg\"uz}
\address{The Mathematical Institute, University of Oxford, Oxford, OX2 6GG, UK}
\email{hulya.arguz@maths.ox.ac.uk}

\author[P.\,Bousseau]{Pierrick Bousseau}
\address{The Mathematical Institute, University of Oxford, Oxford, OX2 6GG, UK}
\email{pierrick.bousseau@maths.ox.ac.uk}

\begin{abstract} 
We study webs of 5-branes with 7-branes in Type IIB string theory from a geometric perspective. Mathematically, a web of 5-branes with 7-branes is a tropical curve in $\RR^2$ with focus-focus singularities introduced.
To any such a web $W$, we attach a log Calabi--Yau surface $(Y,D)$ with a line bundle $L$. 
We then describe supersymmetric webs, which are webs defining 5d superconformal field theories (SCFTs), in terms of the geometry of $(Y,D,L)$. 
We also introduce particular supersymmetric webs called ``consistent webs", and show that any 5d SCFT defined by a supersymmetric web can be obtained from a consistent web by adding free hypermultiplets. 
Using birational geometry of degenerations of log Calabi--Yau surfaces, we provide an algorithm to test the consistency of a web in terms of its dual polygon. 
Moreover, for a consistent web $W$, we provide an algebro-geometric construction of the mirror $\mathcal{X}^{\mathrm{can}}$ to $(Y,D,L)$, as a non-toric canonical 3-fold singularity, and show that M-theory on $\mathcal{X}^{\mathrm{can}}$ engineers the same 5d SCFT as $W$. 
We also explain how to derive explicit equations for $\mathcal{X}^{\mathrm{can}}$ using scattering diagrams, encoding disk worldsheet instantons in the A-model, or equivalently the BPS states of an auxiliary rank one 4d $\mathcal{N}=2$ theory.
\end{abstract}

\maketitle

\setcounter{tocdepth}{1}
\tableofcontents
\section{Introduction}

\subsection{Background and context}
One of the most remarkable predictions of string/M-theory is the existence of 5-dimensional superconformal field theories (5d SCFTs) \cite{seiberg1996five}. Constructing such SCFTs has been of significant interest, and two powerful approaches for doing so are provided by M-theory on canonical $3$-fold singularities on one hand \cite{MR4069495, MR4090085, MR4176441, douglas_katz_vafa, intriligator, morrison_seiberg, xie_yau}, and intersecting branes in Type IIB string theory on the other \cite{aharony1998webs}. A natural question is to compare these two approaches. The aim of this paper is to address this question: given a 5d SCFT engineered by a configuration of branes in Type IIB string theory, can we obtain the same 5d SCFT from M-theory on a dual canonical 3-fold singularity? Conversely, can we give an algebro-geometric description of the class of canonical 3-fold singularities that admit a dual Type IIB string description in terms of intersecting branes?

The answer to these questions is well-known when the configuration of branes in Type IIB string theory is just a web of 5-branes, mathematically described by a tropical curve in $\RR^2$ \cite{aharony1998webs, leung_vafa}. 
In that setting, the dual canonical 3-fold singularity is a Gorenstein toric affine Calabi--Yau 3-fold, which has an explicit combinatorial description, obtained from the web of 5-branes. 
Moreover, every Gorenstein toric affine canonical 3-fold singularity arises this way. 
However, answering these questions is significantly more challenging, when considering $5$-branes together with $7$-branes, mathematically described by a tropical curve in $\RR^2$ with focus-focus singularities introduced \cite{BBT, DeWolfe}. In this situation, while the existence of non-toric dual canonical 3-fold singularities is still predicted via string dualities \cite{BBT}, giving a precise algebro-geometric construction of such singularities has been a longstanding open question -- see \cite{arias2024geometry, bourget2023generalized} for recent progress in this direction. The present paper provides a complete answer to this question, building on recent mathematical advances on mirror symmetry for log Calabi--Yau surfaces. 

\subsection{Outline of the paper}

\subsubsection{Webs of 5-branes and toric mirror symmetry}
In \S\ref{section_webs_5_toric}, we review the known construction of the M-theory dual toric canonical 3-fold singularity $\overline{\mathcal{X}}^{\mathrm{can}}$ to an asymptotic web of 5-branes $\overline{W}^{\mathrm{asym}}$ \cite{aharony1998webs, leung_vafa}. In particular, we explain how to view $\overline{\mathcal{X}}^{\mathrm{can}}$ as the mirror to the polarized toric surface $(\overline{Y},\overline{D},\overline{L})$ with momentum polytope $\overline{P}$ dual to $\overline{W}^{\mathrm{asym}}$, where $\overline{Y}$ is a projective toric surface, $\overline{D}$ is the toric boundary divisor, and $\overline{L}$ is an ample line bundle on $\overline{Y}$. Mirrors to crepant resolutions of $\overline{\mathcal{X}}^{\mathrm{can}}$ are then given by maximal polarized toric degenerations of $(\overline{Y},\overline{D},\overline{L})$ associated to webs $\overline{W}$ with asymptotic web $\overline{W}^{\mathrm{asym}}$ -- see Figure \ref{Fig18}. 
In this paper, we provide a natural generalization of this mirror construction to the case of webs of 5-branes with 7-branes.

\subsubsection{Webs of 5-branes with 7-branes and log Calabi--Yau surfaces}
An asymptotic web of 5-branes with 7-branes $W^{\mathrm{asym}}$ is obtained by introducing 7-branes in an asymptotic web of 5-branes $\overline{W}^{\mathrm{asym}}$.
In \S\ref{section_webs_log_cy}, we construct a log Calabi--Yau surface $(Y,D)$ associated to $W^{\mathrm{asym}}$, which consists of a projective surface $Y$ together with an anticanonical divisor $D$.
Denoting by $(\overline{Y},\overline{D},\overline{L})$ the polarized toric surface associated to $\overline{W}^{\mathrm{asym}}$, we 
obtain $Y$
by blowing up a smooth point on $\overline{D} \subset \overline{Y}$ for each 7-brane, and let $D$ to be the strict transform of $\overline{D}$ in $Y$. We also define a line bundle $L$ on $Y$ from the data of $\overline{L}$ and the configuration of the 7-branes. Then, building on 
\cite{GHK_birational, GHK1},
we explain that different webs of 5-branes with 7-branes related by Hanany--Witten moves, define the same log Calabi--Yau surface with line bundle $(Y,D,L)$, up to blow-ups of the 0-dimensional strata of $D$, but correspond to different ways to present $(Y,D)$ as a blow-up 
\begin{equation} \label{eq_p}
p:(Y,D) \longrightarrow (\overline{Y},\overline{D})
\end{equation}
of a toric surface.

\subsubsection{Consistent webs of 5-branes with 7-branes and birational geometry of degenerations}
In \S\ref{section_consistent_webs_57}, we first describe \emph{supersymmetric} webs. We call a web $W^{\mathrm{asym}}$ supersymmetric if there exists a curve $C$ in $Y$ in the linear systems of the line bundle $|L|$,
that does not pass through the 0-dimensional strata of $D$. The curve $C^\circ =C \cap U$ in $U=Y \setminus D$ is the Seiberg--Witten curve of the 4d $\mathcal{N}=2$ theory obtained by compactifying on $S^1$ the 5d SCFT defined by $W^{\mathrm{asym}}$.  
A crucial observation is that the curve $C$ is not necessarily connected. We show that the existence of connected components of $C$ with negative self-intersection is related to the fact that the web of 5-branes with 7-branes can become disconnected after applying a sequence of Hanany--Witten moves. We say that a supersymmetric asymptotic web $W^{\mathrm{asym}}$ is \emph{consistent} if this does not happen, and all connected components of $C$ have nonnegative self-intersection. Moreover, we show that any 5d SCFT defined by a supersymmetric web can be obtained from a consistent web by adding free hypermultiplets. Hence, rather than studying general supersymmetric asymptotic webs we focus attention on consistent ones.

We define a notion of consistency for webs $W$ obtained as deformation of an asymptotic web $W^{\mathrm{asym}}$. To do this, we first associate to $W$ a degeneration of log Calabi--Yau surfaces 
\begin{equation} \label{eq_nu}
\nu: (\mathcal{Y},\mathcal{D},\mathcal{L}) \longrightarrow \CC\,,
\end{equation}
with general fiber $(Y,D,L)$. This degeneration is obtained by blow-ups from the toric degeneration 
\begin{equation} \label{eq_nu_bar}
\overline{\nu}: (\overline{\mathcal{Y}},\overline{\mathcal{D}},\overline{\mathcal{L}}) \longrightarrow \CC\,,
\end{equation}
defined by the web of 5-branes $\overline{W}$ obtained from $W$ by removing all 7-branes.
The central fiber $(\mathcal{Y}_0, \mathcal{D}_0, \mathcal{L}_0)$ of $\nu$
has dual intersection complex $W$, and hence its irreducible components $(Y_0^v, \partial Y_0^v)$ are indexed by the vertices $v$ of $W$. We call $W$ \emph{consistent} if $W^{\mathrm{asym}}$ 
is consistent and there exists a family 
of curves $(C_t) \in |\mathcal{L}|$, such that the curves $C_0^v=C_0 \cap Y_0^v$ in the central fiber do not pass through the 0-dimensional strata of $\partial Y_0^v$. 

Although the consistency of $W^{\mathrm{asym}}$ implies that the line bundle $L$ is nef, that is, has positive intersection with every curve in $Y$, the consistency of $W$ does not imply in general that the line bundle $\mathcal{L}$ is nef. However, as a special case of the minimal model program for degenerations, we show that it is always possible to modify $\nu$ by a sequence of flops into a new degeneration 
\begin{equation} \label{eq_nu_prime}
\nu': (\mathcal{Y}',\mathcal{D}',\mathcal{L}') \longrightarrow \CC\,,
\end{equation}
such that $\mathcal{L}'$ is nef. We explain that such flops amount to modifying $W$ into a new web $W'$ obtained by pushing the 7-branes. We prove that if $W$ is generic among consistent webs, then the corresponding modified web $W'$ is obtained by fully pushing the 7-branes until they are no longer attached to any 5-brane. 

\subsubsection{Consistent decorated polygons and Symington polygons}

In \S\ref{sec_consistent_DTP}, we reformulate the previous results on webs in the dual language of \emph{decorated toric polygons}, which are a slight variant of the Generalized Toric Polygons (GTPs) previously studied in the literature \cite{BBT, bourget2023generalized}. First, we provide an algorithm to determine when a decorated polygon with a polyhedral decomposition is \emph{consistent}, that is, dual to a generic consistent web of 5-branes with 7-branes. Then, we show that pushing the 7-branes in the web is dual to cutting and gluing operations on the decorated toric polygon, resulting in an integral affine manifold with singularity known as a \emph{Symington polygon}. Finally, we compare our notion of a consistent decorated polygon with existing s-rules and r-rules for GTPs \cite{BBT, van2020symplectic}. 
In particular, we describe in Example \ref{example_gtp} a GTP that satisfies the s-rule and r-rule of \cite{van2020symplectic} but does not correspond to a consistent decorated polygon. The corresponding web is still supersymmetric, and so defines a 5d SCFT, but the associated Seiberg--Witten curve is not connected. In particular, the rank of the 5d SCFT differs from the rank of the GTP as defined in \cite{van2020symplectic}.

\subsubsection{M-theory dual Calabi--Yau 3-folds to consistent webs of 5-branes with 7-branes}

In \S\ref{section_cy3}, we construct for every consistent asymptotic web $W^{\mathrm{asym}}$ of 5-branes with 7-branes a M-theory dual canonical 3-fold singularity $\mathcal{X}^{\mathrm{can}}$ as the mirror of the corresponding log Calabi--Yau surface with line bundle $(Y,D,L)$. This builds on recent mathematical works on mirror symmetry for log Calabi--Yau surfaces \cite{alexeev2024ksba, AAB2, engel, engel_friedman, GHK1, GScanonical, hacking2020secondary}. 
To describe $\mathcal{X}^{\mathrm{can}}$, we first consider a generic consistent web $W$ obtained by deformation of $W^{\mathrm{asym}}$. We let $\overline{W}$ be the web of 5-branes obtained from $W$ by removing the 7-branes, and denote by
\[ \overline{\pi}: \overline{\mathcal{X}} \longrightarrow \CC\]
the corresponding toric Calabi--Yau 3-fold. The pushings of 7-branes in the web $W$, or equivalently the cuts and gluings producing the dual Symington polygon, give a recipe on how to deform
non-torically the union of toric surfaces $\overline{\mathcal{X}}_0:=\overline{\pi}^{-1}(0)$ into a normal crossing union of log Calabi--Yau surfaces $\mathcal{X}_0$ -- see \S\ref{section_cy3_construction} for details. We obtain a smooth Calabi--Yau 3-fold 
\[ \pi: \mathcal{X} \longrightarrow \Delta\]
as the total space of a one-parameter smoothing of $\mathcal{X}_0$ over a disk $\Delta \subset \CC$ containing $0 \in \CC$.
The non-toric deformation from $\overline{\pi}:\overline{\mathcal{X}} \rightarrow \CC$ to $\pi: \mathcal{X} \rightarrow \Delta$ is mirror to the birational map relating the degeneration of log Calabi--Yau surface $\nu'$ in Equation \eqref{eq_nu_prime} and the toric degeneration $\overline{\nu}$ in Equation \eqref{eq_nu_bar}.
Finally, there exists a contraction 
\[ f:\mathcal{X} \longrightarrow \mathcal{X}^{\mathrm{can}} \,,\]
to a canonical 3-fold singularity $\mathcal{X}^{\mathrm{can}}$, equipped with a map 
\[\pi^{\mathrm{can}}:\mathcal{X}^{\mathrm{can}}\longrightarrow \Delta\]
such that $\pi=\pi^{\mathrm{can}}\circ f$. Moreover, $\mathcal{X}^{\mathrm{can}}$ is a non-toric deformation of the toric canonical 3-fold singularity $\overline{\mathcal{X}}^{\mathrm{can}}$ corresponding to the web of 5-branes $\overline{W}^{\mathrm{asym}}$ obtained from $W^{\mathrm{asym}}$ by removing the 7-branes. This non-toric deformation from $\overline{\mathcal{X}}^{\mathrm{can}}$ to $\mathcal{X}^{\mathrm{can}}$ is mirror to the blow-up map $p: (Y,D) \rightarrow (\overline{Y}, \overline{D})$ in Equation \eqref{eq_p}. We show that the central fiber $\mathcal{X}^{\mathrm{can}}_0=(\pi^{\mathrm{can}})^{-1}(0)$ is either a degenerate cusp, cusp, or simple elliptic singularity. Conversely, every canonical 3-fold singularity obtained as total space of a one-parameter smoothing of such a singularity is M-theory dual to a consistent web of 5-branes with 7-branes.

In \S\ref{section_scattering}, we explain how to describe $\mathcal{X}^{\mathrm{can}}$ explicitly using scattering diagrams, as in \cite{GHK1,GScanonical}. Physically, scattering diagrams capture disk worldsheet instantons corrections in the A-model side of mirror symmetry. In the particular situation, when GTPs have all 7-branes on a single edge, we show in Example \ref{example_comparison} that our construction recovers the equations for $\mathcal{X}^{\mathrm{can}}$ derived in \cite{bourget2023generalized} using D6-branes in Type IIA string theory. This situation is relatively easy to handle, as the initial disks coming out of the 7-branes do not interact.  More generally, the initial disks can interact in a very non-trivial way, and so obtaining closed-form expressions for the equations of $\mathcal{X}^{\mathrm{can }}$ is challenging in general, due to the complexity of the combinatorics of scattering diagrams. Nevertheless, we describe the scattering diagrams and explicit equations for $\mathcal{X}^{\mathrm{can }}$ in some non-trivial situations -- see Examples \ref{example_A2}-\ref{example_cubic}. 

Finally, we explain that the disk worldsheet instanton corrections have a dual interpretation as BPS states of an auxiliary rank one 4d $\mathcal{N}=2$ theory living on a D3-brane probing the configuration of 7-branes in Type IIB string theory. 
In particular, the complexity of the disk worldsheet instantons  is an illustration of the known intricacy of the BPS spectrum of 4d $\mathcal{N}=2$ theories, as illustrated in  \cite{galakhov2013wild} for example.

We provide additional examples of consistent web of 5-branes with 7-branes, Symington polygons, and M-theory dual Calabi--Yau 3-folds in \S\ref{section_further_examples}, where we also explain the relation with mirror symmetry for Fano orbifolds as studied in \cite{fanosearch, polygon_mutations, corti2023cluster}. 

\subsection{Future aspects}

\subsubsection{Orientifolds and S-folds} More general webs of branes in Type IIB string theory can be obtained by adding orientifolds and S-folds to webs of 5-branes with 7-branes. We expect the corresponding M-theory dual Calabi--Yau 3-folds to be total spaces of one-parameter smoothings of $\QQ$-Gorenstein semi-log-canonical surface singularities -- see the end of \S\ref{section_slc_sing}. We leave for future work to explore this correspondence.

\subsubsection{Topological vertex and enumerative geometry} It is proposed in the physics literature, that the Nekrasov partition function of the compactifications on $S^1$ of the 5d SCFTs engineered by webs of 5-branes with 7-branes, can be calculated by generalizing the topological vertex formalism -- see for instance \cite{MR3429464}. The Nekrasov partition function is expected to capture the Gromov--Witten invariants and refined stable pair invariants of the M-theory dual Calabi--Yau 3-folds. It would be interesting to investigate if the enumerative geometry of the non-compact Calabi--Yau 3-folds we construct, is indeed captured by a generalization of the topological vertex formalism.

\subsubsection{Quivers and derived categories}
Recently, various suggestions were made for constructing BPS quivers for the compactification on $S^1$ of the 5d SCFTs engineered by webs of 5-branes with 7-branes, generalizing the known construction via dimer models for webs of 5-branes -- see for instance \cite{Franco_twin}. Mathematically, BPS quivers are  expected to describe the derived category of coherent sheaves on the M-theory dual Calabi--Yau 3-folds. It would be interesting to describe the derived categories of coherent sheaves of the Calabi--Yau 3-folds we construct in this paper, to provide a geometric path for constructing BPS quivers.

\subsection{Acknowledgments} 
Valery Alexeev was partially supported by the NSF grant DMS-2201222.
H\"ulya Arg\"uz was partially supported by the NSF grant DMS-2302116, and Pierrick Bousseau was partially supported by the NSF grant DMS-2302117. This work was presented at the String--Math 2024 hosted by ICTP, the LMS Lecture Series 2024 hosted by Imperial College London, and the Geometry and Physics of Mirror Symmetry 2024 Conference hosted by the University of Sheffield -- we thank all the organizers for providing excellent occasions.

\section{Webs of 5-branes and toric mirror symmetry}
\label{section_webs_5_toric}

In this section, we review the well-known correspondence between webs of 5-branes and toric Calabi--Yau 3-folds \cite{aharony1998webs, leung_vafa}, and explain how it can be understood via toric mirror symmetry.

\subsection{Webs of $5$-branes, $5$d SCFTs, and dual polygons} \label{section_tropical_web}

\subsubsection{Webs of $5$-branes and $5$d SCFTs}
We review below the definition of a web of 5-branes, as an embedded graph in $\RR^2$ satisfying a balancing condition at its vertices -- such graphs are often referred to as \emph{tropical curves} (see \cite{mikhalkin, NS}). We denote the sets of vertices, edges (connecting two vertices), and legs (adjacent to a single vertex) of a graph $\overline{W}$, by $V(\overline{W})$, $E(\overline{W})$ and $L(\overline{W})$ respectively.
\begin{definition}
\label{def_web_5_branes}
A \emph{web of 5-branes} is an embedded graph $\overline{W} \subset \RR^2$, such that every edge or leg $e \in E(\overline{W}) \cup L(\overline{W})$ is endowed with a weight $w_e \in \ZZ_{\geq 1}$, and such that the following conditions are satisfied:
\begin{itemize}
        \item[i)] Every edge or leg of $\overline{W}$ is contained in an affine line of rational slope in $\RR^2$.
    \item[ii)] At every vertex $v$, the following balancing condition holds:
\[ \sum_{\substack{e \in E(\overline{W}) \cup L(\overline{W})\\ v \in e}} w_e u_e =0 \,,\]
where the sum is over the edges or legs $e$ adjacent to $v$, and $u_e \in \ZZ^2$ is the primitive integral vector in the direction of $e$ pointing away from $v$.
\end{itemize}
\end{definition}

Recall that Type IIB string theory contains 5-branes of type $(p,q)$ for every coprime $(p,q)\in \ZZ^2 \setminus \{0\}$ \cite{Polchinski}.
A web of 5-branes $\overline{W} \subset \RR^2$ as in Definition \ref{def_web_5_branes} determines a configuration of 5-branes in Type IIB string theory on 
\[ \RR^{1,4} \times \RR^2 \times \RR^3 \]
as follows.
For every edge or leg $e$ of $\overline{W}$, we place $w_e$ 5-branes of type $(p_e,q_e)$ on $\RR^{1,4} \times e \subset \RR^{1,4} \times \RR^2$, where $(p_e,q_e) \in \ZZ^2$ is the primitive integral direction of $e$. It follows from the balancing condition that this configuration of branes is supersymmetric. 

In the following sections, we frequently work with particular webs of $5$-branes, referred to as ``asymptotic'' and defined as follows.

\begin{definition} \label{def_asymptotic_5}
An \emph{asymptotic web of 5-branes} $\overline{W}^{\mathrm{asym}} \subset \RR^2$ is a web of 5-branes with a unique vertex at the origin $0 \in \RR^2$.
\end{definition}

In particular, an asymptotic web of 5-branes consists only of legs intersecting at the origin. An asymptotic web of 5-branes $\overline{W}^{\mathrm{asym}}$ is expected to define a 5d $\mathcal{N}=1$ SCFT describing the low energy dynamics of the light degrees of freedom localized on the common intersection $\RR^{1,4} \times \{0\}$ of the 5-branes.

\begin{figure}[hbt!]
\center{\scalebox{0.9}{\includegraphics{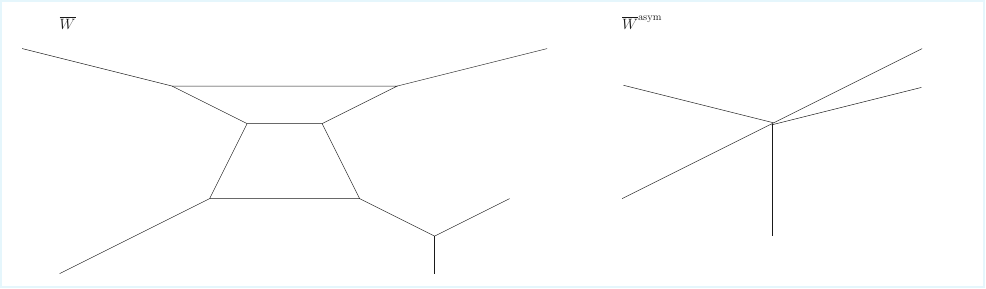}}}
\caption{A web of $5$-branes $\overline{W}$ on the left, and the associated asymptotic web of $5$-branes $\overline{W}^{\mathrm{asym}}$ on the right.}
\label{Fig18}
\end{figure}

Starting with any web of 5-branes $\overline{W}$, we obtain an associated asymptotic web $\overline{W}^{\mathrm{asym}}$, by first scaling $\overline{W}$ by $t>0$ and then taking the limit as $t \rightarrow 0$ -- this amounts to contracting all edges on $\overline{W}$, as illustrated in Figure \ref{Fig18}.
The webs $\overline{W}$ with associated asymptotic web $\overline{W}^{\mathrm{asym}}$ can be viewed as deformations of $\overline{W}^{\mathrm{asym}}$.  
From the viewpoint of the 5d SCFT, deforming the web of branes from $\overline{W}^{\mathrm{asym}}$ to $\overline{W}$ amounts to deforming the theory either by mass parameters or by moving on the Coulomb branch of the theory -- see \cite{MR4127742} for an exposition. 
The union of the Coulomb branches of all the massive deformations of the 5d SCFT can be described as the moduli space of all webs $\overline{W}$ with asymptotic web $\overline{W}^{\mathrm{asym}}$ and considered up to translations of $\RR^2$. 

\subsubsection{The dual polygon to an asymptotic web of $5$-branes} \label{sec_dual_polygon_convex}
The data of an asymptotic web of 5-branes $\overline{W}^{\mathrm{asym}}$, and so of the corresponding 5d SCFT, can be encoded in the data of a dual (compact, convex) lattice polygon $\overline{P}$. For every leg $e$ of $\overline{W}^{\mathrm{asym}}$, with weight $w_e$ and primitive integral direction $(p_e, q_e)$, $\overline{P}$ has a side of direction $(-q_e,p_e)$ and lattice length $w_e$. 
Such polygon $\overline{P}$ is uniquely determined up to integral affine linear transformations.
Conversely, every lattice polygon is dual to an asymptotic web of 5-branes. 
\begin{figure}[hbt!]
\center{\scalebox{.9}{\includegraphics{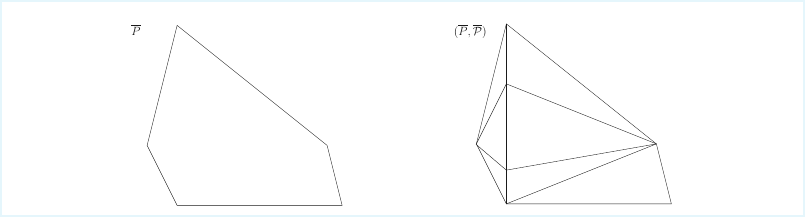}}}
\caption{The polygon $\overline{P}$ dual to the web $\overline{W}^{\mathrm{asym}}$, with polyhedral decomposition $\overline{\mathcal{P}}$ dual to the web $\overline{W}$ in Figure \ref{Fig18}}
\label{Fig20:polygons}
\end{figure}
Throughout this paper, we will assume that we are not in the degenerate situation where $\overline{W}^{\mathrm{asym}}$ is contained in a line and $\overline{P}$ is a line segment -- elsewise there won't be any intersecting branes, thus no non-trivial $5$d SCFT.

Webs of 5-branes $\overline{W}$ with associated asymptotic web  
$\overline{W}^{\mathrm{asym}}$
are dual to integral polyhedral decompositions $\mathcal{\overline{P}}$ of the polygon $\overline{P}$ 
dual to $\overline{W}^{\mathrm{asym}}$, which are \emph{regular} in the sense that their faces are the domains of linearity of a convex continuous piecewise linear function on $\overline{P}$. We will often denote by 
$(\overline{P},\overline{\mathcal{P}})$ a lattice polygon $\overline{P}$ endowed with a regular polyhedral decomposition
$\overline{\mathcal{P}}$.
For example, webs of 5-branes that are obtained as generic deformations of
$\overline{W}$
are 3-valent and are dual to regular integral triangulations of $\overline{P}$ into lattice triangles of size one, that is, with three sides of integral length one. Note that a lattice triangle of size one is integral-affinely equivalent to the standard simplex in $\ZZ^2$ with vertices $(0,0)$, $(1,0)$, and $(0,1)$.

\subsection{Asymptotic webs of 5-branes, Calabi--Yau $3$-folds, and polarized toric surfaces} \label{webs_toric_surfaces}

\subsubsection{Asymptotic webs of $5$-branes and Calabi--Yau $3$-folds}
\label{subsec: asymptotic to CY3}
Let $\overline{W}^{\mathrm{asym}}$ be an asymptotic web of 5-branes and $\overline{P}$ the dual lattice polytope as described in \S \ref{section_tropical_web}. 
It follows from the duality between Type IIB string theory on $S^1$
and M-theory on $T^2$ that the 4d $\mathcal{N}=2$ theory obtained by compactifying on $S^1$ the 5d SCFT defined by $\overline{P}$ has a low-energy dual description given by the worldvolume theory of an M5-brane on $\RR^{1,3} \times \overline{C}^\circ$ in M-theory on $\RR^{1,3} \times  (\CC^\star)^2 \times \RR^3$, 
where $\overline{C}^\circ$ is the algebraic curve
\[ \overline{C}^\circ :=\{ f=0\} \subset (\CC^\star)^2_{x,y}\,.\]
Here, $f \in \CC[x^\pm, y^\pm]$ is a Laurent polynomial with Newton polygon $\overline{P}$, that is, of the form 
\begin{equation}\label{eq_laurent}
 f=\sum_{(a,b)\in \overline{P}_\ZZ} c_{a,b}x^a y^b \,,
\end{equation}
where by $\overline{P}_\ZZ$ we denote the set of integral points in $\overline{P}$. In other words, $\overline{C}^\circ \subset (\CC^\star)^2$ is the Seiberg--Witten curve of this 4d $\mathcal{N}=2$ theory.

Using the duality between M-theory on $S^1$ and Type IIA string theory, and then T-duality between Type IIA and Type IIB string theories, we obtain another dual description of the same 4d $\mathcal{N}=2$ theory as Type IIB string theory on the non-compact 
Calabi--Yau 3-fold 
\begin{equation}\label{eq_Z}
\overline{Z} :=\{ uv=f \} \subset \CC^2_{u,v} \times (\CC^\star)^2\,.\end{equation} 
The 3-fold $\overline{Z}$ is an affine conic bundle over $(\CC^\star)^2$,  which is degenerate over the algebraic curve $\overline{C}^\circ$,
as illustrated in Figure \ref{Fig21}.
\begin{figure}[hbt!]
\center{\scalebox{1.0}{\includegraphics{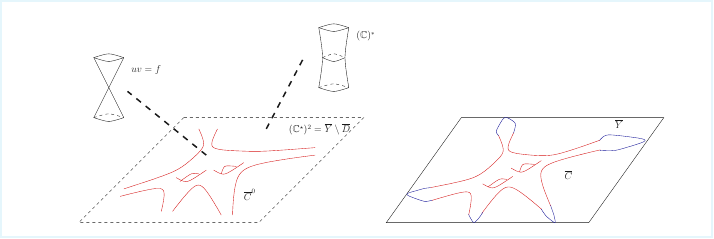}}}
\caption{The Calabi--Yau 3-fold $\overline{Z}$ on the left, and the compactification $\overline{C}$ of $\overline{C}^{\circ}$ in $Y$ on the right.}
\label{Fig21}
\end{figure}

\subsubsection{Asymptotic webs of $5$-branes and polarized toric surfaces}
\label{section_toric_surfaces}

The polygon $\overline{P}$ also determines a toric compactification of $(\CC^\star)^2$.
Indeed, by toric geometry \cite{cox_toric}, $\overline{P}$ is the momentum polytope of a polarized toric surface $(\overline{Y}, \overline{D}, \overline{L})$, that is, a toric surface $\overline{Y}$ compactifying $(\CC^\star)^2$, with its toric boundary divisor $\overline{D}=\overline{Y} \setminus (\CC^\star)^2$, and an ample line bundle $\overline{L}$. 
Explicitly, the fan of $\overline{Y}$ is the polyhedral decomposition of $\RR^2$ defined by $\overline{W}^{\mathrm{asym}}$. The toric boundary divisor $\overline{D}$ is a cycle of irreducible components $\overline{D}_e \simeq \PP^1$ indexed by the legs $e$ of $\overline{W}^{\mathrm{asym}}$, or equivalently by the sides of $\overline{P}$. Moreover, $\overline{L}$ is the unique ample line bundle on $\overline{Y}$ such that $\mathrm{deg}\, \overline{L}|_{
\overline{D}_e}=w_e$ for all $e$, where $w_e$ is the weight of the leg $e$ corresponding to $\overline{D}_e$, or equivalently, the lattice length of the corresponding side of $\overline{P}$.
Let $\mu : \overline{Y} \rightarrow \RR^2$ be the momentum map of the $T^2 \subset (\CC^\star)^2$-action on 
$\overline{Y}$ with respect to the $T^2$-equivariant symplectic form in the K\"ahler class defined by $\overline{L}$. Then, the image of $\mu$ is the polygon $\overline{P}$ and the restriction of $\mu$ to $(\CC^{*})^2 = \overline{Y}\setminus \overline{D}$ is a trivial $T^2$-fibration over the interior of $\overline{P}$.

The monomials $x^a y^b$ with $(a,b) \in \overline{P}_\ZZ$ naturally define a basis of sections of $\overline{L}$, and so the curve $\overline{C}^\circ \subset (\CC^\star)^2$ admits a natural compactification $\overline{C} \subset \overline{Y}$ in the linear system of $\overline{L}$. If the coefficients $c_{a,b}$ of $f$ as in Equation \eqref{eq_laurent} are all non-zero, 
then the curve $\overline{C}$ does not pass through the $0$-dimensional strata of $\overline{D}$ -- see Figure \ref{Fig21}.

\subsubsection{Webs of 5-branes and degenerations of polarized toric surfaces}
\label{section_degeneration_toric_surfaces}

Let $\overline{\mathcal{P}}$ be a regular polyhedral decomposition of $\overline{P}$ and $\overline{\varphi}$ a convex continuous piecewise-linear function on $\overline{P}$ defining $\overline{\mathcal{P}}$ and having integral values on the integral points of $\overline{P}$. Let $\overline{W}$ be the corresponding dual web of 5-branes. Following \cite{NS}, we associate to this data a \emph{toric degeneration} of the polarized toric surface $(\overline{Y}, \overline{D},\overline{L})$ into unions of polarized toric varieties glued along toric strata. We sketch the construction below.

Consider the polyhedral decomposition $\overline{\mathcal{P}}^\vee$ of $\RR^2$ defined by $\overline{W}$, as in \cite[\S 3]{NS}. Then, the union of the convex polyhedral cones, indexed by the cells $\Xi$ of $\overline{\mathcal{P}}^\vee$ and given by the closures in $\RR^2\times \RR$ of the cones $\RR_{\geq 0} \cdot (\Xi \times \{ 1\})$,
form a $3$-dimensional toric fan $\Sigma$. In addition, the projection map to the last coordinate defines a map of fans $\Sigma \to \Sigma_{\CC} =  \RR_{\geq 0}$ to the toric fan of $\CC$, and so a toric morphism
\[  (\overline{\mathcal{Y}},\overline{\mathcal{D}}) \longrightarrow \CC \,. \] 
The general fiber $ \overline{\mathcal{Y}}_t$ and the central fiber $\overline{\mathcal{Y}}_0$ can be described by examining the preimages of the cones $\{0\}$ and $\RR_{\geq 0}$ respectively, under the map of fans  $\Sigma \to \Sigma_{\CC}$ -- see \cite[Lemma $3.4$]{NS} and \cite[Proposition $3.5$ ]{NS} for details. It follows by construction that the preimage of the cone $\{0\}$ of $\Sigma_{\CC}$ is the \emph{asymptotic web} $\overline{W}^{\mathrm{asym}}$ associated to $\overline{W}$ (see \cite[pg 9]{NS}), and hence the general fiber is the corresponding toric surface  $\overline{Y}$. Moreover, the preimage of any point in the interior of $\RR_{\geq 0} = \Sigma_{\CC}$ is $\RR^2$ with the polyhedral decomposition $\overline{\mathcal{P}}^\vee$, which therefore can be viewed as the dual intersection complex of the central fiber $\overline{\mathcal{Y}}_0$. In particular, corresponding to each vertex $v$ of $\overline{\mathcal{P}}^\vee$, there is a toric irreducible component $(\overline{Y}_0^v, \partial \overline{Y}_0^v)$ of the central fiber $ \overline{\mathcal{Y}}_0$, where we denote by $\partial \overline{Y}_0^v$ the toric boundary divisor.

By standard toric geometry, there exists a unique line bundle $\overline{\mathcal{L}}$ on $\overline{\mathcal{Y}}$ such that, for every toric curve $C_e \simeq \mathbb{P}^1$ in $\overline{\mathcal{Y}}$ corresponding to an edge or leg $e$ of $\overline{W}$, we have 
$\mathrm{deg}\, \overline{\mathcal{L}}|_{C_e}=w_e$,
where $w_e$ is the weight of $e$ in $\overline{W}$.
The restriction of $\overline{\mathcal{L}}$ to a general fiber $(\overline{Y},\overline{D})$ is isomorphic to the ample line bundle $\overline{L}$. Moreover, $(\overline{P}, \overline{\mathcal{P}})$ is the intersection complex of the polarized central fiber, as in \cite{GSreal}. In particular, for every vertex $v$ of $\overline{\mathcal{P}}^\vee$, the momentum polygon of the polarized toric surface $(\overline{Y}_0^v, \partial{\overline{Y}}_0^v, \overline{\mathcal{L}}|_{\overline{Y}_0^v})$ is the face of $\overline{\mathcal{P}}$ corresponding to $v$.

Finally, there is a natural way to construct a family of curves $(\overline{C}_t)_{t\in \CC}$ in the linear systems $|\overline{L}_t|$
obtained by restricting $\overline{\mathcal{L}}$ to the fibers $\overline{\mathcal{Y}}_t$. Explicitly, for $t \neq 0$, $\overline{C}_t$ is the closure of the curve $\overline{C}^\circ_t =\{f_t=0\} \subset (\CC^\star)^2$, where 
\begin{equation*}
 f_t=\sum_{(a,b)\in \overline{P}_\ZZ} t^{\overline{\varphi}((a,b))}x^a y^b \,,
\end{equation*}
For $t=0$, the curve $\overline{C}_0$ is a union of curves in the irreducible components $\overline{Y}_0^v$ of the central fiber $\overline{\mathcal{Y}}_0$ and not passing through the 0-dimensional strata of $\partial \overline{Y}_0^v$. 
Moreover, the web of 5-branes $\overline{W}$ coincides with the tropicalization of the 
curves $\overline{C}^\circ_t$, that is, the limit when $t \rightarrow 0$ of the rescaled amoeba $\mathrm{Log}_t(\overline{C}_t^\circ)$, where \begin{align*}\mathrm{Log}_t: &(\CC^\star)^2 \longrightarrow \RR^2\\
&(x,y) \longmapsto (t \log |x|, t \log |y|)\,.\end{align*} 
Equivalently, the web $\overline{W}$ is the dual intersection complex of the limit curve $\overline{C}_0$, that is, vertices of $\overline{W}$ are in one-to-one correspondence with the irreducible components of $\overline{C}_0$, edges are in correspondence with contact points with the toric divisors of the irreducible components of $Y_0$, and weight of edges are contact orders at these contact points.

The degenerations of curves $(\overline{C}_t^\circ)$ induce degenerations of the corresponding Calabi--Yau 3-folds $\{uv=f_t\}$, that is, paths in the moduli space of complex structures on the Calabi--Yau 3-fold $\overline{Z}$. If the web $\overline{W}$ is generic, that is, if $\overline{\mathcal{P}}$ is a regular triangulation into lattice triangles of size one, then the corresponding degeneration of $\overline{Z}$ is called \emph{maximal} and defines a ``large complex structure" degeneration in the sense of mirror symmetry. We review mirror symmetry for these maximal degenerations in the following section \S\ref{section_toric_dual_CY3}. 
For such a degeneration, the central fiber $\overline{\mathcal{Y}}_0$ of the polarized degeneration 
$(\overline{\mathcal{Y}}, \overline{\mathcal{D}}, \overline{\mathcal{L}}
) \rightarrow \CC$ is a union of polarized projective planes $(\PP^2, \partial \PP^2, \mathcal{O}(1))$, since the triangle of size one is the momentum polytope of $(\PP^2, \partial \PP^2, \mathcal{O}(1))$.

\subsection{Toric mirror symmetry and the M-theory dual toric Calabi--Yau 3-fold}
\label{section_toric_dual_CY3}

\subsubsection{Toric mirror symmetry and canonical toric 3-fold singularities}
Let $\overline{W}^{\mathrm{asym}}$ be an asymptotic web of 5-branes, with dual lattice polygon $\overline{P}$. As reviewed in \S\ref{section_tropical_web}, the web $\overline{W}^{\mathrm{asym}}$
in Type IIB string theory defines a 5d SCFT. Another way to produce a 5d SCFT is to consider M-theory on a canonical 3-fold singularity. Therefore, it is natural to ask if the 5d SCFT defined by $\overline{W}^{\mathrm{asym}}$ admits a dual description as M-theory on a canonical 3-fold singularity $\overline{\mathcal{X}}^{\mathrm{can}}$.

As reviewed in \S \ref{webs_toric_surfaces}, the 4d $\mathcal{N}=2$ theory obtained by compactification on $S^1$ of the 5d SCFT defined by $\overline{P}$ is dual to Type IIB string theory on the Calabi--Yau 3-fold $\overline{Z}$ defined by Equation \eqref{eq_Z}. If the 5d SCFT has a dual description as M-theory on a canonical 3-fold singularity $\overline{\mathcal{X}}^{\mathrm{can}}$, then, by duality of M-theory on $S^1$ with Type IIA string theory, the same 4d $\mathcal{N}=2$ theory should be described as Type IIA string theory on $\overline{\cX}^{\mathrm{can}}$. Hence, $\overline{Z}$ and $\overline{\mathcal{X}}^{\mathrm{can}}$ should be mirror Calabi--Yau 3-folds. Therefore, we will obtain a description of the M-theory dual Calabi--Yau 3-fold $\overline{\mathcal{X}}^{\mathrm{can}}$ using toric mirror symmetry.

For the non-compact Calabi--Yau 3-fold $\overline{Z}$ as in Equation \eqref{eq_Z}, the mirror can be easily described combinatorially: 
$\overline{\mathcal{X}}^{\mathrm{can}}$
is the affine toric Calabi--Yau 3-fold with fan the cone $C(\overline{P})$
over $\overline{P} \times \{1\} \subset \RR^2 \times \RR$.
The projection $C(\overline{P}) \rightarrow \RR_{\geq 0}$ 
on the last $\RR$-factor induces a toric morphism 
\begin{equation} \label{eq_pi} 
\overline{\pi}: \overline{\mathcal{X}}^{\mathrm{can}} \longrightarrow \CC\,.
\end{equation}
The fibers 
$\overline{\mathcal{X}}^{\mathrm{can}}_t :=\overline{\pi}^{-1}(t)$ for $t \neq 0$ are all isomorphic to $(\CC^\star)^2$. On the other hand, the central fiber $\overline{\mathcal{X}}^{\mathrm{can}}_0 :=\overline{\pi}^{-1}(0)$ is an union of affine toric $\overline{X}_0^v$ surfaces, in one-to-one correspondence with the vertices of $\overline{P}$, glued along their toric divisors, in one-to-one correspondence with the sides of $\overline{P}$.  In other words,
the polyhedral decomposition of $\RR^2$ induced by the asymptotic web $\overline{W}^{\mathrm{asym}}$ is the intersection complex of $\overline{\mathcal{X}}^{\mathrm{can}}_0$, that is, the faces of this decomposition are the momentum polytopes of the irreducible components of $\overline{\mathcal{X}}^{\mathrm{can}}_0$.

The affine toric variety $\overline{\mathcal{X}}^{\mathrm{can}}$ is Gorenstein and has canonical singularities that can be described as follows.  Except when $\overline{P}$ is a triangle of size one, in which case $\overline{\mathcal{X}}^{\mathrm{can}}=\CC^3$ is smooth, $\overline{\mathcal{X}}^{\mathrm{can}}$ is singular at the common intersection point $x_0$ of all the irreducible components of $\overline{\cX}^{\mathrm{can}}_0$. Moreover, if $\overline{P}$ has a side of lattice length $m >1$, connecting two vertices $v$ and $v'$, then $\overline{\mathcal{X}}^{\mathrm{can}}$
has also a one-parameter family of $A_{m-1}$-surface singularities along $\overline{X}_0^v \cap \overline{X}_0^{v'}$. Finally, $\overline{\cX}^{\mathrm{can}}$ is smooth everywhere else. 

\subsubsection{Toric mirror symmetry and crepant resolutions}
A regular polyhedral decomposition $\mathcal{\overline{P}}$ of $\overline{P}$ defines a projective crepant partial resolution
\[ f_{\overline{\cP}}: \overline{\mathcal{X}}_{\overline{\cP}} \longrightarrow \overline{\mathcal{X}}^{\mathrm{can}}\,,\]
where $\overline{\cX}_{\overline{\cP}}$ is the toric 3-fold with fan the cone over $\mathcal{\overline{P}} \times \{1\}$ in $\RR^2 \times \RR$. 
Composing with $\overline{\pi}$ given by Equation \eqref{eq_pi}, we obtain a 
map
\[ \overline{\pi}_{\overline{\cP}}: \overline{\cX}_{\overline{\cP}} \longrightarrow \CC \,,\]
whose fibers $\overline{\cX}_{\overline{\cP},t}:= \overline{\pi}_{\overline{\cP}}^{-1}(t)$ over $t \neq 0$
are isomorphic to $(\CC^\star)^2$. 
The central fiber $\overline{\cX}_{\overline{\cP},0} := \overline{\pi}_{\overline{\cP}}^{-1}(0)$ has dual intersection complex intersection given by $\mathcal{\overline{P}}$, or equivalently intersection complex given by the polyhedral decomposition of $\RR^2$ defined by the web $\overline{W}$ dual to $\overline{\cP}$.

The partial resolutions $f_{\overline{\cP}}: \overline{\cX}_{\overline{\cP}} \rightarrow \overline{\cX}^{\mathrm{can}}$ are mirror to the one-parameter degenerations of $\overline{Z}$ induced by the one-parameter polarized toric degenerations of the polarized toric surface $(\overline{Y},\overline{D},\overline{L})$ associated to $\overline{\cP}$ as in \S\ref{section_degeneration_toric_surfaces}.
In particular, when $\overline{\cP}$ is a triangulation into triangles of size one, then $\overline{\cX}_{\overline{\cP}}$ is a smooth Calabi--Yau 3-fold, and $f_{\overline{\cP}}: \overline{\cX}_{\overline{\cP}} \rightarrow \overline{\cX}^{\mathrm{can}}$ is a projective crepant resolution that is mirror to the large complex structure degenerations of 
$\overline{Z}$ associated to $\overline{\mathcal{P}}$. 
The union of the Coulomb branches of all the massive deformations of the 5d SCFT can be described as the union of the K\"ahler cones of these projective crepant resolutions $\overline{\cX}_{\overline{\cP}} \rightarrow \overline{\cX}^{\mathrm{can}}$.

\section{Webs of 5-branes with 7-branes and log Calabi--Yau surfaces}
\label{section_webs_log_cy}

In \S\ref{sec:seven_branes}, we provide a geometric interpretation of Hanany--Witten moves.
After giving a review of webs of 5-branes with 7-branes in \S\ref{subsec:Webs of 5-branes with 7-branes}, we explain in \S\ref{section_webs_log_CY} how to attach a log Calabi--Yau surface with line bundle to an asymptotic web of 5-branes with 7-branes, and in \S\ref{sec: degeneration} how to more generally attach a degeneration of log Calabi--Yau surfaces with line bundle to webs of 5-branes with 7-branes.

\subsection{Asymptotic configurations of $7$-branes and open Calabi--Yau surfaces} \label{sec:seven_branes}
In this section, after a review of the notion of a Hanany--Witten move between configurations of 7-branes in \S\ref{section_7_branes}, we introduce log Calabi--Yau and open Calabi--Yau surfaces in 
\S\ref{section_log_open}. Then, we attach an open Calabi--Yau surface to every configuration of 7-branes and we give a geometric interpretation of Hanany--Witten moves in terms of elementary cluster birational transformations of the corresponding open Calabi--Yau surfaces in \S\ref{section_geometric_HW}. As a result, we establish in Theorem \ref{thm_7_branes} a one-to-one correspondence between configurations of $7$-branes and open Calabi--Yau surfaces. Finally, we present a string-theoretic interpretation of the open Calabi--Yau surface $U$ associated to a configuration of 7-branes in \S\ref{section_string_open_CY}: the 7-dimensional theory obtained by compactifying on $S^1$ the 8-dimensional worldvolume theory of the 7-branes admits a dual description as M-theory on the complex surface $U$.

\subsubsection{Asymptotic configurations of $7$-branes}
\label{section_7_branes}
Type IIB string theory contains 7-branes 
of type $(p,q)$ for every coprime $(p,q)\in \ZZ^2 \setminus \{0\}$ \cite{Polchinski}. Starting from Type IIB string theory on 
\[ \RR^{1,4} \times \RR^2 \times \RR^3\,,\]
placing a $(p,q)$ 7-brane on 
\[ \RR^{1,4} \times \{x\} \times \RR^3\]
introduces a focus-focus singularity in the natural integral affine structure on $\RR^2$ at the point $x \in \RR^2$. The monodromy matrix, defined by tracing a loop in the clockwise direction around $x$, is given by
\[ M_{(p,q)}= \begin{pmatrix}
1-pq & p^2 \\
-q^2 & 1+pq 
\end{pmatrix} \,.\]
In particular, the line $x+ \RR (p,q)$ is a  monodromy-invariant direction around $x$. 

We will consider particular configurations of 7-branes, defined as follows.

\begin{definition} \label{def_7_branes}
For $Q\in \ZZ_{\geq 0}$, an \emph{asymptotic configuration of $Q$-many 7-branes} is the data of a coprime vector $(p_i,q_i) \in \ZZ^2 \setminus \{0\}$ for every $1 \leq i \leq Q$, together with a point $x_i$ on the half-line $\RR_{>0}(p_i,q_i)$.
\end{definition}

The point $x_i$ in Definition \ref{def_7_branes} is viewed as the position in $\RR^2$ of a 
$(p_i,q_i)$ 7-brane, see Figure \ref{Fig34}.
An asymptotic configuration of 7-branes defines an integral affine manifold with singularities $B$ obtained from $\RR^2$ by cutting along the half-lines $x_i + \RR_{\geq 0}(p_i,q_i)$ and imposing that the clockwise monodromy around $x_i$ through the cut is given by the matrix $M_{(p_i,q_i)}$.  
The integral affine manifold with singularities $B$ is well-defined even if $(p_i,q_i)=(p_j,q_j)$ for $i \neq j$, since the matrices $M_{(p_i,q_i)}$ and $M_{(p_j,q_j)}$ are equal, and so in particular commute.

\begin{figure}[hbt!]
\center{\scalebox{0.9}{\includegraphics{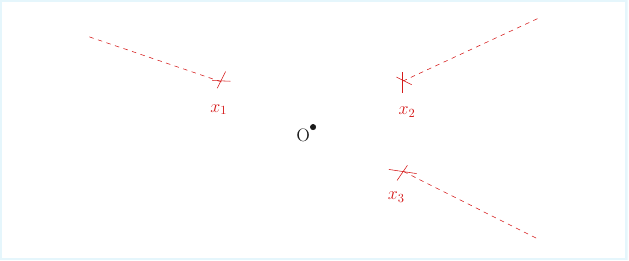}}}
\caption{An asymptotic configuration of 7-branes.}
\label{Fig34}
\end{figure}

\begin{definition}
\label{Def: equivalence 7 branes}
Two asymptotic configurations of 7-branes are called \emph{equivalent} if they are related to each other by either of the following: 
\begin{itemize}
    \item[i)] The natural common action of $SL(2,\ZZ)$ on the vectors $(p_i,q_i)$ and on $\RR^2$, or
\item[ii)] A change of the position $x_i$ of the $(p_i,q_i)$ 7-brane along its monodromy-invariant direction $\RR_{>0}(p_i,q_i)$, or
\item[iii)] Hanany--Witten moves \cite{hanany1997type}, as discussed below.
\end{itemize}
\end{definition}

A \emph{Hanany--Witten move along a $(p_i,q_i)$ $7$-brane $x_i$} is the operation of bringing the position of $x_i$ across the origin along the monodromy-invariant direction $\RR (p_i,q_i)$, so that $x_i$ is now contained in the ray $\RR_{<0} (p_i,q_i)$. In particular, this new configuration of 7-branes is no longer an asymptotic configuration of 7-branes as in Definition \ref{def_7_branes}. 
However, it can be transformed into one by viewing $x_i$ as a $(-p_i,-q_i)$ 7-brane, and rotating the cut in the clockwise direction from $x_i +\RR_{\geq 0}(p_i,q_i)$ to $x_i +\RR_{< 0}(p_i,q_i)$. 
During such a rotation, the cut of the $(p_i,q_i)$ 7-brane intersects the $(p_j,q_j)$ 7-branes with $\det((p_i,q_i),(p_j,q_j))<0$.
Such a $(p_j,q_j)$ 7-brane crosses the cut of the $(p_i,q_i)$ 7-brane in the anti-clockwise direction around $x_i$, and so becomes a $(p_j',q_j')$ 7-brane, 
where
\[ \begin{pmatrix}
 p_j' \\
 q_j' 
\end{pmatrix}= M_{(p_i,q_i)}^{-1} \begin{pmatrix}
 p_j \\
q_j 
\end{pmatrix}
=\begin{pmatrix}
1+p_i q_i & -p_i^2 \\
q_i^2 & 1-p_i q_i 
\end{pmatrix}\begin{pmatrix}
 p_j \\
q_j
\end{pmatrix} = \begin{pmatrix}
 p_j \\
q_j 
\end{pmatrix}+\mathrm{det}((p_j,q_j),(p_i,q_i)) \begin{pmatrix}
p_i \\
q_i  
\end{pmatrix}
\,.\]
The resulting asymptotic configuration of 7-branes is said to be related to the initial one by Hanany--Witten move of the $(p_i,q_i)$ 7-brane.
The new types $(p_j',q_j')$ of the 7-branes are given by 
\begin{equation} \label{eq_seed_mutation}   
(p_j',q_j') = 
     \begin{cases}
       (-p_i,-q_i) &\quad\text{if}\,\,\, j=i\\
       (p_j,q_j)+\max(\mathrm{det}((p_j,q_j),(p_i,q_i)),0)\, (p_i,q_i) &\quad\text{if}\,\,\, j \neq i \,. 
     \end{cases}
\end{equation}
In the situation $j \neq i $ in \eqref{eq_seed_mutation}, we refer to the operation of changing the direction of a $7$-brane from $(p_j,q_j)$ to $(p_j',q_j')$ as ``applying a shear'' to it.

\begin{remark}
    In the mirror symmetry literature, moving a focus-focus singularity of an integral affine manifold along its monodromy-invariant direction is sometimes referred to as ``moving worms" \cite{GHK1, KSaffine}.
\end{remark}

Our goal is to give a geometric description of the set of asymptotic configurations of 7-branes, up to the above equivalences. For this, we first describe log and open Calabi--Yau surfaces.

\subsubsection{Log and open Calabi--Yau surfaces}
\label{section_log_open}
A \emph{log Calabi--Yau surface} $(Y,D)$ consists of a normal projective surface $Y$ over $\CC$ together with a reduced, anti-canonical divisor $D$ in $Y$, such that the pair $(Y,D)$ is locally toric, that is, for every point $y \in Y$, there exists a complex analytic open neighborhood $V$ of $y$ in $Y$, an affine toric surface $(\overline{Y}, \overline{D})$, a point $\overline{y} \in \overline{Y}$, and a complex analytic neighborhood $\overline{V}$ of $\overline{y}$ in $\overline{Y}$
such that 
\[ (Y \cap V, D \cap V) \simeq (\overline{Y} \cap \overline{V}, \overline{D} \cap \overline{V})\,.\]  
We also assume that $D$ is non-empty and singular, so that $D$ is either an irreducible rational curve with a single node or a cycle of smooth rational curves,
as in \cite{GHK_moduli}. In particular, a singular point of $Y$
is also necessarily a singular point of $D$, and is given locally by an isolated toric surface singularity, that is, by a quotient of $\CC^2$ by a finite cyclic group. 

\begin{remark}
    Log Calabi--Yau surfaces $(Y,D)$ with $Y$ smooth are called \emph{Looijenga pairs} in \cite{GHK1, GHK_moduli} and \emph{anticanonical pairs} in \cite{friedman2015geometry}.
\end{remark}

Given a log Calabi--Yau surface $(Y,D)$, the complement $U=Y\setminus D$ is a smooth non-compact Calabi--Yau surface, in the sense that it admits a holomorphic volume form. We refer to such a surface as open Calabi--Yau surfaces:

\begin{definition} \label{def_open_CY}
An \emph{open Calabi--Yau surface} is a smooth algebraic surface $U$ over $\CC$, such that there exists a log Calabi--Yau surface $(Y,D)$ with $U=Y\setminus D$. We refer to $(Y,D)$ as a \emph{log Calabi--Yau compactification} of $U$.
\end{definition}

In the following sections, we will consider open Calabi--Yau surfaces up to deformation equivalence, defined as follows.

\begin{definition}
\label{Def: def eqv open CY}
Two open Calabi--Yau surfaces $U_1$ and $U_2$ are called \emph{deformation equivalent} if there exists a smooth projective family of surfaces $\pi: \mathcal{Y} \to T$ over a scheme $T$, together a divisor $\mathcal{D}$ in $\mathcal{Y}$, and points $t_1,t_2 \in T$ satisfying the following conditions:
\begin{itemize}
    \item[i)] The restriction $\pi|_{\mathcal{D}} : \mathcal{D} \to T$ is a locally trivial fibration, 
    \item[ii)] $\forall t \in T$, the fiber $(Y_t, D_t):= (\pi^{-1}(t), {\pi|_{\mathcal{D}}}^{-1}(t))$ is a log Calabi--Yau surface, and
    \item[iii)] $U_1 = Y_{t_1} \setminus D_{t_1}$ and $U_2 = Y_{t_2} \setminus D_{t_2}$.  
\end{itemize}
\end{definition}

\begin{example} 
Let $Y$ be a projective toric surface and $D$ its toric boundary divisor. Then, $(Y,D)$ is a log Calabi--Yau surface, and the corresponding open Calabi--Yau surface is $U=Y \setminus D =(\CC^\star)^2$.
\end{example}

\begin{example}
If $(Y,D)$ is a log Calabi-Yau surface and $p$ is a singular point of $D$, then, denoting $\widetilde{Y}$ the blow-up of $Y$ at $p$ and $\widetilde{D}$ the strict transform of $D$ in $\widetilde{Y}$, the pair $(\widetilde{Y}, \widetilde{D})$ is a log Calabi--Yau surface. We say that $(\widetilde{Y}, \widetilde{D})$  is obtained from $(Y,D)$ by \emph{corner blow-up}. The open Calabi--Yau surface $U=Y \setminus D$ does not change under a corner blow-up.
\end{example}

\begin{example} \label{example_toric_model}
If $(Y,D)$ is a log Calabi-Yau surface and $p$ is a smooth point of $D$, then, denoting $\widetilde{Y}$ the blow-up of $Y$ at $p$ and $\widetilde{D}$ the strict transform of $D$ in $\widetilde{Y}$, the pair $(\widetilde{Y}, \widetilde{D})$ is a log Calabi--Yau surface. We say that $(\widetilde{Y}, \widetilde{D})$  is obtained from $(Y,D)$ by \emph{interior blow-up}. The open Calabi--Yau surface $U=Y \setminus D$ changes under an interior blow-up. In particular, interior blow-ups of a toric log Calabi--Yau surface are examples of non-toric log Calabi--Yau surfaces. According to \cite[Proposition 1.3]{GHK1}, for every log Calabi--Yau surface $(Y,D)$, there exists a log Calabi--Yau surface $(\widetilde{Y},\widetilde{D})$ obtained from $(Y,D)$ by corner blow-ups, and such that $(\widetilde{Y},\widetilde{D})$ admits a \emph{toric model}, that is, a morphism $(\widetilde{Y},\widetilde{D}) \rightarrow (\overline{Y},\overline{D})$ that is a composition of interior blow-ups to a toric surface $(\overline{Y},\overline{D})$.
\end{example}

\subsubsection{Asymptotic configurations of $7$ branes and open Calabi--Yau surfaces}
\label{section_geometric_HW}
The main result of this section is Theorem 
\ref{thm_7_branes} below, that establishes  a one-to-one correspondence between asymptotic configurations of $7$-branes up to \emph{equivalence} as in Definition \ref{Def: def eqv open CY}, and open Calabi--Yau surfaces up to \emph{deformation equivalence} as in Definition \ref{Def: equivalence 7 branes}. 

We first explain how to associate a deformation class of open Calabi--Yau surfaces to every asymptotic configuration of 7-branes. 

\begin{construction}
\label{constr: open cy surface}
Consider an asymptotic configurations of $Q$-many 7-branes of type $(p_i,q_i)$, $1 \leq i \leq Q$. Let $(\overline{Y}, \overline{D})$ be a smooth projective toric surface whose fan contains all the rays $\RR_{\geq 0}(p_i,q_i)$ for $1\leq i \leq Q$. 
For every $1\leq i \leq Q$, choose a smooth point $y_i$ of $\overline{D}$ contained in the irreducible component of $\overline{D}$ corresponding to the ray $\RR_{\geq 0}(p_i,q_i)$, and such that $y_i \neq y_j$ if $i \neq j$. Let $Y$ be the blow-up of $\overline{Y}$ at the $Q$ many distinct points $y_i$, and let $D$ be the strict transform of $\overline{D}$ in $Y$. Then, $(Y,D)$ is a log Calabi--Yau surface, and we define the \emph{open Calabi--Yau surface of the asymptotic configuration of 7-branes} by $U:= Y \setminus D$. The preimage in $U$ of $(\CC^\star)^2=\overline{Y}\setminus \overline{D}$ is a dense copy of $(\CC^\star)^2$ in $U$.

In this construction, the toric surface $\overline{Y}$ is not unique. However, different choices of toric surfaces $\overline{Y}$ are related by toric blow-ups and blow-downs,
which do not affect the open Calabi--Yau surfaces $U$ or the torus $(\CC^\star)^2 \subset U$. While this construction also depends a priori on the choice of the points $y_i$ on $\overline{D}$, different choices lead to deformation-equivalent open Calabi--Yau surfaces. We conclude that the deformation class of the open Calabi--Yau surface $U$ only depends on the initial asymptotic configuration of 7-branes, and not on any additional choice.    
\end{construction}

We now present a geometric interpretation of Hanany--Witten moves in terms of the birational geometry of the corresponding open Calabi--Yau surfaces. Let $A$ be an asymptotic configuration of $Q$-many 7-branes, and let $A_i'$ be the asymptotic configuration of $Q$-many 7-branes obtained from $A$ by Hanany--Witten move along the $i$-th 7-brane for some $1\leq i \leq Q$. Construction \ref{constr: open cy surface} applied to $A$ produces an open Calabi--Yau surface $U_A = Y_A \setminus D_A$, where $(Y_A,D_A)$ is a log Calabi--Yau surface obtained by blowing up a toric surface $(\overline{Y}_A, \overline{D}_A)$, and a dense torus $(\CC^\star)^2_A$. Similarly, there is an open Calabi--Yau surface $U_{A_i'}$ and a dense torus $(\CC^\star)^2_{A_i'}$ associated to $A$. We explain below that the open Calabi--Yau surfaces $U_A$ and $U_{A_i'}$ are isomorphic, up to deformations. However, the tori $(\CC^\star)^2_A$ and $(\CC^\star)^2_{A_i'}$ in $U_A \simeq U_{A_i'}$ do not coincide but are related by a birational transformation known as en elementary cluster transformation.

\begin{figure}[hbt!]
\center{\scalebox{0.9}{\includegraphics{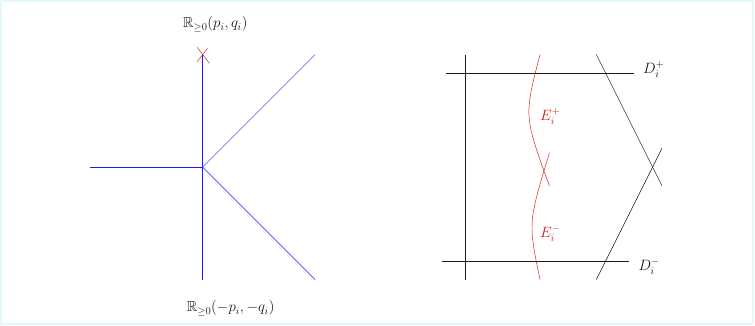}}}
\caption{Elementary cluster transformation: after a blow-up producing the exceptional curve $E_i^+$, contract the $(-1)$-curve $E_i^-$.}
\label{Fig35}
\end{figure}

Without loss of generality, we can assume that the fan of $\overline{Y}_A$ contains the line $\RR (p_i,q_i)$. Then, the projection $\RR^2 \rightarrow \RR^2/\RR(p_i,q_i) =\RR$ induces a toric morphism $\nu: \overline{Y}_A \rightarrow \PP^1$, which is generically a $\PP^1$-fibration. Moreover, the irreducible components $D_i^+$
and $D_i^-$ of $\overline{D}$ corresponding respectively to the rays $\RR_{\geq 0}(p_i,q_i)$ and $\RR_{\geq 0}(-p_i,-q_i)$ of the fan, are sections of $\nu$, see Figure \ref{Fig35}. 
By Construction \ref{constr: open cy surface}, $Y_A$ is obtained from the toric surface $\overline{Y}_A$ by blowing up points $y_j$ on the toric boundary for each $1\leq j \leq Q$. In particular, blowing up the point $y_i \in D_i^+$
creates an exceptional $(-1)$-curve $E_i^+$ in $Y_A$. The strict transform of the $\PP^1$-fiber of $\nu$ passing through $x_i$ is also a $(-1)$-curve $E_i^-$. Hence, one can construct a new toric surface $Y_{A_i'}$ from $Y_A$ by contracting all the exceptional curves coming from blowing up the points $y_j$ with $j \neq i$, and the $(-1)$-curve $E_i^-$. The induced birational map
$\overline{Y}_A \dashrightarrow \overline{Y}_{A_i}$  is an example of 
``elementary cluster transformation" as in \cite[\S 3.1]{GHK_birational}, and it follows from \cite{GHK_birational}
that the fan of the toric surface $\overline{Y}_{A_i}$ is obtained from the fan of $\overline{Y}_A$ by the piecewise-linear transformation given by \eqref{eq_seed_mutation}. In particular, the toric surface $Y_{A_i'}$ can be taken as the starting point of the Construction \ref{constr: open cy surface} for the asymptotic configurations $A_i'$ obtained from $A$ by Hanany--Witten move of the $i$-th 7-brane.
The corresponding log Calabi--Yau surface $(Y_{A_i'}, D_{A_i'})$ coincides with $(Y_A,D_A)$, and so the corresponding open Calabi--Yau surface $U_{A_i'}$ coincides with $U_A$. Moreover, the elementary cluster transformation restricted to the tori $(\CC^\star)^2_A$ and $(\CC^\star)^2_{A_i'}$
defines a birational transformation
\[ (\CC^\star)^2_A \dashrightarrow (\CC^\star)^2_{A_i'}\,,\]
given by 
\begin{align*}
    x &\longmapsto x \\
    y &\longmapsto y(1+x)^k
\end{align*}
for an appropriate choice of coordinates $x,y$ on $(\CC^\star)^2$ and for some $k \in \ZZ$. 

The above geometric interpretation of Hanany--Witten moves has a natural reformulation in terms of cluster varieties \cite{GHK_birational}.
An asymptotic configuration of $Q$ many 7-branes of type $(p_i,q_i)$ naturally defines a \emph{seed} $\mathbf{s}$, consisting of the $Q$-dimensional lattice $\ZZ^Q =\bigoplus_{i=1}^Q \ZZ e_i$, with the standard basis $(e_i)_{1\leq i \leq Q}$, and the skew-symmetric from $\omega$ defined by 
\[ \omega(e_i,e_j):=\det((p_i,q_i),(p_j,q_j)) =p_i q_j - p_j q_i \,,\]
for all $1\leq i,j \leq Q$.
It follows from \eqref{eq_seed_mutation} that Hanany--Witten moves for asymptotic configurations of 7-branes are exactly \emph{mutations} of seeds in the theory of cluster algebras \cite{FZ}. The seed $\mathbf{s}$ defines a $Q$-dimensional $\mathcal{X}$ cluster variety with a natural Poisson structure induced by $\omega$ \cite{FG} and a cluster chart $(\CC^\star)^Q \subset \mathcal{X}$.
A seed mutation does not change the $\mathcal{X}$ cluster variety but changes the cluster chart by an elementary cluster modification.
The skew-symmetric form $\omega$ has rank $2$, and so the symplectic leaves of the Poisson structure on the $\mathcal{X}$ cluster variety are 2-dimensional. By \cite{GHK_birational}, the symplectic leaves of the $\mathcal{X}$ cluster variety coincide up to codimension $2$ with the open Calabi--Yau surfaces deformation-equivalent to the open Calabi--Yau surface $U$ associated to the asymptotic configuration of 7-branes by Construction \ref{constr: open cy surface}.

\begin{remark}
    The relation between Hanany--Witten moves and quiver and polygon mutations, which are both related to cluster mutations, was previously observed in the physics literature in \cite{arias2024geometry, Franco_twin}. For the purposes of the present paper, the geometric reformulation of this relation in terms of cluster varieties and open Calabi--Yau surfaces plays an essential role.
\end{remark}

It follows from the above geometric interpretation  of Hanany--Witten moves that open Calabi--Yau surfaces obtained by applying Construction \ref{constr: open cy surface} to equivalent asymptotic configurations of 7-branes are deformation-equivalent. In particular, Construction \ref{constr: open cy surface} induces a map from the set of asymptotic configurations of 7-branes, up to equivalence, to the set of open Calabi--Yau surfaces, up to deformation equivalence.

\begin{theorem} \label{thm_7_branes}
Construction \ref{constr: open cy surface} defines a one-to-one correspondence between asymptotic configurations of 7-branes, up to equivalence, and open Calabi--Yau surfaces, up to deformation equivalence. 
\end{theorem}

\begin{proof}
Denote by $F$ the map induced by Construction \ref{constr: open cy surface} from the set of asymptotic configurations of 7-branes, up to equivalence, to the set of open Calabi--Yau surfaces, up to deformation equivalence. 
We first show that the map $F$ is surjective. Let $U$ be an open Calabi--Yau surface. By Definition \ref{def_open_CY}, there exists a log Calabi--Yau surface $(Y,D)$ such that $U=Y \setminus D$. 
As reviewed in Example \ref{example_toric_model}, up to replacing $(Y,D)$ by a corner blow-up, we can assume that $(Y,D)$ admits a toric model, that is, a morphism $(Y,D) \rightarrow (\overline{Y}, \overline{D})$ to a toric surface $(\overline{Y}, \overline{D})$, obtained by successively blowing up smooth points of $\overline{D}$. Moreover, up to replacing $U$ by a deformation-equivalent open Calabi--Yau surface, one can assume that $(Y,D) \rightarrow (\overline{Y}, \overline{D})$ is the blow-up of distinct smooth points $y_1, \dots, y_Q$ of $\overline{D}$.  
For every $1 \leq i \leq Q$, let $(p_i,q_i)$ be the primitive direction of the ray in the fan of $\overline{Y}$ corresponding to the irreducible component of $\overline{D}$ containing the point $y_i$. Then, $U$ is the image by the map $F$ of the asymptotic configuration of $Q$ many 7-branes of type $(p_i,q_i)$, $1 \leq i \leq Q$. This proves the surjectivity of $F$.

To prove that the map $F$ is injective, let $A_1$ and $A_2$ be two asymptotic configurations of 7-branes such that the corresponding open Calabi--Yau surfaces $U_{A_1}$ and $U_{A_2}$ are isomorphic. For every $k\in \{1,2\}$, we have $U_{A_k}=Y_{A_k} \setminus D_{A_k}$, where  $(Y_{A_k}, D_{A_k})$ is a log Calabi--Yau surface defined via a toric model $(Y_{A_k}, D_{A_k}) \rightarrow (\overline{Y}_{A_k}, \overline{D}_{A_k})$.
By \cite[Proposition 3.27]{HK_HMS} (see also \cite[Theorem 1]{Blanc}), any two toric models of a log Calabi--Yau surface are connected by a sequence of corner blow-ups/blow-downs and elementary cluster transformations. Since elementary cluster transformations correspond to Hanany--Witten moves at the level of the asymptotic configurations of 7-branes,  
it follows that $A_1$ and $A_2$ are equivalent in the sense of Definition \ref{Def: equivalence 7 branes}. This concludes the proof that the map $F$ is injective.
\end{proof}

\begin{remark}
The problem of classifying asymptotic configurations of 7-branes up to equivalence has been extensively studied in the physics literature - see for example \cite{brain_webs} and references there. In particular, the article \cite{brain_webs}
contains explicit counter-examples to a previous conjecture of \cite{uncovering} on classifying the asymptotic configurations of 7-branes up to equivalence by the conjugacy class of their global monodromy and the index in $\ZZ^2$ of the lattice spanned by the $(p_i,q_i)$'s. Theorem \ref{thm_7_branes} provides an algebro-geometric reformulation of this subtle classification question. For so-called ``positive" open Calabi--Yau surfaces, that is, which are deformation-equivalent to affine surfaces, it is proved in \cite{mandel_rank_two} that the conjugacy class of the global monodromy uniquely determines the open Calabi--Yau surface  up to deformation, and so, by Theorem \ref{thm_7_branes}, uniquely determines the asymptotic configuration of 7-branes up to equivalence.
\end{remark}

\subsubsection{String-theoretic interpretation}
\label{section_string_open_CY}
Let $A$ be an asymptotic configuration of 7-branes, corresponding to an integral affine manifold with singularity $B_A$ and, via Construction \ref{constr: open cy surface}, to an open Calabi--Yau surface $U_A$. It follows from \cite{symington} that there exists a direct topological relationship between $U_A$ and $B_A$. Indeed, $U_A$ admits a topological $T^2$-fibration over $B_A$, which is smooth away from the singularities $\{x_i\}_{1\leq i \leq Q}$, and with nodal fibers over the singularities $\{x_i\}_{1\leq i \leq Q}$. More precisely, fixing an identification of the first homology group of the $T^2$-fiber over the origin with $\ZZ^2$, the fiber over the $(p_i,q_i)$ 7-brane $x_i$ is obtained by contracting to a point a 1-cycle of class $(p_i, q_i)$ in $T^2$.
Using the duality between Type IIB string theory on $S^1$ and M-theory on $T^2$, 
we obtain the following physics interpretation of the open log Calabi--Yau surface $U_A$: 
The compactification over $S^1$ of the 8-dimensional theory on $\RR^{1,7}$ defined by Type IIB string theory on $\RR^{1,7} \times B_A$, with $(p_i,q_i)$
7-branes on $\RR^{1,7} \times \{x_i\}$, admits a dual description as the 7-dimensional theory on $\RR^{1,6}$ defined by M-theory on $\RR^{1,6} \times U_A$. 

While the connection between 7-branes and singular $T^2$-fibrations is well-known, for example in the context of F-theory, the point of Construction \ref{constr: open cy surface} is to give an algebro-geometric description of the total space of this $T^2$-fibration for a particular choice of complex structure for which the $T^2$-fibers are Lagrangian submanifolds. In the context of F-theory, it is more usual to work with a complex structure for which the $T^2$-fibers are holomorphic subvarieties. When $U_A$ admits a complete hyperk\"ahler metric, such a complex structure can be obtained by hyperk\"ahler rotation. While $U_A$ sometimes admits a complete hyperk\"ahler metric, it is not always the case. It is in general expected that when the 7-branes are close enough of the origin, there exists a possibly incomplete hyperk\"ahler metric on an open set of $U_A$ containing the singular fibers. For instance, for a single 7-brane, the relevant hyperk\"ahler metric is the Ooguri--Vafa metric, which is incomplete \cite{GMN1, OV}. One might be concerned about the lack of existence of  a complete hyperk\"ahler metric for the UV consistency of an M-theory background. 
However, we will only use such M-theory background as a geometric description of the low-energy of some field theories, and we do not expect the incompleteness of the hyperk\"ahler metrics to be relevant for such purposes. The situation is analog to the case of 4d $\mathcal{N}=2$ theories of Class S \cite{gaiotto_class_S}, which are described at low-energy by a M5-brane wrapping an algebraic curve in the cotangent bundle $T^\star C$ of a Riemann surface $C$, whereas the non-compact Calabi--Yau surface $T^\star C$ does not admit a complete hyperk\"ahler metric in general, but only a possibly incomplete hyperk\"ahler metric on a neighborhood of the zero-section $C \subset T^\star C$ \cite{Feix, kaledin1997hyperkaehler}.

\subsection{Webs of 5-branes with 7-branes}
\label{subsec:Webs of 5-branes with 7-branes}
In this section, we explain how to modify webs of 5-branes by allowing some 5-branes to end on 7-branes \cite{BBT,DeWolfe}.
We first define a 7-brane data, and then describe how to construct a web of 5-branes with 7-branes from such data.

\begin{definition} \label{def_web_57_at_infinity}
Let $\overline{W} \subset \RR^2$ be a web of 5-branes.
A \emph{7-brane data} for $\overline{W}$ 
is an assignment of the following data, for each leg $e \in L(\overline{W})$:
\begin{itemize}
\item[i)] A nonnegative integer $n_e$, and
\item[ii)] $n_e$ positive integers $a_{e,i}$, indexed by $1\leq i \leq n_e$, such that $\sum_{i=1}^{n_e} a_{e,i} \leq w_e$, where $w_e$ is the weight of $e$ in $\overline{W}$. 
\end{itemize}
We will denote by $\mathbf{a} := ((a_{e,i})_{1\leq i \leq n_e})_{e \in L(\overline{W})}$ a 7-brane data for $\overline{W}$.
\end{definition}

\begin{figure}[hbt!]
\center{\scalebox{0.9}{\includegraphics{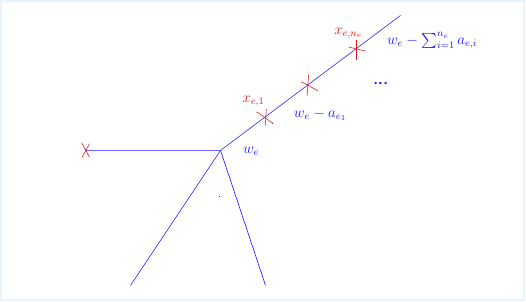}}}
\caption{A web of 5-branes with 7-branes}
\label{Fig30}
\end{figure}

\begin{construction}
Let $\overline{W} \subset \RR^2$ be a web of 5-branes, and $\mathbf{a} := ((a_{e,i})_{1\leq i \leq n_e})_{E \in L(\overline{W})}$ a 7-brane data for $\overline{W}$.
We define an integral affine manifold with singularities $B$, by introducing for every leg $e$ of $\overline{W}$, $n_e$ focus-focus singularities $x_{e,1}, \dots, x_{e,n_e}$ on $e$ with monodromy-invariant direction the direction $(p_e,q_e)$ of $e$, ordered as in Figure \ref{Fig30}. Let $W \subset B$ be the embedded graph in $B$ obtained from $\overline{W}$ as follows:
\begin{itemize}
\item[i)] First, for every leg $e=v +\RR_{\geq 0}(p_e,q_e)$ of $\overline{W}$, with incident vertex $v$ and primitive integral direction $(p_e,q_e)$, we add 2-valent vertices on $e$ at the positions $x_{e,1}, \dots, x_{e,n_e}$ of the focus-focus singularities, so that the leg $e$ is replaced by edges 
\[[v, x_{e,1}], \,[x_{e,1}, x_{e,2}], \dots, [x_{e, n_e-1}, x_{e,n_e}]\] 
and the leg $x_{e,n_e}+\RR_{\geq 0}(p_e,q_e)$. 
\item[ii) ]Moreover, if $\sum_{i=1}^{n_e} a_{e,i}=w_e$, we also remove the leg $x_{e,n_e}+\RR_{\geq 0}(p_e,q_e)$ from the graph.
\item[iii)] Finally, we assign weight $w_e$ to the edge
$[v, x_{e,1}]$, weights $w_e- \sum_{j=1}^i a_{e,j}$ to the edge $[x_{e,i}, x_{e,i+1}]$ for all $1 \leq i \leq n_e-1$, and weight $w_e- \sum_{j=1}^{n_e} a_{e,j}$ to the leg $x_{e,n_e}+\RR_{\geq 0}(p_e,q_e)$ if $\sum_{i=1}^{n_e} a_{e,i}<w_e$. 
\end{itemize}
We denote by $W \subset B$ the resulting embedded weighted graph, and we refer to $W \subset B$ as the \emph{web of 5-branes with 7-branes} defined by the web of 5-branes $\overline{W}$ and the 7-brane data $\mathbf{a} := ((a_{e,i})_{1\leq i \leq n_e})_{E \in L(\overline{W})}$.
Indeed, the integral affine manifold with singularities $B$ defines a Type IIB string theory background 
\[ \RR^{1,4} \times B \times \RR^3\,,\]
with a $(p_e,q_e)$ 7-brane on 
\[ \RR^{1,4} \times \{x_{e,i}\} \times \RR^2\]
for each of the focus-focus singularities $x_{e,i}$, $1 \leq i \leq n_e$, and the embedded graph 
$W \subset B$ determines a configuration of 5-branes
\begin{equation}
\label{Eq: R14}
     \RR^{1,4} \times W \times \{0\}
\end{equation}
in this background. 
The weights on the edges and legs of $W$ encode the fact that, for 
every $E \in L(\overline{W})$ and $1\leq i\leq n_e$, $a_{e,i}$ many 5-branes of type $(p_e,q_e)$ are ending on the 7-brane of type $(p_e,q_e)$ with position $x_{e,i}$.
\end{construction}

We also define asymptotic webs of 5-branes with 7-branes, in a way parallel to Definition \ref{def_asymptotic_5} for asymptotic webs of 5-branes:

\begin{definition} \label{def_asymptotic_57}
An \emph{asymptotic web of 5-branes with 7-branes} $W^{\mathrm{asym}}$ is a web of 5-branes with 7-branes defined by a 7-brane data for an asymptotic web of 5-branes as in Definition \ref{def_asymptotic_5}.
\end{definition}

As in \S \ref{section_tropical_web} for webs of 5-branes, every web of 5-branes with 7-branes $W$ is a deformation
of an asymptotic web of 5-branes with 7-branes $W^{\mathrm{asym}}$.

Unlike the case of webs of 5-branes, a web of 5-brane with 7-branes is not necessarily supersymmetric, and so might not define a 5d $\mathcal{N}=1$ SCFT \cite{BBT}. We will study the question of determining when a web of 5-branes with 7-branes is supersymmetric 
and defines a 5d SCFT in \S \ref{section_consistent}. 
To do this, we will first describe in \S \ref{section_webs_log_CY} how to attach a log Calabi--Yau surface and a line bundle to a web of 5-branes with 7-branes.

\subsection{Asymptotic webs of 5-branes with 7-branes and log Calabi--Yau surfaces} \label{section_webs_log_CY}

Let $W^{\mathrm{asym}}$ be an asymptotic web of 5-branes with 7-branes, defined by an asymptotic web of 5-branes $\overline{W}^{\mathrm{asym}}$ and 
a 7-brane data $\mathbf{a}= ((a_{e,i})_{1\leq i \leq n_e})_{e \in L(\overline{W}^{\mathrm{asym}})}$ as defined in \S\ref{def_web_57_at_infinity}.
In what follows, we explain how to attach to 
$W^{\mathrm{asym}}$
a log Calabi--Yau surface $(Y,D)$ and a line bundle $L$ on $Y$.

Recall from \S\ref{webs_toric_surfaces} that the asymptotic 5-brane $\overline{W}^{\mathrm{asym}}$ determines a toric surface $(\overline{Y}, \overline{D})$ with an ample line bundle $\overline{L}$. 
Every leg $e \in L(\overline{W}^{\mathrm{asym}})$
corresponds to a toric irreducible component $\overline{D}_e$ of $\overline{D}$. 
Denote by $\overline{D}_e^\circ$ the open stratum of $\overline{D}_e$, that is, the complement in $\overline{D}_e$ of the 0-dimensional strata of $\overline{D}$.
For every leg $e \in L(\overline{W}^{\mathrm{asym}})$, pick $n_{e}$ distinct points $y_{e,i}$ of
$\overline{D}_e^\circ$, where $1 \leq i \leq n_e$.
Let $Y$ be the projective surface obtained by blowing up the toric surface $\overline{Y}$ at the points $y_{e,i}$ for all $e \in L(\overline{W}^{\mathrm{asym}})$ and $1\leq i \leq n_e$. 
Then, denoting by $D$ the strict transform of $\overline{D}$ in $Y$, with irreducible components $D_e$, the pair $(Y,D)$ is a log Calabi--Yau surface. Finally, we define a line bundle $L$ on the surface $Y$ by 
\begin{equation} \label{eq_line_bundle_L}
L := (\pi^\star \overline{L}) \otimes \mathcal{O}_Y \left(-\sum_{e\in L(\overline{W}^{\mathrm{asym}})} \sum_{i=1}^{n_e} a_{e,i} E_{e,i} \right)\,,\end{equation}
where $\pi: Y \rightarrow \overline{Y}$ is the blow-up map, and $E_{e,i}=\pi^{-1}(y_{e,i})$
are the exceptional divisors in $Y$ over the points $y_{e,i}$ -- see Figure \ref{Fig31}.

\begin{figure}[hbt!]
\center{\scalebox{0.9}{\includegraphics{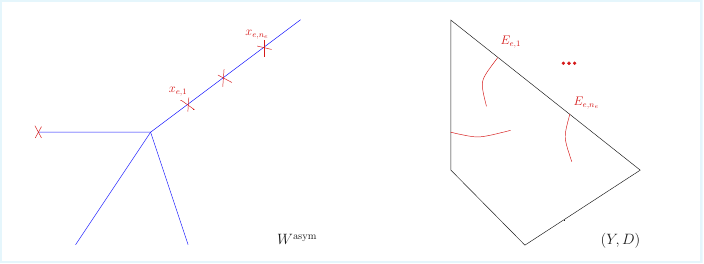}}}
\caption{The log Calabi--Yau surface associated to an asymptotic web of 5-branes with 7-branes.}
\label{Fig31}
\end{figure}

The above definition of $(Y,D)$ depends on the choice of the points $y_{e,i}$ to blow up on the toric boundary $\overline{D}$ of $\overline{Y}$. Modifying the position of points $y_{e,i}$ changes the complex structure on $Y$. 
In particular, the deformation class of the log Calabi--Yau surface with line bundle $(Y,D,L)$ is independent of any choice. However, depending on the ``type'' of a web $W^{\mathrm{asym}}$ as discussed below, we will impose certain restrictions on the point configuration $y_{e,i}$.
There are two possible types of asymptotic webs of $5$-branes with $7$-branes, $W^{\mathrm{asym}}$, depending on if all $5$ branes end on $7$-branes or not.

\begin{definition}
We call an asymptotic web of $5$-branes with $7$-branes $W^{\mathrm{asym}}$ of \emph{Type I}, respectively of \emph{Type II}, in the following situations:
\begin{itemize}
    \item[i)] \textbf{Type I}: If $W^{\mathrm{asym}}$ contains at least a leg, that is at least one 5-brane of $W^{\mathrm{asym}}$ does not end on a 7-brane. Equivalently, we have $L \cdot D >0$.
    In particular, there exists $e \in L(\overline{W}^{\mathrm{asym}})$ such that 
    \begin{equation}
    \mathrm{deg}\, L|_{D_e} = w_e - \sum_{i=1}^{n_e} a_{e,i} >0\,.
    \end{equation}
   \item[ii)] \textbf{Type II}: If $W^{\mathrm{asym}}$ does not contain any leg, that is all 5-branes end on 7-branes. Equivalently, we have $L \cdot D=0$, that is,
     \begin{equation} \label{eq_deg_0}\mathrm{deg}\, L|_{D_e} = w_e - \sum_{i=1}^{n_e} a_{e,i} =0\,,\end{equation}
     for all $e \in L(\overline{W}^{\mathrm{asym}})$. 
\end{itemize}
\end{definition}

For Type I webs $W^{\mathrm{asym}}$, we do not impose any restrictions on the configuration of the points $y_{e,i}$.
In particular, by \cite[Proposition 2.9]{GHK_moduli}, one can assume that $(Y,D)$ is \emph{generic in its deformation class}, in the sense that there does not exist any effective line bundle $L'$ on $Y$ whose restriction $L'|_{D}$ is trivial.

For Type II webs, we require that the points $y_{e,i}$ 
satisfy the \emph{Menelaus condition} with respect to $L$, that is, that the line bundle $L|_D$ is trivial -- see for example \cite{Mik_menelaus} for a previous use of this terminology. By \cite[Proposition 2.9]{GHK_moduli}, one can then assume that $(Y,D)$ is \emph{generic among its deformations preserving the Menelaus condition with respect to $L$}, in the sense that the only effective line bundles $L'$ on $Y$ with trivial restriction $L'|_{D}$ are proportional to $L$.

\begin{remark}
\label{Rem: Menelaus}
     Since $D$ is a nodal curve of arithmetic genus one, 
     the group of $\mathrm{Pic}^0(D)$ of line bundles on $D$ with degree $0$ on every irreducible component of $D$ is isomorphic to $\CC^\star$. 
     By Equation \ref{eq_deg_0}, we have $L|_D \in \mathrm{Pic}^0(D)$, and so the Menelaus condition can be written as $L|_D=1 \in \CC^\star$. Recall that the points $y_{e,i}$ lie in the interior of the big torus orbits in the boundary components $\overline{D}_e$'s of $\overline{D}$. Denote by $z_e$ a coordinate on the component $\overline{D}_e$, with the convention that we choose an orientation $\overline{D}$ and denote the coordinates of the two nodal points intersecting adjacent components by $0$ and $\infty$ respectively, following this orientation.
    Denoting by $z_e(y_{e,i}) \in \CC^*$ the coordinates of $y_{e,i}$, the Menelaus condition amounts to the requirement:
     \begin{equation}
\label{eg: Menelaus}
         \prod_{e \in  L(\overline{W}^{\mathrm{asym}})} \prod_{i=1}^{n_e} z_e(y_{e,i})^{a_{e,i}} = 1 \,.
     \end{equation}
     In particular, the Menelaus condition is a codimension one condition in the moduli space $\prod_{e \in  L(\overline{W}^{\mathrm{asym}})} (\overline{D}_e^{\circ})^{n_e}$ of configurations of points $y_{e,i}$.
\end{remark}

\begin{definition}
    \label{Def: generic}
A log Calabi--Yau surface $(Y,D)$ with a line bundle $L$ associated to an asymptotic web of $5$-branes with $7$-branes  is called \emph{generic with respect to $L$}, if $(Y,D)$ and $L$ satisfy the following:

\begin{itemize}
    \item[i)] If $W^{\mathrm{asym}}$ is of type I, then $(Y,D)$ is generic in its deformation class.
\item[ii)]  If $W^{\mathrm{asym}}$ is of type II, then $(Y,D)$ is generic among its deformations preserving the Menelaus condition with respect to $L$.
\end{itemize}
\end{definition}

In the remainder of the paper, given the data of a logarithmic Calabi--Yau surface $(Y,D)$ with a line bundle $L$, associated with $W^{\mathrm{asym}}$, we assume that $(Y,D)$ is generic with respect to $L$.

\begin{example}
Let $Y=\PP^2$ and $D$ a triangle formed by three lines $\ell_1, \ell_2$, and $\ell_3$, then three points $p_1 \in \ell_1$, $p_2 \in \ell_2$, and $p_3 \in \ell_3$ satisfy the Menelaus condition if and only if they are collinear. This is the set-up of the classical Menelaus theorem in elementary geometry.
\end{example}

\begin{remark}
\label{rem: compatibility}
    If $(Y,D)$ is the log Calabi--Yau surface corresponding to an asymptotic web of 5-branes with 7-branes $W^{\mathrm{asym}}$, then the complement $U:=Y 
    \setminus D$ is the open Calabi--Yau surface attached
    as in \S \ref{sec:seven_branes}
    to the asymptotic configuration of 7-branes in $W^{\mathrm{asym}}$. The 5-branes in $W^{\mathrm{asym}}$ determine a particular log Calabi--Yau compactification $(Y,D)$ of $U$ and a line bundle $L$ on $Y$. 
\end{remark}

\subsection{Webs of $5$-branes with $7$-branes and degenerations of log Calabi--Yau surfaces} \label{sec: degeneration}

In this section, we explain how a web of $5$-branes with $7$-branes $W$
defines a one-parameter degeneration $(\mathcal{Y},\mathcal{D}) \to \mathbb{C}$ of the log Calabi--Yau surface $(Y,D)$, associated to $W^{\mathrm{asym}}$ as in \S \ref{subsec:Webs of 5-branes with 7-branes}. We also explain how to construct a degeneration $\mathcal{L}$ of the line bundle $L$ on $Y$.

Let $W$ be a web of $5$-branes with $7$-branes, obtained from a web of $5$ branes $\overline{W}$ which is dual to an integral convex piecewise linear function as in \S\ref{section_degeneration_toric_surfaces}.
As reviewed in \S\ref{section_degeneration_toric_surfaces}, 
the web $\overline{W}$ defines a polarized toric degeneration 
\[  (\overline{\mathcal{Y}}, \overline{\mathcal{D}}, \overline{\mathcal{L}}
) \longrightarrow \CC \,,
\]
with general fiber the polarized toric surface $(\overline{Y},\overline{D},\overline{L})$ defined by the associated asymptotic web $\overline{W}^{\mathrm{asym}}$ as in \S\ref{section_toric_surfaces}.

We obtain the degeneration $\mathcal{Y} \to \CC$ from 
$\overline{\mathcal{Y}} \to \CC$ as follows. For each leg $e \in L(\overline{W})$, and for $1 \leq i \leq n_e$, pick a section $\sigma_{e,i}$ of $\overline{\mathcal{Y}} \to \CC$ satisfying the following conditions:
\begin{itemize}
    \item[i)] for all $t \in \CC \setminus \{ 0 \}$, $\sigma_{e,i}(t)$ is contained in the big torus orbit of the irreducible component $\overline{D}_{t,e} \subset \overline{D}$, corresponding to the leg in the asymptotic fan parallel to $e$.
    \item[ii)] For $t = 0$, $\sigma_{e,i}(0)$ is contained in the big torus orbit of the irreducible component $\overline{D}_{0,e}$ in the double locus of $\overline{\mathcal{Y}}_0$, corresponding to the leg parallel to $e$.
\end{itemize}

Now, we set $\mathcal{Y}$ to be the blow-up of $\overline{\mathcal{Y}}$ along the image of the sections $\sigma_{e,i}$, and denote by $\mathcal{D}$ the strict transform of $\overline{\mathcal{D}}$.
In addition, we define a line bundle $\mathcal{L}$ on $\mathcal{Y}$ by 
\begin{equation} \label{eq_line_bundle_LL}
\mathcal{L} := (p^\star \overline{\mathcal{L}}) \otimes \mathcal{O}_{\mathcal{Y}} \left(-\sum_{e\in L(\overline{W}^{\mathrm{asym}})} \sum_{i=1}^{n_e} a_{e,i} \mathcal{E}_{e,i} \right)\,,\end{equation}
where $p: \mathcal{Y} \rightarrow \overline{\mathcal{Y}}$ is the blow-up map, and $\mathcal{E}_{e,i}$
are the exceptional divisors in $\mathcal{Y}$ over the images of the sections $\sigma_{e,i}$.
By construction, the restriction of $\mathcal{L}$ to a general fiber $(\mathcal{Y}_t,\mathcal{D}_t)$, which is deformation equivalent to $(Y,D)$, 
is deformation equivalent to the line bundle $L$ defined in Equation \eqref{eq_line_bundle_L}. From now on we refer to 
\begin{equation} \label{eq: degeneration}
 (\mathcal{Y},\mathcal{D},\mathcal{L}) \longrightarrow \CC \,,
\end{equation}
as the \emph{degeneration of log Calabi--Yau surfaces associated to $W$}.

\section{Consistent webs of $5$-branes with $7$-branes}
\label{section_consistent_webs_57}

In \S\ref{section_consistent}, we provide a geometric definition of supersymmetric webs of 5-branes with 7-branes and of a more restricted class of consistent web of 5-branes with 7-branes. In \S\ref{section_geometric_classification}, we prove a geometric classification of 5d SCFTs defined by consistent webs of 5-branes with 7-branes. In \S\ref{section_consistent_webs}, we establish a relationship between consistent webs of 5-branes with 7-branes and the birational geometry of one-parameter degenerations of log Calabi--Yau surfaces. In \S\ref{sec pushing}, we reformulate the existence of minimal models for these degenerations in terms of pushing 7-branes in the web of 5-branes with 7-branes. After describing local examples of pushing 7-branes in \S\ref{sec_examples}, 
we define generic consistent webs of 5-branes with 7-branes and we give a characterization of these webs in terms of pushing 7-branes in \S\ref{section_fully_pushed}.

\subsection{Consistent asymptotic webs of $5$-branes with $7$-branes and 5d SCFTs} \label{section_consistent}
In this section, we first define a geometric notion of  \emph{supersymmetric} asymptotic webs, that are expected to define 5d SCFTs. Then, we introduce the more restrictive notion of \emph{consistent} asymptotic webs, and we argue that any 5d SCFT defined by a supersymmetric web can be obtained from a 5d SCFT defined by a consistent web by adding free hypermultiplets. 
In particular, restricting attention to consistent webs entails no loss of generality when constructing non-trivial 5d SCFTs.

Let $W^{\mathrm{asym}}$ be an asymptotic web of $5$-branes with $7$-branes, and $(Y,D,L)$ the log Calabi--Yau surface  with line bundle associated to $W^{\mathrm{asym}}$, 
so that $(Y,D)$ is generic with respect to $L$ as in Definition \ref{Def: generic}. 
Then,  $W^{\mathrm{asym}}$ is called \emph{supersymmetric} 
if the 5d theory on the common intersection $\RR^{1,4}$ of the $5$-branes 
as in Equation \eqref{Eq: R14} has $\mathcal{N}=1$ supersymmetry, that is, $8$ real supercharges. We can then expect the low-energy description of this theory to define a 5d SCFT.
We give a geometric characterization of supersymmetric asymptotic webs in Definition \ref{def_supersymmetric_asymp_web} below.

Recall that $U= Y\setminus D$ is the open Calabi--Yau surface associated to the configuration of $7$-branes in $W^{\mathrm{asym}}$. It follows from the duality between Type IIB string theory on $S^1$ and M-theory on $T^2$, and from the M-theory interpretation of $U$ given in \S\ref{section_string_open_CY} that, if the 5d theory defined by $W^{\mathrm{asym}}$ has $\mathcal{N}=1$ supersymmetry, then, the 4d $\mathcal{N}=2$ theory obtained by compactifying this 5d theory on $S^1$ should admit a dual description as the  4d $\mathcal{N}=2$ theory on $\RR^{1,3}$
defined by M-theory on $\RR^{1,3} \times U \times \RR^3$ with an M5-brane wrapped on $\RR^{1,3} \times C^\circ$, where $C^\circ$ is an algebraic curve in $U$, the \emph{Seiberg--Witten curve} of the 4d $\mathcal{N}=2$ theory. Moreover, the compactification $C$ of $C^\circ$ in $Y$ should be an algebraic curve in $Y$ contained in the linear system $|L|$ of the line bundle $L$, and not containing any 
0-dimensional stratum of $D$.
In the T-dual context of string junctions ending on 7-branes, the existence of an M-theory dual holomorphic curve as a criterion for supersymmetry has been discussed in
    \cite{constraintsBPS, Istringjunctions, KSstringjunction, KLstringnetwork, MOstringjunction,  MNS}.

This motivates the following mathematical definition of supersymmetric asymptotic webs.

\begin{definition} \label{def_supersymmetric_asymp_web}
An asymptotic web of 5-branes with 7-branes $W^{\mathrm{asym}}$
is \emph{supersymmetric} if, denoting by $(Y,D,L)$ the log Calabi--Yau surface with line bundle associated to $W^{\mathrm{asym}}$, there exists a (possibly disconnected) curve $C$ in the linear system $|L|$ and not containing any 
0-dimensional stratum of $D$.
\end{definition}

As described in \S\ref{section_webs_log_CY}, the log Calabi--Yau surface $(Y,D)$ is obtained from a toric surface 
$(\overline{Y},\overline{D})$ by interior blow-ups of points $y_{e,i}$ on $\overline{D}$, with exceptional curves $E_{e,i}$. Moreover, by Equation 
\eqref{eq_line_bundle_L}, we have $L=p^\star \overline{L} \otimes \mathcal{O}(-\sum_{e,i} a_{e,i}E_{e,i})$, where $a_{e,i} \in \ZZ_{\geq 1}$ is the number of 5-branes ending on the 7-brane $x_{e,i}$, and $p: \overline{Y} \rightarrow Y$ is the blow-up map. In particular, if $C$ is a smooth curve in $|L|$ that does not contain any of the exceptional curves $E_{e,i}$, we have $C \cdot E_{e,i}=a_{e,i}$, and the image $\overline{C}:=p(C)$ of $C$ in $\overline{Y}$ is a curve in the linear system $|\overline{L}|$ passing through the point $y_{e,i}$
with multiplicity $a_{e,i}$. In particular, $y_{e,i}$ is a singular point of $\overline{C}$ if $a_{e,i}>1$, and $p|_C: C \rightarrow \overline{C}$ is a resolution of the singularities of $\overline{C}$ that separate the $a_{e,i}$ branches of $\overline{C}$ passing through $y_{e,i}$ -- see Figure \ref{Fig33}.

\begin{remark}
    In the physics literature, the proposed Seiberg--Witten curves for the compactification on $S^1$ of a 5d SCFT defined by a web of 5-branes with 7-branes are the curves $\overline{C}^\circ = \overline{C} \cap (\CC^\star)^2$. For instance, the condition that $y_{e,i}$ is a multiple point of the curve $\overline{C}$ with multiplicity $a_{e,i}$ is equivalent to the divisibility conditions on Laurent polynomials proposed in \cite[(2.15)]{SW_curve} to characterize Seiberg--Witten curves. Instead, we propose to view the curve $C^\circ =C \cap U$ in the open Calabi--Yau surface $U$ as the proper description of the Seiberg--Witten curve.
\end{remark}

\begin{figure}[hbt!]
\center{\scalebox{0.9}{\includegraphics{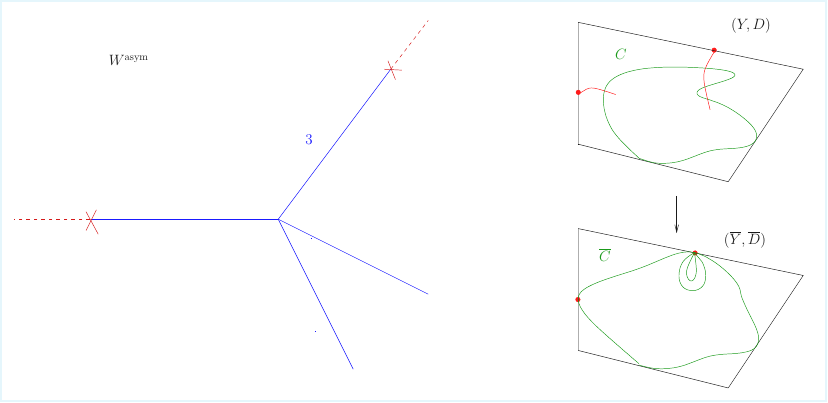}}}
\caption{A curve $C \in |L|$ in the log Calabi--Yau surface $(Y,D)$, and the corresponding curve $\overline{C} \in |\overline{L}|$ in the toric surface $(\overline{Y},\overline{D})$.}
\label{Fig33}
\end{figure}

In what follows, we compare the notion of supersymmetric asymptotic web given by Definition \ref{def_supersymmetric_asymp_web} 
with supersymmetric conditions usually formulated in the physics literature in terms of Hanany--Witten 
moves on webs of 5-branes with 7-branes.
Recall from \S\ref{sec:seven_branes} that Hanany--Witten moves on 
configurations of $7$-branes are defined by pushing a focus-focus singularity across the origin, and rotating the adjacent cut clockwise, while changing the positions of all other $7$-branes by a shear as in Equation \eqref{eq_seed_mutation}. 
Hanany--Witten moves on webs of $5$-branes with $7$-branes change in addition the configuration of $5$-branes as follows.
Consider the Hanany--Witten move obtained by moving the $(p_e, q_e)$ $7$-brane $x_{e,i}$ across the origin, for a leg $e\in L(W^{\mathrm{asym}})$ and some $1\leq i \leq n_e$. The 7-brane $x_{e,i}$ becomes a $(-p_e, -q_e)$ $7$-brane $x'_{e,i}$, and the other 7-branes are changed according to the shear transformation defined by Equation \eqref{eq_seed_mutation}. We apply the same shear transformation to all the $5$-branes that do not end at $x_{e,i}$. 
Moreover, the $5$-branes that ended in $x_{e,i}$ disappear, and the number $a_{e,i}'$ of new 5-branes passing through the origin and ending on $x_{e,i}'$ is determined by imposing the balancing condition of the 5-brane charge at the origin. The balancing condition uniquely determines $a_{e,i}'$ as an integer, but this integer is not necessarily positive. 
Denoting by $w_{-e} \in \ZZ_{\geq 0}$ the weight of the leg of direction $(-p_e, -q_e)$ in $\overline{W}^{\mathrm{asym}}$, so that the weight of the edge
connecting the origin to $x_{e,i}'$ is $w_{-e} +a_{e,i}'$, see Figure \ref{Fig:HW_web}, we are in one of the three following situations:
\begin{itemize}
    \item[i)] \textbf{Case I}: We have
    \begin{equation}
\label{Eq: case 1}
a_{e,i}' < -w_{-e}   \,.     
    \end{equation}
In this situation, the resulting configuration of $5$-branes with $7$-branes $(W^{\mathrm{asym}})'$ obtained by the Hanany--Witten move along $x_{e,i}$ contains an edge with negative multiplicity.
     \item[ii)] \textbf{Case II}: We have
   \begin{equation}
\label{Eq: case 2}
-w_{-e} \leq a_{e,i}'\leq 0 \,.
    \end{equation}   
In this situation, the 7-brane $x_{e,i}'$ can be viewed as the endpoint of $-a_{e,i}'$ 5-branes which come from infinity, or ending on another 7-brane, see Figure \ref{Fig19}. 
These 5-branes are not connected to the part of the web of 5-branes containing the origin, and so can be pushed to infinity together with the 7-brane $x_{e,i}'$. 
      \item[iii)] \textbf{Case III}: We have:
 \begin{equation}
\label{Eq: case 3}
a_{e,i}' >0  \,.
    \end{equation}       
In this situation, the resulting configuration of $5$-branes with $7$-branes $(W^{\mathrm{asym}})'$ obtained by the Hanany--Witten move along $x_{e,i}$ is another asymptotic web of 5-branes with 7-branes as in Definition \ref{def_asymptotic_57}.   
\end{itemize}

\begin{figure}[hbt!]
\center{\scalebox{1.0}{\includegraphics{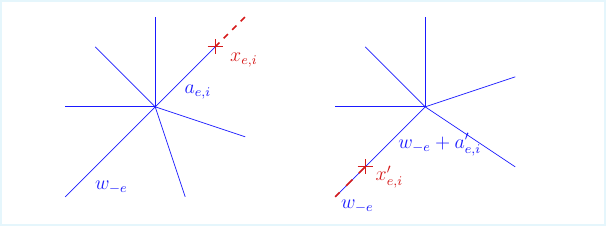}}}
\caption{Hanany--Witten move for a web of 5-branes with 7-branes.}
\label{Fig:HW_web}
\end{figure}

In the physics literature, Case I is interpreted as evidence that the web is not supersymmetric since an edge with a negative multiplicity can be viewed as a supersymmetry-breaking anti-5-brane. 
We show in Lemma \ref{lem_susy} below that this interpretation is compatible with our geometric Definition \ref{def_supersymmetric_asymp_web} of a supersymmetric web. To do that, we first explain that 
the above combinatorial description of Hanany--Witten moves at the level of webs of 5-branes with 7-branes has a natural geometric interpretation in terms of elementary cluster transformations of log Calabi--Yau surfaces, 
as discussed in \S\ref{section_geometric_HW} in the case of configurations of 7-branes.

Recall from \S \ref{section_webs_log_CY} that the log Calabi--Yau surface $(Y,D)$ is a non-toric blow-up of a toric surface $(\overline{Y},\overline{D})$ with fan given by the support of the web of 5-branes $\overline{W}^{\mathrm{asym}}$. 
Up to replacing $(\overline{Y},\overline{D})$ by a toric blow-up, one can assume that the fan of $(\overline{Y},\overline{D})$ contains both rays $\RR_{\geq 0}(p_e,q_e)$, 
and $\RR_{\geq 0}(-p_e,-q_e)$.
Then, as in \S\ref{section_geometric_HW}, 
the projection $\RR^2 \rightarrow \RR^2/\RR(p_e,q_e)=\RR$ induces a toric morphism $\nu: \overline{Y} \rightarrow \PP^1$, which is generically a $\PP^1$-fibration. 
Moreover, the irreducible components $\overline{D}_e$
and $\overline{D}_e'$ of $\overline{D}$ corresponding respectively to the rays $\RR_{\geq 0}(p_e,q_e)$ and $\RR_{\geq 0}(-p_i,-q_i)$ of the fan, are sections of $\nu$. The construction of $Y$ from $\overline{Y}$ involves blowing up the point $y_{e,i}$ on $\overline{D}_e$, which 
creates an exceptional $(-1)$-curve $E_{e,i}$ in $Y$. The strict transform of the $\PP^1$-fiber of $\nu$ passing through $y_{e,i}$ is also a $(-1)$-curve $E_{e,i}'$, which intersects the strict transform $D_e'$ of $\overline{D}_e'$. 
Hence, one can construct a new toric surface $\overline{Y}'$ from $Y$ by contracting all the exceptional curves coming from blowing up the points $y_{f,j}$ with $(f,j) \neq (e,i)$, and the $(-1)$-curve $E_{e,i}'$. 
Geometrically, the Hanany--Witten move from $x_{e,i}$ to $x_{e,i}'$ encodes the passage from describing $Y$ as a blow-up of the toric surface $\overline{Y}$ to describing $Y$ as a blow-up of the toric surface $\overline{Y}'$. 
In particular, the weight $w_{-e}$ in $W^{\mathrm{asym}}$
is given by the intersection number
\begin{equation}\label{eq_w}
w_{-e}=L \cdot D_e' \,,\,,\end{equation}
and the number $a_{e,i}'$ of 5-branes ending on the 7-brane $x_{e,i}'$ is given by the intersection number 
\begin{equation} \label{eq_a}
a_{e,i}' = L \cdot E_{e,i}'\,.\end{equation}

We can finally state how the Definition \ref{def_supersymmetric_asymp_web} of supersymmetric webs is related to Hanany--Witten moves:

\begin{lemma} \label{lem_susy}
Let $W^{\mathrm{asym}}$ be a supersymmetric asymptotic web of 5-branes with 7-branes. Then, edges of negative multiplicities are never produced by applying a sequence of Hanany--Witten moves to $W^{\mathrm{asym}}$, that is, Case I as in Equation \eqref{Eq: case 1} never occurs. 
\end{lemma}

\begin{proof}
Since $W^{\mathrm{asym}}$ is supersymmetric, it follows from Definition \ref{def_supersymmetric_asymp_web} that there exists a curve $C$ in $|L|$ not containing any 0-dimensional stratum of $D$. Keeping the notation used in the discussion of the geometric interpretation of Hanany--Witten moves, let $E_{e,i}'$ be the $(-1)$-curve corresponding to the 7-brane $x_{e,i}'$.  
Denote by $C'$ the union of irreducible components of $C$ not containing $E_{e,i}'$, so that $C=C'+wE_{e,i}'$ for some $w \in \ZZ_{\geq 0}$.
By Equations \eqref{eq_a} and \eqref{eq_w}, we have 
\[ a_{e,i}'+w_{-e}= C \cdot E_{e,i}' + C \cdot D_e'=C' \cdot E_{e,i}'-w +C' \cdot D_{e}' +w = C' \cdot E_{e,i}' +C' \cdot D_{e}'\,.\]
Since $C'$ does not contain $E_{e,i}'$, we have $ C' \cdot E_{e,i}' \geq 0$. On the other hand, since $C'$ does not contain any 0-dimensional stratum of $D$, we also have $C' \cdot D_{e}'\geq 0$. Hence, we conclude that 
$a_{e,i}'+w_{-e} \geq 0$, that is, we are not in Case I as in Equation \eqref{Eq: case 1}.
\end{proof}

Finally, we define below a \emph{consistent} asymptotic web of 5-branes with 7-branes as a supersymmetric web such that, for any sequence of Hanany--Witten moves, we are in Case III and never in Case II, that is, no 7-brane and no 5-brane can be detached from the web and moved freely to infinity after any sequence of Hanany--Witten moves. We first give a geometric reformulation of this condition in Lemma \ref{lem_caseIII_consistent} below using the following terminology: An \emph{interior} (-1)-curve in a log Calabi--Yau surface $(Y,D)$ is a rational curve $E$ in $Y$ such that $E^2=-1$ and which is not contained in $D$, so that $E \cdot D=1$ by the adjunction formula.
An \emph{internal} (-2)-curve in $(Y,D)$ is a rational curve $E$ in $Y$ such that $E^2=-2$ and which is not contained in $D$, so that $E\cdot D=0$ by the adjunction formula.

\begin{lemma} \label{lem_caseIII_consistent}
    Let $W^{\mathrm{asym}}$ be a supersymmetric asymptotic web of 5-branes with 7-branes, with associated log Calabi--Yau surface with line bundle $(Y,D,L)$. Then, the following are equivalent:
    \begin{itemize}
        \item[i)] For any sequence of Hanany--Witten moves applied to $W^{\mathrm{asym}}$, we are always in Case III as in Equation \eqref{Eq: case 3}, and never in Case II as in Equation \eqref{Eq: case 2}.
        \item[ii)] For every curve $E$ in $Y$ whose strict transform in a corner blow-up of $Y$ is either an interior $(-1)$-curve or an internal $(-2)$-curve, we have $L \cdot E>0$.
    \end{itemize}
\end{lemma}

\begin{proof}
The fact that ii) implies i) follows from the above geometric interpretation of Cases I-III in \eqref{Eq: case 1}--\eqref{Eq: case 3}
for Hanany--Witten moves. Hence, it remains to
show that i) implies ii). Let $E$ be a curve in $Y$ whose strict transform in a corner blow-up of $Y$ is an irreducible $(-1)$-curve. By \cite[Proposition 3.27]{HK_HMS} (see also \cite[Theorem 1]{Blanc}), any two toric models of a log Calabi--Yau surface are connected by a sequence of corner blow-ups/blow-downs and elementary cluster transformations. Hence, there exists a sequence of Hanany--Witten moves on $W^{\mathrm{asym}}$ such that $E$ is the 
$(-1)$-curve $E_{e,i}'$ corresponding to the pushed 7-brane $x_{e,i}'$. The inequality \eqref{Eq: case 3} for this sequence of Hanany--Witten moves implies that $L \cdot E>0$. One can argue similarly for an internal $(-2)$-curve and this concludes the proof.
\end{proof}

Finally, we are ready to define the notion of ``consistency'' for asymptotic webs of $5$-branes with $7$-branes.

\begin{definition}
\label{def: consistency}
An asymptotic web of 5-branes with 7-branes $W^{\mathrm{asym}}$
is \emph{consistent} if it satisfies the following two conditions:
\begin{itemize}
    \item[i)] In the associated log Calabi--Yau surface with line bundle $(Y,D,L)$, there exists a curve $C\in |L|$  not containing any 0-dimensional stratum of $D$, that is, $W^{\mathrm{asym}}$ is supersymmetric.
    \item[ii)] For every curve $E$ in $Y$ whose strict transform in a corner blow-up of $Y$ is an interior $(-1)$-curve or an internal $(-2)$-curve, we have $L \cdot E>0$.
\end{itemize}
\end{definition}

If $W^{\mathrm{asym}}$ is a supersymmetric but not consistent asymptotic web, then, after Hanany--Witten moves, some 7-branes and 5-branes can be detached from the remainder part of the web. In terms of the log Calabi--Yau surface with line bundle $(Y,D,L)$, there are several possibilities if we are in Case II as in Equation \eqref{Eq: case 2}:
\begin{itemize}
    \item[i)] A 7-brane is sent to infinity as on the left of Figure \ref{Fig19}. In this case, there exists an interior $(-1)$-curve $E$ such that $L \cdot E=0$ and that can be contracted to a point.
    \item[ii)] A 5-brane attached to a 7-brane is sent to infinity as in the middle of Figure \ref{Fig19}. 
    In this case, there is an interior $(-1)$-curve $E$ such as a multiple $kE$ is a connected component of $C$, and so $C$ can be replaced by $C\setminus kE$.
    \item[iii)] A 5-brane connected to two 7-branes is sent to infinity as on the right of Figure \ref{Fig19}. 
    In this case, there is an internal $(-2)$-curve $E$ such as a multiple $kE$ is a connected component of $C$, and so $C$ can be replaced by $C\setminus kE$.
\end{itemize}
Each instance of case i) strictly decreases the Picard rank of $Y$, and thus case i) can occur only finitely many times. Cases (ii) and (iii) are also limited in number, since the number of connected components of a curve in the linear system $|L|$ is bounded. Therefore, iterating these operations finitely many times, we finally obtain a log Calabi--Yau surface with line bundle $(Y',D',L')$ corresponding to a consistent asymptotic web. 
The curve $C \in |L|$
is isomorphic to the disjoint union of a curve isomorphic to a curve $C'\in |L'|$ and of finitely many disjoint multiple interior $(-1)$-curves and internal $(-2)$-curves in $Y$. Physically, every 7-brane or 5-brane sent to infinity corresponds to a free hypermultiplet, as described in \cite[Appendix A.3]{BRGE1}.
Hence, the 5d SCFT defined by a supersymmetric asymptotic web of 5-branes with 7-branes can always be obtained from the 5d SCFT defined by a consistent web by adding free hypermultiplets. Therefore, in what follows, we focus our study on consistent webs.

\begin{figure}[hbt!]
\center{\scalebox{.9}{\includegraphics{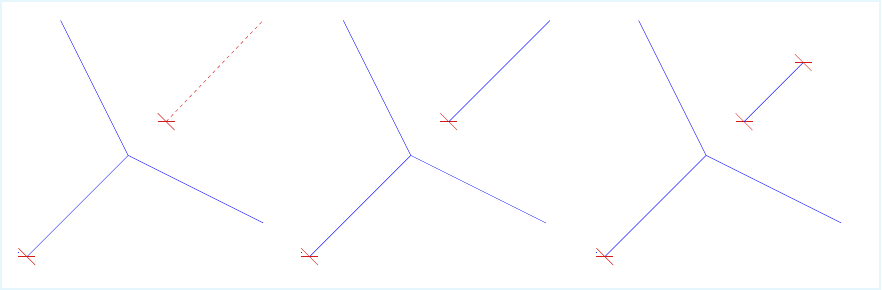}}}
\caption{Possible results of a Hanany--Witten move in Case II as in Equation
\eqref{Eq: case 2}.}
\label{Fig19}
\end{figure}

\subsection{Geometric classification of consistent asymptotic webs of 5-branes with 7-branes}
\label{section_geometric_classification}
We first show that the consistency of webs implies positivity properties of the corresponding line bundle $L$, namely the non-negativity of the self-intersection $L^2$ and the nefness of $L$.

\begin{lemma} \label{lem_L_nef}
    Let $W^{\mathrm{asym}}$ be a consistent asymptotic web of 5-branes with 7-branes, with associated log Calabi--Yau surface $(Y,D)$ and line bundle $L$. Then, the following conditions hold:
    \begin{itemize}
    \item[i)] for every irreducible component $C_i$ of a general curve $C \in |L|$, we have $C_i^2 \geq 0$,
    \item[ii)] we have $L^2 \geq 0$ and the line bundle $L$ is \emph{nef}, that is, we have $L \cdot C' \geq 0$ for every curve $C'$ in $Y$,
    \item[iii)] the linear system $|L|$ has no fixed component, that is, there is no curve $C'$ in $Y$ such that every $C \in |L|$ contains $C'$ as an irreducible component.
    \end{itemize}
\end{lemma}

\begin{proof}
FFirst assume that $W^{\mathrm{asym}}$ is of Type I, so that $(Y,D)$ is generic. Let $C$ be a general curve in $|L|$. Since $(Y,D)$ is generic, 
every reduced irreducible component of $C$ with negative self-intersection is necessarily an interior $(-1)$-curve. Hence, one can write
$C=C'+\sum_i a_i E_i$, where all irreducible components of $C'$ have nonnegative self-intersection, $w_i \in \ZZ_{\geq 0}$, and $E_i$ are interior $(-1)$-curves. If $E_i \cap E_{j}\neq 0$, then $(E_i+E_j)^2 \geq 0$, and the curve $E_i+ E_j$ can be deformed non-trivially in $Y$, in contradiction with the assumption that $C$ is general in $|L|$. Hence, the curves $E_i$ are disjoint. Since all irreducible components of $C'$ have nonnegative self-intersection, $\mathcal{O}(C')$ is nef by \cite[Lemma 4.14]{friedman2015geometry}. On the other hand, we have $C \cdot E_i = C' \cdot E_i -a_i>0$ for all $i$ since $W^{\mathrm{asym}}$ is consistent, and so $L=\mathcal{O}(C)$ is also nef. Moreover, we have 
\[C^2=(C')^2+ 2 \sum_i a_i (C' \cdot E_i)-\sum_i a_i^2=(C')^2+\sum_i a_i (C' \cdot E_i)
+\sum_i a_i((C' \cdot E_i)-a_i) \geq 0 \,,
\]
which proves Lemma \ref{lem_L_nef} ii).

If $C^2>0$, then $L=\mathcal{O}(C)$ is big and nef, and so, by \cite[Theorem 4.12]{friedman2015geometry}, we have the following possibilities:
\begin{itemize}
    \item[i)] $L=\mathcal{O}(C)$ does not have a fixed component. Then $C=C'$ and $a_i=0$.
    \item[ii)] $C=C'+E$ for an interior $(-1)$-curve such that $C' \cdot E =1$. In particular, we have $C \cdot E=0$, which contradicts the consistency of $W^{\mathrm{asym}}$.
\end{itemize}
Therefore, we are always in Case i) and Lemma \ref{lem_L_nef} i) and iii) follow.

If $C^2=0$, then we necessarily have $a_i=0$ for all $i$, and so Lemma \ref{lem_L_nef} i) also follows in this case. Finally, if $C^2=0$, Lemma \ref{lem_L_nef} iii) follows by \cite[Theorem 4.19]{friedman2015geometry}. Finally, the situation where $(Y,D)$ is of Type II can be treated similarly.
\end{proof}

The following result gives a formula for the self-intersection $L^2$.

\begin{lemma} \label{lem_L2}
Let $W^{\mathrm{asym}}$ be an asymptotic web of 
5-branes with 7-branes, defined by an asymptotic web of 5-branes $\overline{W}^{\mathrm{asym}}$ and a 7-brane data $\mathbf{a}=((a_{e,i})_{1\leq i\leq n_e})_{e\in L(\overline{W}^{\mathrm{asym}})}$. Let $(Y,D,L)$ be the corresponding log Calabi--Yau surface with line bundle, and $\overline{P}$ the lattice polygon dual to $\overline{W}^{\mathrm{asym}}$. Then, we have 
\begin{equation} \label{eq_L2}
L^2 = 2\, \mathrm{Area}(\overline{P})\, - \sum_{e\in L(\overline{W}^{\mathrm{asym}})} \sum_{i=1}^{n_e} a_{e,i}^2 \,.\end{equation}
\end{lemma}

\begin{proof}
Using that the exceptional curves $E_{e,i}$ are $(-1)$-curves, 
it follows from the definition of $L$ in Equation \eqref{eq_line_bundle_L} that we have 
\[ L^2 = \overline{L}^2 - \sum_{e,i} a_{e,i}^2 \,,\]
where $\overline{L}$ is the ample line bundle on the toric surface $(\overline{Y},\overline{D})$ defined by the lattice polygon $\overline{P}$ associated to $\overline{W}^{\mathrm{asym}}$ as in \S \ref{webs_toric_surfaces}. Finally, we have $\overline{L}^2=2\, \mathrm{Area}(\overline{P})$
by standard toric geometry.
\end{proof}

\begin{lemma}
\label{lem: connected}
Let $W^{\mathrm{asym}}$ be a consistent asymptotic web of 5-branes with 7-branes, with associated log Calabi--Yau surface with line bundle $(Y,D,L)$. Then, the following holds:
\begin{itemize}
    \item[i)] If $L^2 >0$, then a general $C\in |L|$ is a smooth connected curve of genus 
    \begin{equation} \label{eq: genus}
    g= \frac{1}{2}(L^2 - L \cdot D) +1 \,.\end{equation}
    \item[ii)] If $L^2=0$, then we have $L \cdot D=0$, and a general 
    $C \in |L|$ is a disjoint union of $k \geq 1$ smooth genus one curves.
\end{itemize}
\end{lemma}

\begin{proof}
In case i), using that $L^2>0$, $L$ is nef, and all connected components $C_i$ of $C$ satisfy $C_i^2 \geq 0$ by Lemma \ref{lem_L_nef}, the Case III of the proof of \cite[Theorem 4.12]{friedman2015geometry} implies that the general curve $C \in |L|$ is connected. Equation \eqref{eq: genus} follows from the adjunction formula for curves on surfaces and from the fact that $D$ is an anticanonical divisor on $Y$. In case ii), the result follows from \cite[Theorem 4.19]{friedman2015geometry}.
\end{proof}

\begin{remark} \label{remark_r_rule}
If $L^2>0$, that is, in case i) of Lemma \ref{lem: connected}, we necessarily have $g = \frac{1}{2}(L^2-L \cdot D)+1 \geq 0$
since the genus of a smooth connected curve is non-negative. The non-negative integer $g$ is the rank, that is, the dimension of the Coulomb branch of the corresponding 5d SCFT defined by the consistent 
asymptotic web $W^{\mathrm{asym}}$.
This numerical constraint on webs has been discussed in \cite{BRGE1,constraintsBPS, Istringjunctions,  van2020symplectic, van20215d} and is referred to as the ``r-rule" in \cite{ van2020symplectic, van20215d}. However, contrary to general expectations in the physics literature,  if $W^{\mathrm{asym}}$
is a supersymmetric but not necessarily consistent asymptotic web, then $g = \frac{1}{2}(L^2-L \cdot D)+1$ is not in general the rank of the 5d SCFT defined by $W^{\mathrm{asym}}$ and does not have to be nonnegative. Indeed, while $g = \frac{1}{2}(L^2-L \cdot D)+1$ is always the \emph{arithmetic} genus of a curve $C \in |L|$, such a curve is not necessarily connected if the web is not consistent. In general, $C$ is a disjoint union of a smooth connected curve of genus $g'\geq 0$ and of $k\geq 0$ smooth curves of genus $0$, so that $g'=g+k$, and the 5d SCFT defined by the supersymmetric web is obtained by adding $k$ free hypermultiplets to a rank $g'$ 5d SCFT defined by a consistent web, as explained below Definition \ref{def: consistency}. We present a detailed example of this phenomenon in Example \ref{example_gtp}.

If $L^2=0$, that is, we are in case ii) of  Lemma \ref{lem: connected}, the rank of the corresponding 5d SCFT is equal to the number $k$ of disjoint genus one curves in $|L|$.
\end{remark}

\begin{lemma}
\label{lem: consist nef}
Let $W^{\mathrm{asym}}$ be an asymptotic web of $5$-branes with $7$-branes, and $(Y,D,L)$ the associated log Calabi--Yau surface with line bundle. 
\begin{itemize}
    \item[i)] If $L^2 > 0$, then $W^{\mathrm{asym}}$ is consistent if and only if $L$ is nef, the linear system $|L|$ has no fixed component, and, for every curve $E$ in $Y$ whose strict transform in a corner blow-up of $Y$ is an interior $(-1)$-curve or an internal $(-2)$-curve, we have $L \cdot E>0$. 
     \item[ii)] If $L^2 = 0$, then $W^{\mathrm{asym}}$ is consistent if and only if $L$ is nef and the linear system $|L|$ is non-empty and has no fixed component.
\end{itemize}
\end{lemma}

\begin{proof}
In both i) and ii), the ``only if" direction follows from the Definition 
\ref{def: consistency} of a consistent asymptotic web and from Lemma \ref{lem_L_nef}.

To prove the ``if" direction of i),
we first note that, if $L^2>0$, $L$ is nef, and no component of $D$ is a fixed component of $|L|$, then the linear system $|L|$
is non-empty by \cite[Lemma 4.10]{friedman2015geometry}.
If $|L|$ has no fixed component, then, by \cite[Theorem 4.12]{friedman2015geometry}, $|L|$ has also no base point, except if $L \cdot \widetilde{D}=1$, in which case the unique smooth point $p \in D$ such that 
$L|_{D}=\cO_{D}(p)$ is the unique base point of $|L|$. If $|L|$ is base point free, then the general element of $|\widetilde{L}|$ is smooth by Bertini theorem \cite[Corollary III.10.9]{Hartshorne}. If $|L|$ has a base point $p$, then, since $L\cdot D=1$, $p$ is a smooth point of every element of $|L|$, and after blowing up $p$, we obtain a base point free linear system, whose general element is smooth by Bertini theorem \cite[Corollary III.10.9]{Hartshorne}. In any case, the general curve $C \in |L|$ is smooth. Since $|L|$ has no fixed component, it follows that every connected component $C_i$ of $C$ satisfies $C_i^2 \geq 0$, and so $W^{\mathrm{asym}}$ is consistent.
Finally, the ``if" direction in ii) follows from \cite[Theorem 4.19]{friedman2015geometry}.
\end{proof}

Finally, we are ready to provide a geometric classification of consistent asymptotic webs $W^{\mathrm{asym}}$, up to Hanany--Witten moves. 
To do this, denoting by $(Y,D,L)$ the associated log Calabi--Yau surface with line bundle, we analyze the image of the contraction map $c$ on $Y$, which contracts all curves that have intersection number $0$ with $L$. 
The existence of such a contraction map $c$ is ensured by the nefness of $L$ shown in Lemma \ref{lem_L_nef} and the \emph{basepoint free theorem} in birational geometry \cite[Theorem 3.3]{KM}. By Lemma \ref{lem_L_nef}, we have one of the three following cases to analyze:
\begin{itemize}
    \item[i)] $L^2> 0$ and $L\cdot D> 0$: In this case, the image of $(Y,D,L)$ under the contraction map $c$ is a \emph{polarized log Calabi--Yau surface} $(Y^{\mathrm{pol}},D^{\mathrm{pol}},L^{\mathrm{pol}})$, that is, $L^{\mathrm{pol}}$ is ample. Moreover, the linear system $|L^{\mathrm{pol}}|$ has no fixed component.
     \item[ii)] $L^2> 0$ and $L\cdot D = 0$:  In this case, $D$ is contracted to a point that corresponds to a cusp singularity in the image of the contraction $c$. We will also denote the images of $Y$ and $L$ by $Y^{\mathrm{pol}}$ and $L^{\mathrm{pol}}$, respectively.
     The line bundle $L^{\mathrm{pol}}$ is ample and $|L^{\mathrm{pol}}|$ has no fixed component. We refer to the tuple $(Y^{\mathrm{pol}},L^{\mathrm{pol}})$ as a \emph{polarized Calabi--Yau surface with a cusp singularity}.
      \item[iii)] $L^2 = 0$: In this case,  $L\cdot D = 0$ by of Lemma \ref{lem: connected}, ii) and the contraction map $c$ corresponds to an elliptic fibration
      \[ c: Y \longrightarrow \PP^1 \,,\]
such that $D$ is a singular fiber of $c$, and $L=\mathcal{O}(kF)$, where $F$ is a general fiber of $c$, and $k\geq 1$. 
\end{itemize}
Hence, as in \S\ref{section_geometric_HW}, we arrive at the following conclusion: If two consistent asymptotic webs of $5$-branes with $7$-branes are related by Hanany--Witten moves, then; in case i) above the images of the contraction maps applied to the associated log Calabi--Yaus with line bundles, correspond to isomorphic polarized log Calabi--Yau's $(Y^{\mathrm{pol}},D^{\mathrm{pol}},L^{\mathrm{pol}})$, up to deformation equivalence. Similarly, in case ii), the images of the contraction maps correspond to isomorphic polarized Calabi--Yau surfaces with cusp singularities, up to deformation equivalence. Finally, in case iii) above we get isomorphic elliptic fibrations, up to deformation. 
Before stating the main result of this section, we introduce a couple of notations. In the situation i), denote by

$\mathbf{CWebs}_I$: the set of all consistent asymptotic webs such that the associated log Calabi--Yau with line bundle $(Y,D,L)$ satisfies $L^2>0$ and $L \cdot D > 0$, and;

$\mathbf{PolLogCY}$: the set of all polarized log Calabi--Yau surfaces $(Y^{\mathrm{pol}},D^{\mathrm{pol}},L^{\mathrm{pol}})$ such that $|L^{\mathrm{pol}}|$ has no fixed component, up to deformation equivalence. Define
\begin{align}
\label{eq FI}
F_I : \mathbf{CWebs}_I & \longrightarrow \mathbf{PolLogCY} \\
\nonumber
W^{\mathrm{asym}} & \longmapsto (Y^{\mathrm{pol}},D^{\mathrm{pol}},L^{\mathrm{pol}})
\end{align}
where $(Y^{\mathrm{pol}},D^{\mathrm{pol}},L^{\mathrm{pol}})$ is the polarized log Calabi--Yau surface corresponding to the image of the contraction map $c$, as discussed above. Similarly, in the situation ii), denote by

$\mathbf{CWebs}_{II,+}$: the set of all consistent asymptotic webs such that the associated log Calabi--Yau with line bundle $(Y,D,L)$ satisfies $L^2>0$ and $L \cdot D = 0$, and;

$\mathbf{PolCY^{cusp}}$: the set of all polarized Calabi--Yau surfaces with cusp singularities
$(Y^{\mathrm{pol}},L^{\mathrm{pol}})$ such that $|L^{\mathrm{pol}}|$ has no fixed component, up to deformations preserving the ampleness of $L^{\mathrm{pol}}$. Define
\begin{align}
\label{eq FII+}
F_{II,+} : \mathbf{CWebs}_{II,+} & \longrightarrow \mathbf{PolCY^{cusp}} \\
\nonumber
W^{\mathrm{asym}} & \longmapsto (Y^{\mathrm{pol}},L^{\mathrm{pol}})
\end{align}
where $(Y^{\mathrm{pol}},L^{\mathrm{pol}})$ is the image of the contraction map $c$, as discussed above. Finally, in situation iii), denote by

$\mathbf{CWebs}_{II,0}$: the set of all consistent asymptotic webs such that the associated log Calabi--Yau with line bundle $(Y,D,L)$ satisfies $L^2=0$ and hence $L \cdot D = 0$, and;

$\mathbf{EllLogCY}$: the set of all log Calabi--Yau surfaces $(Y,D)$, with a line bundle $L$, such that $|L|$ defines an elliptic fibration $Y \to \PP^1$, up to deformations preserving the elliptic fibration. These log Calabi--Yau surfaces can be explicitly classified. The fibers of $c$ away from $D$ have at worst $I_1$ singular fibers for general deformations. On the other hand, $D$ is a cycle of rational curves. Since $Y$ is a rational elliptic surface, 
if $Q$ denotes the number of $I_1$ fibers away from $D$, then $D$ is a cycle of 
$12-Q$ irreducible rational curves.  
One can show that we necessarily have $3 \leq Q \leq 11$. Moreover, for every such $3\leq Q\leq 11$ with $Q\neq 4$, there exists a unique deformation class of rational elliptic surfaces, and if $Q=4$, there exist exactly two classes of deformation of rational surfaces \cite{persson}. If $L=\mathcal{O}(kF)$, where $k\geq 1$ and $F$ a general fiber of $c$, the corresponding 5d SCFT is the rank $k$ version of the $E_{Q-3}$ theory if $n \neq 4$, or of the $E_1$ and $\widetilde{E}_1$ theories if $Q=4$ \cite{closset_rank_one, DeWolfe, uncovering, morrison_seiberg, seiberg1996five, yamada_yang}.

Define
\begin{align}
\label{eq FII0}
F_{II,0} : \mathbf{CWebs}_{II,0} & \longrightarrow \mathbf{EllLogCY} \\
\nonumber
W^{\mathrm{asym}} & \longmapsto (c: (Y,D,L) \to \PP^1)
\end{align}
where $c: (Y,D,L) \to \PP^1$ is the contraction map $c$, as discussed above.

\begin{theorem}
\label{thm: geo classification}
The maps $F_I, F_{II,+}$ and $F_{II,0}$ in equations \eqref{eq FI}, \eqref{eq FII+} and \eqref{eq FII0} respectively, are one-to-one correspondences. 
\end{theorem}

\begin{proof}
Injectivity follows as in the proof of Theorem \ref{thm_7_branes}.
Indeed, if two webs of 5-branes with 7-branes define two toric models of the same log Calabi--Yau surface with line bundle up to corner blow-ups and blow-downs, then, by \cite[Proposition 3.27]{HK_HMS} these two toric models are connected by a sequence of corner blow-ups/blow-downs, and so the corresponding webs are connected by a sequence of Hanany--Witten moves.

To prove the surjectivity of $F_I$ or $F_{II,+}$, let $(Y,D,L)$ be a smooth resolution of $(Y^{\mathrm{pol}}, D^{\mathrm{pol}}, L^{\mathrm{pol}}) 
\in \mathbf{PolLogCY}$ or $(Y^{\mathrm{pol}}, L^{\mathrm{pol}})   
\in \mathbf{PolCY^{cusp}}$. Modulo additional corner blow-ups, we can assume that $D$ has at least three irreducible components. Since $L^{\mathrm{pol}}$ is ample, the line bundle $L$ is big and nef, and so it follows from the proof of \cite[Theorem 5.4]{engel_friedman} that there exists a toric model $p: (Y,D) \rightarrow (\overline{Y}, \overline{D})$, with interior exceptional curves $E_i$, and an ample line bundle $\overline{L}$ on $\overline{Y}$ such that
$L=p^\star \overline{L} \otimes \mathcal{O}(-\sum_{i}a_iE_i)$
with $a_i \in \ZZ_{\geq 0}$. Let $W^{\mathrm{asym}}$ be the web of 5-branes with 7-branes obtained from the web of 5-branes associated to the polarized toric surface $(\overline{Y}, \overline{D}, \overline{L})$ by ending $a_i$ 5-branes on a $(p_i,q_i)$
7-brane, where $\RR_{\geq 0}(p_i,q_i)$ is the ray of $\overline{W}$
corresponding to the toric divisor of $\overline{Y}$ containing the point $p(E_i)$. Then, by construction, we have $F_I(W^{\mathrm{asym}})=(Y^{\mathrm{pol}}, D^{\mathrm{pol}}, L^{\mathrm{pol}})$ or $F_{II,+}(W^{\mathrm{asym}})=(Y^{\mathrm{pol}}, L^{\mathrm{pol}})$.
Finally, it remains to prove the surjectivity of $F_{II,0}$. For  $(Y,D,L) \in \mathbf{EllLogCY}$, we have $L=\mathcal{O}(kD)$ for some $k\in \ZZ_{\geq 1}$. Hence, the analog of \cite[Proposition 1.5]{engel_friedman} holds, and so one can run the argument in the proof of \cite[Theorem 5.4]{engel_friedman}, and conclude the surjectivity of $F_{II,0}$ as done above for $F_I$ and $F_{II,+}$.
\end{proof}

It follows from Theorem \ref{thm: geo classification} that the maps 
$F_I$, $F_{II,+}$, and $F_{II,0}$ define a one-to-one correspondence between the set of 5d SCFTs defined by consistent asymptotic webs of 5-branes with 7-branes and the set 
\[ \mathbf{PolLogCY} \cup \mathbf{PolCY^{cusp}} \cup \mathbf{EllLogCY}\,, \]
containing polarized log Calabi--Yau surfaces, polarized Calabi--Yau surfaces with a cusp singularity, and rational elliptic log Calabi--Yau surfaces. This generalizes the geometric classification of 5d SCFTs defined by webs of 5-branes in terms of polarized toric surfaces reviewed in \S\ref{webs_toric_surfaces}.

\subsection{Consistent webs of 5-branes with 7-branes and birational geometry of degenerations}
\label{section_consistent_webs}
Throughout this section we denote $W$ a general web of $5$-branes with $7$-branes, and  $W^{\mathrm{asym}}$ the associated asymptotic web, as in Definition \ref{def_asymptotic_5}. 

As explained in \S\ref{sec: degeneration} -\eqref{eq: degeneration}, there is a degeneration of log Calabi--Yau surfaces with line bundles
\begin{equation}\label{eq_degeneration}
\pi: (\mathcal{Y},\mathcal{D},\mathcal{L}) \longrightarrow \CC\end{equation}
associated to $W$, whose general fiber $(Y,D,L)$ is the log Calabi--Yau surface with line bundle associated to $W^{\mathrm{asym}}$ (see \S\ref{section_webs_log_CY}). The dual intersection complex of the central fiber $(\mathcal{Y}_0,\mathcal{D}_0)$ is the polyhedral decomposition on $\RR^2$ defined by the web of $5$-branes $\overline{W}$. In particular, for every vertex $v$ of $\overline{W}$, there is a corresponding irreducible component $Y_0^v$ of $\mathcal{Y}_0$. We denote by $\partial Y_0^v$ the union of the intersections with $Y_0^v$ of $\mathcal{D}_0$ and of the double locus of $\mathcal{Y}_0$. Moreover, for every $t\in \CC$, we denote $(\mathcal{Y}_t,\mathcal{D}_t,\mathcal{L}_t):=\pi^{-1}(t)$.
 
\begin{definition}
\label{def: W consistent}
A web of $5$-branes with $7$-branes $W$ is called \emph{consistent} if $W^{\mathrm{asym}}$ is consistent, and there exists a family of curves $(C_t)_{t\in \CC}$ in $\pi: (\mathcal{Y},\mathcal{D},\mathcal{L}) \longrightarrow \CC$ such that the following conditions hold:
\begin{itemize}
    \item[i)] For every $t \in \CC$, we have $C_t \in |\mathcal{L}_t|$.
    \item[ii)] For every vertex $v$ of $W$, the curve $C_0^v := C_0 \cap Y_0^v$ in $Y_0^v$ does not intersect the 0-dimensional strata of $\partial Y_0^v$.
\end{itemize}
\end{definition}

If $W$ is consistent,
then the line bundle $\mathcal{L}_t$ is nef for $t \neq 0$ 
by Lemma \ref{lem_L_nef} since $W^{\mathrm{asym}}$ is consistent.
However, the line bundle $\mathcal{L}_0$ on the central fiber $\mathcal{Y}_0$ might not be nef in general -- hence, generally the line bundle $\mathcal{L}$ on $\mathcal{Y}$ is not nef. However, we show below that after a sequence of ``M1-flops'', we can ensure that it is nef as well. To do this, we first recall the definition of an M1 flop applied to $\mathcal{Y}$.

Let $E$ be an interior $(-1)$-curve in an irreducible component $Y_0^v$ of $\mathcal{Y}_0$, corresponding to a vertex $v$ of $W$. Assume that 
the intersection point $p=E \cap \partial Y_0^v$ is not contained in the $\mathcal{D}_0$. Thus, there exists another component $Y_0^{v'}$ of $\mathcal{Y}_0$, corresponding to a vertex $v'$ of $W$, such that  $p \in Y_0^{v}  \cap Y_0^{v'}$. 
An M1 flop of $E$ applied to $\mathcal{Y}_0$ is the blow-down of $E$ in $Y_0^v$, followed by the blow-up of the point $p$ in $Y_0^{v'}$. We denote by $E'$ the exceptional curve in the blow-up of $Y_0^{v'}$. 

An M1 flop of $E$ applied to $\mathcal{Y}_0$ uniquely extends to a birational modification of $\mathcal{Y}$, which is the identity away from $\mathcal{Y}_0$. Indeed, the normal bundle of $E \simeq \PP^1$ in $\mathcal{Y}$ is $\mathcal{O}_{\PP^1}(-1) \oplus \mathcal{O}_{\PP^1}(-1)$ and the M1 flop on $\mathcal{Y}$ is the Atiyah flop corresponding to $E$ -- see \cite{atiyah1958analytic}. 
We refer to this Atiyah flop as the M1 flop of $E$ applied to $\mathcal{Y}$, following the terminology of \cite{alexeev2024ksba}.
Since an M1 flop is an isomorphism on codimension $2$, denoting by $\mathcal{Y}'$ the resulting 3-fold obtained from $\mathcal{Y}$ by an M1 flop, we have a natural 
isomorphism of the Picard groups $\mathrm{Pic}(\mathcal{Y}) \simeq  \mathrm{Pic}(\mathcal{Y}') $. Hence, the line bundle $\mathcal{L}$ corresponds to a line bundle $\mathcal{L}'$ on $\mathcal{Y}'$. A key feature of an M1 flop is that it changes the intersection numbers by:
\begin{equation}
\label{Eq: intersection numbers}
    \mathcal{L} \cdot E = -  \mathcal{L}' \cdot E' \,. 
\end{equation}

\begin{lemma} \label{lem_consistent_flop}
Let $W$ be a consistent web of $5$-branes with $7$-branes. Then, there exists a finite sequence of M1-flops such that the line bundle $\mathcal{L}$ on $\mathcal{Y}$ becomes nef.
\end{lemma}

\begin{proof}
Since $W$ is consistent, there exists a family of curves $(C_t)_t$ in $\mathcal{Y}$, such that the central fiber $C_0 \in |\mathcal{L}_0|$ in $\mathcal{Y}_0$ does not contain any irreducible component of $\mathcal{D}_0$ or of the double locus of $\mathcal{Y}_0$.
Assume that $\mathcal{L}$ is not nef, and let $E$ be an irreducible curve such that $\mathcal{L} \cdot E < 0$.  Since $C_0$ does not contain any irreducible component of $\mathcal{D}_0$ or of the double locus of $\mathcal{Y}_0$, $E$ must be an interior $(-1)$-curve by \cite[Lemma 1.4]{SB}. Applying an M1 flop of $E$ and denoting by $\mathcal{Y}'$ the resulting 3-fold with line bundle $\mathcal{L}'$ as before, we get a (-1)-curve $E'$ in $\mathcal{Y}'$, which by \eqref{Eq: intersection numbers} has positive intersection with $\mathcal{L}'$. The resulting curve $C_0'$ after the flop still does not contain any irreducible component of the double locus of $\mathcal{Y}_0'$ or of $\mathcal{D}_0'$.
On the other hand, denoting by $D_E$ the irreducible component of the double locus of $\mathcal{Y}_0$ intersecting $E$, it follows from \cite[Lemma 1.1]{SB} that 
\begin{equation} \label{eq_flop_decreases}
\mathcal{L}' \cdot D_E = \mathcal{L} \cdot D_E +\mathcal{L}\cdot E < \mathcal{L}\cdot D_E\,,\end{equation}
and so the total intersection number of $\mathcal{L}$
with the double locus of $\mathcal{Y}_0$ strictly decreases.

The iteration of all the M1-flops of all the curves obstructing the nefness of $\mathcal{L}$ terminates in finitely many steps. Indeed, since the curve $C_0$ never contains an irreducible component of the double locus of $\mathcal{Y}_0$, the total intersection number of $\mathcal{L}$ with this double locus is always nonnegative. On the other hand, this intersection number strictly decreases at each step by \eqref{eq_flop_decreases}.
\end{proof}

\begin{remark}
    If $(\mathcal{Y}',\mathcal{D}',\mathcal{L}')$ is obtained from $(\mathcal{Y},\mathcal{D},\mathcal{L})$ by a sequence of flops and $\mathcal{L}'$ is nef, then $(\mathcal{Y}',\mathcal{D}',\mathcal{L}')$ is a \emph{minimal model} of $(\mathcal{Y},\mathcal{D},\mathcal{L})$ in the sense of birational geometry. In particular, Lemma \ref{lem_consistent_flop} is a particular case of more general results on the existence of minimal models in the Minimal Model Program. Our proof of Lemma \ref{lem_consistent_flop} is a variation on the proof in \cite{SB} of a similar result
    in the context of degenerations of K3 surfaces.
\end{remark}

\subsection{Flops and pushing 7-branes}
\label{sec pushing}
We introduce modifications of webs of $5$-branes with $7$-branes, obtained by pushing $7$-branes along their monodromy invariant directions. We then discuss their geometric interpretation in terms of flops of degenerations of log Calabi--Yau surfaces.

We first define the operation ``pushing 7-branes'' in an asymptotic web of 5-branes with 7-branes.
Let $W^{\mathrm{asym}}$ be an asymptotic web of 5-branes with 7-branes. Consider a $(p_i,q_i)$ 7-brane of position $x_{e,i} \in \RR_{>0}(p_i,q_i)$ in $W^{\mathrm{asym}}$. Then, as in the description of Hanany--Witten moves in \S\ref{section_consistent}, we can move the 7-brane along its monodromy invariant direction until it crosses the origin. We denote by $x_{e,i}' \in \RR_{<0}(p_i,q_i)$ the new position of the 7-brane. Note that the monodromy cut from $x_{e,i}'$ covers the union of the line segment connecting $x_{e,i}'$ to the origin and of the ray $\RR_{\geq 0}(p_i,q_i)$. 
Moreover, we modify the web of 5-branes at follows. Denote by $w_{-e}$ the weight of $\RR_{<0}(p_i,q_i)$ in $W^{\mathrm{asym}}$
if $\RR_{<0}(p_i,q_i)$ is an edge of $W^{\mathrm{asym}}$, and set $w_{-e}=0$ else. Then, after moving the 7-brane across the origin, the multiplicity on the half-line $x_{e,i}'+\RR_{<0}(p_i,q_i)$ remains $w_{-e}$. 
On the other hand, imposing the balancing condition at the origin fixes the weight $\mu$ of the line segment connecting $x_{e,i}'$ to the origin. We could either have:
\begin{itemize}
    \item[i)] \textbf{Case I}: We have
    \begin{equation}
\label{Eq: case 01}
\mu <0   \,.     
    \end{equation}
In this situation, the web of 5-branes with 7-branes $W^{\mathrm{asym}}$ is not supersymmetric and we say that pushing the 7-brane $x_{e,i}$ across the origin is \emph{forbidden}.
     \item[ii)] \textbf{Case II}: We have
   \begin{equation}
\label{Eq: case 02}
0 \leq \mu \leq w_{-e} \,.
    \end{equation}   
In this situation, one can interpret the 7-brane $x_{e,i}'$ as being the endpoint of $w_{-e}-\mu$ 5-branes on the half-line $x_{e,i}' +\RR_{<0}(p_i,q_i)$. These 5-branes are not connected to the part of the web of 5-branes containing the origin, and so can be freely moved  together with the 7-brane $x_{e,i}'$ on the half-line $\RR_{<0}(p_i,q_i)$. 
      \item[iii)] \textbf{Case III}: We have:
 \begin{equation}
\label{Eq: case 03}
\mu >w_{-e}\,.
    \end{equation} 
In this situation, one can interpret the 7-brane $x_{e,i}'$ as being the endpoints of $\mu -w_{-e}$ 5-branes on the line segment connecting $x_{e,i}'$ to the origin. 
\end{itemize}
In both Case II and Case III, we say that pushing the 7-brane $x_{e,i}$ across the origin is \emph{allowed}, and we call the resulting configuration of 5-branes and 7-branes a \emph{modified web of 5-branes with 7-branes} obtained from $W^{\mathrm{asym}}$ by pushing the 7-brane $x_{e,i}$ across the origin.

We now define the operation of pushing 7-branes in a web of 5-branes with 7-branes.  Let $W$ be a web of 5-branes with 7-branes. 
For every vertex $v$ of $W$, we define a \emph{local asymptotic web of 5-branes with 7-branes} $W_v$ centered at $v$ by keeping every edge of $W$ adjacent to $v$ and every 7-brane where some 5-branes adjacent to $v$ are ending. 

\begin{definition}
A \emph{modified web of 5-branes with 7-branes} is a configuration of 5-branes with 7-branes obtained from a web of 5-branes with 7-branes $W$ by a finite sequence of pushing operations of 7-branes in the asymptotic webs of 5-branes with 7-branes $W_v$ centered at the vertices $v$ of $W$.
\end{definition}

We now discuss the geometric interpretation of 
modified webs of 5-branes with 7-branes. 
Let $\pi: (\mathcal{Y}, \mathcal{D}, \mathcal{L}) \rightarrow \CC$ be the degeneration of log Calabi--Yau surfaces associated to a web of 5-branes with 7-branes $W$. The irreducible components $Y_0^v$ of the central fiber $\mathcal{Y}_0=\pi^{-1}(0)$ are in one-to-one correspondence with the vertices $v$ of $W$.
Let $\partial Y_0^v$ be the intersection with $Y_0^v$ of the union of the double locus of $\mathcal{Y}$ and of the divisor $\mathcal{D}$, and let $L_0^v$ the restriction of $\mathcal{L}$ to $Y_0^v$. Then, 
$(Y_0^v, \partial Y_0^v, L_0^v)$ is the log Calabi--Yau surface with line bundle associated to the local asymptotic web of 5-branes with 7-branes $W_v$ 
centered at $v$ as in \S \ref{section_webs_log_CY}. Let $W'$ be a modified web of 5-branes with 7-branes obtained from $W$ by pushing operations of 7-branes. Then, $W'$ corresponds to a degeneration of log Calabi--Yau surfaces 
\[ \pi': (\mathcal{Y}',\mathcal{D}',\mathcal{L}') \longrightarrow \CC\,,\]
obtained from  $\pi: (\mathcal{Y}, \mathcal{D}, \mathcal{L}) \rightarrow \CC$ by 
a sequence of M1 flops. Indeed, pushing a 7-brane $x_{e,i}$ across $v$ in $W_v$ as in Case II-\eqref{Eq: case 02} with $w_{-e} \geq 1$ 
amounts to flopping the $(-1)$-curve $E_{e,i}'$ as in \S\ref{section_consistent}. On the other hand, pushing a 7-brane $x_{e,i}$ across $v$ in $W_v$ as in 
Case II-\eqref{Eq: case 02} with $w_{-e}=\mu=0$, or 
as in
Case III-\eqref{Eq: case 03} does not change the degeneration, but only the toric model of the log Calabi--Yau surface $(Y_0^v, \partial Y_0^v, L_0^v)$, as in the geometric description of Hanany--Witten moves in  \S\ref{section_consistent}.

\begin{lemma} \label{lem: consistent_pushed}
Let $W$ be a consistent web of 5-branes with 7-branes. Then, there exists a modified web $W'$ of 5-branes with 7-branes obtained from $W$ by pushing operations of 7-branes, such that all local asymptotic webs of 5-branes with 7-branes $W_v'$ centered at the vertices $v$ of $W'$ are consistent.
\end{lemma}

\begin{proof}
Since $W$ is consistent, Lemma \ref{lem_consistent_flop} implies the existence of a degeneration of log Calabi--Yau surfaces 
$\pi': (\mathcal{Y}',\mathcal{D}',\mathcal{L}') \rightarrow \CC$,
obtained from  $\pi: (\mathcal{Y}, \mathcal{D}, \mathcal{L}) \rightarrow \CC$ by a sequence of M1 flops, such that all line bundles $(L')_0^v$ are nef.
This sequence of M1 flops can be realized through successive pushing operations of 7-branes, 
yielding a modified web $W'$. Indeed, as in the proof of 
Lemma \ref{lem_caseIII_consistent}, by \cite[Proposition 3.27]{HK_HMS} (see also \cite[Theorem 1]{Blanc}), any two toric models of a log Calabi--Yau surface are related 
by a sequence of corner blow-ups/blow-downs and elementary cluster transformations.
Consequently, any interior $(-1)$-curve can be obtained, after a finite sequence of pushing operations of 7-branes, as the 
$(-1)$-curve $E_{e,i}'$ corresponding to a pushed 7-brane $x_{e,i}'$.
\end{proof}

\subsection{Pushing 7-branes: local examples}
\label{sec_examples}

In this section, we describe two key local examples of a pushing operation of a 7-brane, through a 3-valent and a 4-valent vertex respectively.

\subsubsection{Pushing a 7-brane through a 3-valent vertex}
\label{sec_3_valent}

Let $\overline{W}$ be a web of 5-branes consisting of a single 3-valent vertex at the origin. Denote by $\overline{v}_1$, 
$\overline{v}_2$, $\overline{v}_3$ the primitive directions of the 3-edges pointing away from the origin. 
Let $w_1, w_2, w_3$
be the weights of the edges, and denote $v_1=w_1 \overline{v}_1$, $v_2=w_2 \overline{v}_2$, $v_3=w_3 \overline{v}_3$. By the balancing condition, we have $v_1+v_2+v_3=0$. Finally, denote by $m$ the multiplicity of the 3-valent vertex, that is,
\[ m=|\det(v_1,v_2)|=|\det(v_2,v_3)|=|\det(v_3,v_1)|\,.\]
Let $W$ be the web of 5-branes with 7-branes obtained from $\overline{W}$ by requiring that all the $w_3$ 5-branes of direction $v_3$ end on the same 7-brane, see Figure \ref{Fig:3_valent_web}.

\begin{figure}[hbt!]
\center{\scalebox{.9}{\includegraphics{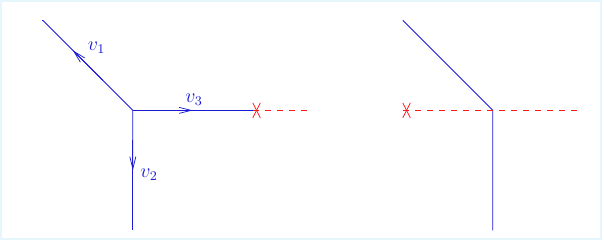}}}
\caption{Before and after pushing a 7-brane through a 3-valent vertex.}
\label{Fig:3_valent_web}
\end{figure}

After moving the 7-brane across the origin, the new balancing condition at the origin is $v_1'+v_2'+v_3'=0$, where $v_2'$ is obtained from $v_2$ by crossing the monodromy cut:
\[ v_2'=v_2+\mathrm{det}(v_2, \overline{v}_3) \overline{v}_3\,,\]
and so $v_3'$ is given by 
\[ v_3' = v_3 -\mathrm{det}(v_2, \overline{v}_3) \overline{v}_3 = \left(w_3 - \frac{m}{w_3}\right) \overline{v}_3 = 
(w_3^2 - m) \frac{\overline{v}_3}{w_3} \,.\]
It follows that we are in Cases I-II-III as in \eqref{Eq: case 01}-\eqref{Eq: case 02}-\eqref{Eq: case 03} if and only if $w_3^2>m$, $w_3^2=m$, and $w_3^2 <m$ respectively. In particular, we are in Case II if and only if $w_3^2=m$, in which case the 7-brane can be pushed through the origin and is no longer attached to any 5-brane, as in Figure \ref{Fig:3_valent_web}.
The condition $w_3^2=m$ can be rephrased in terms of the lattice triangle $\overline{P}$ dual to $\overline{W}$. First note that the weights $w_1, w_2, w_3$ are the lattice lengths of the edges $e_1, e_2, e_3$ of $\overline{P}$
corresponding to the edges of $\overline{W}$. Denoting by $h$ the \emph{height} of $\overline{P}$ with respect to $e_3$ the lattice distance between $e_3$ and the opposite vertex of $\overline{P}$, we have $m=w_3 h$. It follows that the condition $w_3^2=m$ is equivalent to $w_3=h$.

We now describe the geometric interpretation of pushing the 7-brane through the origin $w_3=h$. 
Let $(\overline{Y}, \overline{D}, \overline{L})$ be the polarized toric surface corresponding to $\overline{W}$.
The toric boundary divisor $\overline{D}$ has three irreducible toric divisors $\overline{D}_1,\overline{D}_2, \overline{D}_3$ corresponding to $e_1, e_2, e_3$, and the log Calabi--Yau surface $Y$ is obtained from $\overline{Y}$ by blowing up a smooth point $x_3$ of the toric boundary on $\overline{D}_3$, see Figure \ref{Fig:contraction_3_valent}. 
Moreover, denoting by $E$ the exceptional divisor, we have $L=\overline{L} \otimes \mathcal{O}(-w_3 E)$. Denote by $H$ the image in $Y$ of the strict transform $E'$ of the fiber passing through $x_3$ of the $\PP^1$-fibration on the toric blow-up of $\overline{Y}$ obtained by adding the ray $\RR_{<0}v_3$ to the fan. One can check that $L \cdot H =h-w_3$. In particular, when $w_3=h$, we have $L \cdot H=0$ and so $L$ is nef. On the other hand, we have $H^2=0$ since $(E')^2=-1$ and the toric divisor corresponding to $\RR_{<0}v_3$ is contracted in $Y$. It follows that there is a fibration of $Y$ with fibers of class $H$. Finally, using that $L \cdot E=w_3>0$, one concludes that the map contracting all curves in $Y$ having zero-intersection with $L$ is a map $c: Y \rightarrow \PP^1$. The fact that the surface $Y$ is contracted onto $\PP^1$ is the geometric realization of the fact that, after pushing the 7-brane across the origin, the 3-valent vertex of the web disappears and the remaining 5-branes form a straight line (taking into account the monodromy cut), which can be viewed as the fan of $\PP^1$.

\begin{figure}[hbt!]
\center{\scalebox{.9}{\includegraphics{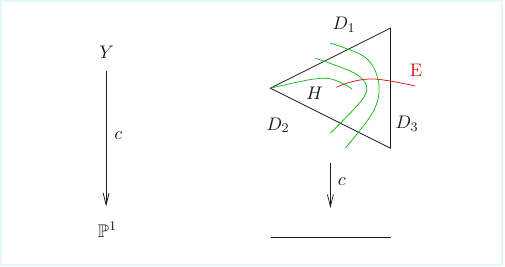}}}
\caption{The contraction map $c:Y \rightarrow \PP^1$.}
\label{Fig:contraction_3_valent}
\end{figure}

For example, when $\overline{P}$ is the lattice triangle of size one, then $Y$ is the blow-up of $\PP^2$ at one point, the one-parameter family of curves of class $H$ is the pencil of strict transform of lines in $\PP^2$ passing through the blown-up point, that is, the natural $\PP^1$-fibration on $Y \simeq \mathbb{F}_1$. The contraction $c: Y \simeq \mathbb{F}_1 \rightarrow \PP^1$ is then the contraction onto the base of this fibration.

For later purposes, we introduce the notion of ``elementary" triangle, by imposing not only the condition $w_3=h$, but also the condition $w_1=w_2=1$.

\begin{definition} \label{def_elementary_triangle}
    A lattice triangle $\overline{P}$ is called \emph{elementary} if there exists an edge $e_3$ of $P$ whose lattice length is equal to the height of $P$ with respect to $e_3$, and such that the remaining edges $e_1$ and $e_2$ both have lattice length one.
\end{definition}

\begin{figure}[hbt!]
\center{\scalebox{.9}{\includegraphics{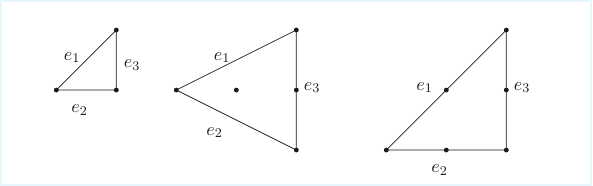}}}
\caption{Examples of lattice triangles with $w_3=h$.}
\label{Fig:basic_triangles}
\end{figure}

\begin{remark}
    If $w_3=h$, then we automatically have $w_1=w_2$. This results from the fact that the web is balanced after pushing the 
    7-brane (see also \cite[Proposition 2.12]{alexeev2024ksba} for a direct proof). It follows that, in the Definition \ref{def_elementary_triangle} of an elementary triangle, it is enough to require that one of the two remaining edges has lattice length one.
\end{remark}

\begin{example}
The left and middle lattice triangles in Figure \ref{Fig:basic_triangles} are elementary. The right triangle in Figure \ref{Fig:basic_triangles} satisfies $w_3=h$ but is not elementary.
\end{example}

\subsubsection{Pushing a 7-brane through a 4-valent vertex}

Let $\overline{W}$ be a web of 5-branes consisting of a single 4-valent vertex at the origin with edges of primitive 
directions of the form $\overline{v}_1$, $\overline{v}_2$, 
$\overline{v}_3$, $\overline{v}_4=-\overline{v}_3$. Let $w_1, w_2, w_3, w_4$ be the weights of the edges, and denote $v_i=w_i \overline{v}_i$ for $1 \leq i \leq 4$. By the balancing condition, we have $\sum_{i=1}^4 v_i=0$.
Let $W$ be the web of 5-branes with 7-branes obtained from $\overline{W}$ by requiring that all the 5-branes in direction $v_3$ all end on the same 7-brane, see Figure \ref{Fig:pushing_4_valent}.

\begin{figure}[hbt!]
\center{\scalebox{.9}{\includegraphics{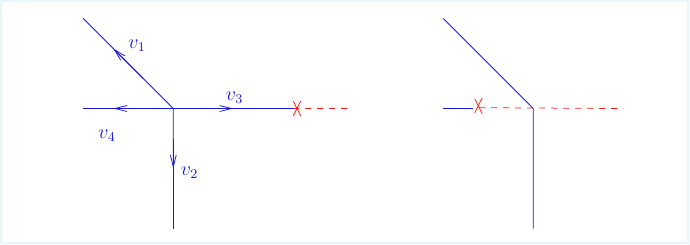}}}
\caption{Before and after pushing a 7-brane through a 4-valent vertex.}
\label{Fig:pushing_4_valent}
\end{figure}

Arguing as in \S \ref{sec_3_valent}, we obtain that, after moving the 7-brane through the origin, the number of 5-branes starting at the origin and ending on the 7-brane is 
$|\det(v_2, \overline{v}_3)|+w_4-w_3$. Note that, by the balancing condition, we have $|\det(v_2, \overline{v}_3)|=|\det(v_1, \overline{v}_3)|$.
Hence, we are in Cases I-II-III as in \eqref{Eq: case 01}-\eqref{Eq: case 02}-\eqref{Eq: case 03} if and only if $w_3-w_4>|\det(v_2, \overline{v}_3)|$, $w_3-w_4=|\det(v_2, \overline{v}_3)|$, and $w_3 -w_4<|\det(v_2, \overline{v}_3)|$ respectively. In particular, we are in Case II if and only if $w_3-w_4=|\det(v_2, \overline{v}_3)|$, in which case the 7-brane can be pushed through the origin and is no longer attached to 5-branes connected to the origin -- see Figure \ref{Fig:pushing_4_valent}.
The condition $w_3-w_4=|\det(v_2, \overline{v}_3)|$ can be rephrased in terms of the lattice trapezoid $\overline{P}$ dual to $\overline{W}$. First note that the weights $w_1, w_2, w_3, w_4$ are the weights of the edges $e_1, e_2, e_3$ of $\overline{P}$
corresponding to the edges of $\overline{W}$.
Then, we have $|\det(v_2, \overline{v}_3)|=|\det(v_1, \overline{v}_3)|=h$, where $h$ is the height of $\overline{P}$ with respect to $e_3$, that is, the lattice distance between $e_3$ and $e_4$.
It follows that the condition  $w_3-w_4=|\det(v_2, \overline{v}_3)|$ is equivalent to $w_3-w_4=h$.

Finally, we describe the geometric interpretation of pushing the 7-brane through the origin when $w_3-w_4=h$. 
Let $(\overline{Y}, \overline{D}, \overline{L})$ be the polarized toric surface corresponding to $\overline{W}$.
The toric boundary divisor $\overline{D}$ has four irreducible toric divisors $\overline{D}_1,\overline{D}_2, \overline{D}_3, \overline{D}_4$ corresponding to $e_1, e_2, e_3, e_4$, and the log Calabi--Yau surface $Y$ is obtained from $\overline{Y}$ by blowing up a smooth point $x_3$ of the toric boundary on $\overline{D}_3$, see Figure \ref{Fig:Y_4_valent}. Moreover, denoting by $E$ the exceptional divisor, we have $L=\overline{L} \otimes \mathcal{O}(-w_3 E)$. Denote by $E'$ the strict transform of the fiber passing through $x_3$ of the $\PP^1$-fibration on $\overline{Y}$ induced by the fact that both $\RR_{\geq 0}v_3$ and $\RR_{\leq 0}v_3$ are rays of the fan of $\overline{Y}$. One can check that $L \cdot H =h-w_3$. 
In particular, when $w_3-w_4=h$, we have $L \cdot H = -w_4$ and $L$ is not nef.
When $W$ is part of a bigger consistent web, one can flop $E'$, and the surface $Y$ becomes $Y'=\mathbb{F}_0=\PP^1 \times \PP^1$. In addition, the resulting line bundle $L'$ satisfies $L' \cdot D_3'=L'\cdot D_4'=0$ and $L'\cdot D_1'=L'\cdot D_2'=h$, and so is nef.
Contracting all curves with zero-intersection with $L'$
produces a map $c: Y'=\PP^1 \times \PP^1 \rightarrow \PP^1$.
 The fact that the surface $Y'$ is contracted onto $\PP^1$ is the geometric realization of the fact that, after pushing the 7-brane across the origin, the 4-valent vertex at the origin of the web disappears and the remaining 5-branes form a straight line (taking into account the monodromy cut), which can be viewed as the fan of $\PP^1$.

\begin{figure}[hbt!]
\center{\scalebox{.9}{\includegraphics{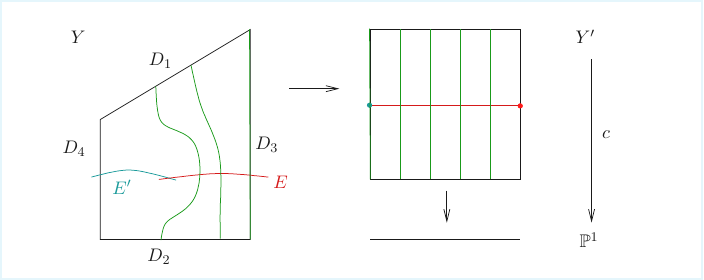}}}
\caption{The log Calabi--Yau surface $(Y,D)$ and the contraction map $c: Y' \rightarrow \PP^1$.}
\label{Fig:Y_4_valent}
\end{figure}

For later purposes, we introduce the notion of
``elementary" lattice trapezoid, by imposing the condition $w_3-w_4=h$, but also the condition $w_1=w_2=1$.

\begin{definition} \label{def_elementary_trapezoid}
A lattice trapezoid $\overline{P}$ is called \emph{elementary} if there exist parallel edges $e$ and $e'$ of $P$ such that the difference $|e|-|e'|$ between the lattice lengths of $e$ and $e'$ is equal to the height of $P$ with respect to $e$, and such that the remaining edges both have lattice length one.
\end{definition}

\begin{figure}[hbt!]
\center{\scalebox{.9}{\includegraphics{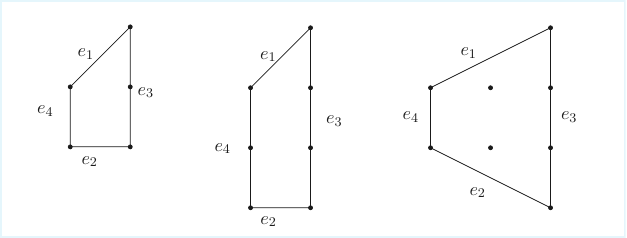}}}
\caption{Examples of elementary lattice trapezoids.}
\label{Fig:basic_trapezoids}
\end{figure}

\begin{example}
The three lattice trapezoids in Figure \ref{Fig:basic_trapezoids} are elementary.
\end{example}

\begin{remark} \label{remark_elementary_subdivision}
    If $w_3-w_4=h$, then we have automatically $w_1=w_2$. This results from the fact that the web is balanced after pushing the 
    7-brane (see also \cite[Proposition 2.12]{alexeev2024ksba} for a direct proof). It follows that, in the Definition \ref{def_elementary_trapezoid} of an elementary trapezoid, it is enough to require that one of the two remaining edges has lattice length one. It also follows that any lattice triangle with $w_3=h$, or any lattice trapezoid with $w_3-w_4=h$, can be naturally subdivided into a union of elementary triangles and elementary trapezoids, see Figure \ref{Fig:elementary_subdivision} for an example. 
\end{remark}

\begin{figure}[hbt!]
\center{\scalebox{.9}{\includegraphics{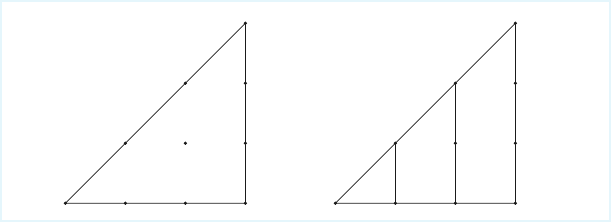}}}
\caption{On the left, a non-elementary triangle. On the right, a subdivision into one elementary triangle and two elementary trapezoids.}
\label{Fig:elementary_subdivision}
\end{figure}

\subsection{Generic consistent webs and fully pushed 7-branes}
\label{section_fully_pushed}

In this section, we introduce the notion of a generic consistent web of 5-branes with 7-branes, and describe generic consistent webs in terms of modified webs with ``fully pushed" 7-branes.

There is a natural notion of perturbation for webs of 5-branes $\overline{W}$ 
(see \cite[Definition 2.17]{mikhalkin}): we say that $\overline{W}'$ is a perturbation of $\overline{W}$ if there exists a one-parameter family of webs of 5-branes $(\overline{W}'_t)_{1\geq t>0}$ with the same combinatorial type as $\overline{W}'$, such that $\overline{W}'_{1}=\overline{W}'$, and $\overline{W} =\mathrm{lim}_{t \rightarrow 0} \overline{W}'_t$, where the limit is taken in the sense of weighted graphs, that is, with weights adding up when edges coincide in the limit. We generalize this notion of perturbation for webs of 5-branes with 7-branes:

\begin{definition}
Let $W$ and $W'$ be two webs of 5-branes with 7-branes defined by webs of 5-branes $\overline{W}$ and $\overline{W}'$, and 7-brane data $\mathbf{a} := ((a_{e,i})_{1\leq i \leq n_e})_{e \in L(\overline{W})}$ and $\mathbf{a}' := ((a'_{e',i'})_{1\leq i' \leq n_{e'}})_{e' \in L(\overline{W}')}$ respectively.
We say that $W'$ is a perturbation of $W$ if $\overline{W}'$ is a perturbation of $\overline{W}$, and for every leg $e \in L(\overline{W})$, there exists a bijection $i \mapsto i'$ between 
$\{1,\cdots,n_e\}$ and $\sqcup_{e' \in L_e} \{1,\cdots, n_{e'}\}$ such that $a_{e,i}=a_{e',i'}$, where $L_e$ is the set of legs $e' \in L(\overline{W}')$ obtained by perturbation of $e$.
\end{definition}

\begin{definition}
\label{def: generic}
    A consistent web of $5$-branes with $7$-branes $W$ is called \emph{generic} if no perturbation of $W$ with a different combinatorial type is consistent.
\end{definition}

\begin{example}
A web of 5-branes $\overline{W}$ is generic if and only if the polyhedral decomposition to $\overline{W}$
is a triangulation into triangles of size one, that is, if $\overline{W}$ is 3-valent and every vertex of $\overline{W}$
is of multiplicity one (see \cite[Proposition 2.19]{mikhalkin}).
\end{example}

\begin{example} \label{ex_basic_generic}
There exist generic webs of 5-branes with 7-branes with 3-valent vertices of multiplicity $>1$ or 4-valent vertices. 
For instance, 3-valent vertices dual to elementary triangles and 4-valent vertices dual to elementary trapezoids as in \S \ref{sec_examples} are generic.
Indeed, for a 3-valent vertex dual to an elementary triangle, using the notation of \ref{sec_3_valent}, the vertex adjacent to the $w_3$ 5-branes ending on the 7-brane in a non-trivial perturbation of the web is necessarily 3-valent with multiplicity $m'<m=w_3$, and so would not be consistent.

A 3-valent (resp. 4-valent) vertex, whose dual triangle (resp. trapezoid) satisfies $w_3=h$ (resp. $w_3-w_4=h$) but is non-elementary, is non-generic. Indeed, in this case, the dual polygon admits by Remark \ref{remark_elementary_subdivision} a non-trivial subdivision into elementary triangles and trapezoids, which induce a non-trivial perturbation.
We will show in Theorem \ref{thm_generic_vertices} that every vertex of a generic web of 5-branes with 7-branes is in fact dual to either an elementary triangle or an elementary trapezoid.
\end{example}

Before proving the main results of this section, we provide a geometric interpretation of generic consistent webs of 5-branes, in terms of degenerations of log Calabi--Yau surfaces. 
Let $W$ be a consistent web of 5-branes with 7-branes, and $\pi: (\mathcal{Y},\mathcal{D},\mathcal{L}) \rightarrow \CC$ the corresponding degeneration of log Calabi--Yau surfaces with general fiber $(Y,D,L)$.
By Lemma \ref{lem_consistent_flop}, there exists a degeneration of log Calabi--Yau surfaces 
$\pi': (\mathcal{Y}',\mathcal{D}',\mathcal{L}') \rightarrow \CC$,
obtained from  $\pi: (\mathcal{Y}, \mathcal{D}, \mathcal{L}) \rightarrow \CC$ by a sequence of M1 flops, such that $\mathcal{L}'$ is nef. By the base point free theorem, there exists a morphism
$c: (\mathcal{Y}',\mathcal{D}',\mathcal{L}') \rightarrow ((\mathcal{Y}')^{\mathrm{pol}},(\mathcal{D}')^{\mathrm{pol}},(\mathcal{L}')^{\mathrm{pol}})$
contracting all curves $C'$ in $\mathcal{Y}'$ such that 
$\mathcal{L}' \cdot C'=0$, so that the line bundle $(\mathcal{L}')^{\mathrm{pol}}$ is ample. The following result describes the central fiber $((\mathcal{Y}')^{\mathrm{pol}}_0,(\mathcal{D}')^{\mathrm{pol}}_0,(\mathcal{L}')^{\mathrm{pol}}_0)$ of 
$((\mathcal{Y}')^{\mathrm{pol}},(\mathcal{D}')^{\mathrm{pol}},(\mathcal{L}')^{\mathrm{pol}})$ when $W$ is generic.

\begin{lemma} \label{lem: generic_central_fiber}
Let $W$ be a generic consistent web of 5-branes with 7-branes. Then, if $L^2>0$, the irreducible components of the central fiber $((\mathcal{Y}')^{\mathrm{pol}}_0,(\mathcal{D}')^{\mathrm{pol}}_0,(\mathcal{L}')^{\mathrm{pol}}_0)$ are all isomorphic to $(\PP^2, \partial \PP^2, \mathcal{O}(1))$, where $\partial \PP^2$ is the toric boundary divisor of $\PP^2$. If $L^2=0$, then the irreducible components of the central fiber $((\mathcal{Y}')^{\mathrm{pol}}_0,(\mathcal{D}')^{\mathrm{pol}}_0,(\mathcal{L}')^{\mathrm{pol}}_0)$ are all isomorphic to $(\PP^1, \partial \PP^1, \mathcal{O}(1))$, where $\partial \PP^1$ is the toric boundary divisor of $\PP^1$.
\end{lemma}

\begin{proof}
We first prove the result for $L^2>0$. An irreducible component of $((\mathcal{Y}')^{\mathrm{pol}}_0,(\mathcal{D}')^{\mathrm{pol}}_0,(\mathcal{L}')^{\mathrm{pol}}_0)$ is a polarized log Calabi--Yau surface $(Y_i, D_i, L_i)$. If $(Y_i, D_i, L_i)$ is not isomorphic to $(\PP^2, \partial \PP^2, \mathcal{O}(1))$, then, 
by \cite[\S 2.4]{alexeev2024ksba}, there exists a non-trivial degeneration of $(Y_i, D_i, L_i)$ to a union of copies of $(\PP^2, \partial \PP^2, \mathcal{O}(1))$.
The existence of such a degeneration 
implies that the web $W$ admits a non-trivial consistent perturbation, contradicting the assumption that $W$ is generic. The case where $L^2=0$ can be treated similarly.
\end{proof}

In order to state our main criterion for generic consistency of webs, we introduce the notion of  ``fully pushed" 7-branes:

\begin{definition}
\label{def_fully_pushed}
Let $W$ be a web of 5-branes with 7-branes and $W'$ a modified web of 5-branes with 7-branes obtained from $W$ by pushing operations of 7-branes, as explained in \S \ref{sec pushing}.
We say that 7-branes of $W'$ are  \emph{fully pushed} if the following conditions hold:
\begin{itemize}
    \item[i)] All the 7-branes are disjoint from the web $W'$, that is, no 5-brane of $W'$ is ending on a 7-brane. 
    \item[ii)] For every $(p_i,q_i)$ 7-brane $x_i$, the half-line 
    $x_i +\RR_{<0}(p_i,q_i)$ intersects the web $W'$, that is, the 7-brane $x_i$ cannot be pushed further without crossing $W'$.
\end{itemize}
\end{definition}

We can now state the main result of this section:

\begin{theorem}
\label{thm: generic push in}
Let $W$ be a web of $5$-branes with $7$-branes. Then, $W$ is a generic consistent web if and only if there exists a modified web of 5-branes with 7-branes $W'$, obtained from $W$ by a finite sequence of pushing operation of 7-branes, such that the following conditions hold:
\begin{itemize}
\item[i)]
The 7-branes of $W'$ are fully pushed, as in Definition \ref{def_fully_pushed}.
\item[ii)] All the vertices of $W'$ are 3-valent of multiplicity one.
\end{itemize}
\end{theorem}

\begin{proof}
Let $W$ be a generic consistent web of 5-branes with 7-branes. By Lemma \ref{lem: consistent_pushed}, since $W$ is consistent, there exists a modified web of 5-branes with 7-branes $W'$, obtained from $W$ by a finite sequence of pushing operation of 7-branes, such that the corresponding line bundle $\mathcal{L}'$ on the degeneration of log Calabi--Yau surfaces $(\mathcal{Y}',\mathcal{D}',\mathcal{L}') \rightarrow \CC$ is nef. By the base point free theorem, there exists a morphism
$c: (\mathcal{Y}',\mathcal{D}',\mathcal{L}') \rightarrow ((\mathcal{Y}')^{\mathrm{pol}},(\mathcal{D}')^{\mathrm{pol}},(\mathcal{L}')^{\mathrm{pol}})$
contracting all curves $C'$ in $\mathcal{Y}'$ such that 
$\mathcal{L}' \cdot C'=0$, and so that the line bundle $(\mathcal{L}')^{\mathrm{pol}}$ is ample. The polyhedral decomposition defined by the web $W'$
is the dual intersection complex of the central fiber
$((\mathcal{Y}')^{\mathrm{pol}}_0,(\mathcal{D}')^{\mathrm{pol}}_0,(\mathcal{L}')^{\mathrm{pol}}_0)$.
Since $W'$ is generic, the irreducible components of
$((\mathcal{Y}')^{\mathrm{pol}}_0,(\mathcal{D}')^{\mathrm{pol}}_0,(\mathcal{L}')^{\mathrm{pol}}_0)$ are described by Lemma \ref{lem: generic_central_fiber}, and it follows that the 7-branes are disjoint from the web $W'$, and that the vertices of $W'$ are 3-valent of multiplicity one. To end the proof that the 7-branes are fully pushed, it remains to check condition ii) in Definition \ref{def_fully_pushed}. 
If this condition were not satisfied, there would exist a $(p_i,q_i)$ 7-brane $x_i$ such that the half-line $x_i +\RR_{<0}(p_i,q_i)$ does not intersect $W'$. Then, there would exist an interior $(-1)$-curve $E$ in a corner blow-up of $Y$ whose tropicalization is $x_i +\RR_{<0}(p_i,q_i)$, and such that $L \cdot E=0$. But since $W$ is a consistent web, we have $L \cdot E>0$ by Definition \ref{def: consistency} ii), which is a contradiction.

Conversely, assume that $W$ is a web such that there exists a modified web $W'$ as in Theorem \ref{thm: generic push in}. Then, there exists a degeneration of log Calabi--Yau surfaces $(\mathcal{Y}',\mathcal{D}',\mathcal{L}') \rightarrow \CC$, obtained from $(\mathcal{Y},\mathcal{D},\mathcal{L}) \rightarrow \CC$ by a sequence of M1-flops, such that there exists a contraction $c: (\mathcal{Y}',\mathcal{D}',\mathcal{L}') \rightarrow ((\mathcal{Y}')^{\mathrm{pol}},(\mathcal{D}')^{\mathrm{pol}},(\mathcal{L}')^{\mathrm{pol}})$, such that the irreducible components of the central fiber
$((\mathcal{Y}')^{\mathrm{pol}}_0,(\mathcal{D}')^{\mathrm{pol}}_0,(\mathcal{L}')^{\mathrm{pol}}_0)$ are all isomorphic to $(\PP^2,\partial \PP^2, \mathcal{O}(1))$. There exists a curve $C_0$ in $(\mathcal{Y}')^{\mathrm{pol}}_0$, with dual intersection graph $W'$, and whose irreducible components are lines in $(\PP^2,\partial \PP^2, \mathcal{O}(1))$ not contained in $\partial \PP^2$.
By log smooth deformation theory as in \cite{NS}, the curve $C_0$ deforms into a family of curves in $(\mathcal{Y}')^{\mathrm{pol}}$, whose generic fiber $C$ is a smooth curve in $Y$ contained in the linear system $L$ and which does not contain any irreducible component of $D$. This shows that 
conditions i) and ii) in the Definition \ref{def: W consistent} of the consistency of $W$, and that the first condition i) in the Definition \ref{def: consistency} of the consistency 
$W^{\mathrm{asym}}$ are satisfied.
The second condition ii) in Definition \ref{def: consistency}
follows from the assumption that the 7-branes of $W'$ are fully pushed and Definition \ref{def_fully_pushed} ii). This concludes the proof that $W^{\mathrm{asym}}$ and $W$ are consistent.
Finally, since all vertices of $W'$ are 3-valent of multiplicity one, $W'$ does not admit any non-trivial perturbation, and so $W$ is generic as in Definition \ref{def: generic}.
\end{proof}

Recall from Example \ref{ex_basic_generic} that 3-valent vertices dual to elementary triangles and 4-valent vertices dual to elementary trapezoids as in Definitions 
\ref{def_elementary_triangle}-\ref{def_elementary_trapezoid} 
are generic vertices. We show below that all the vertices of a generic consistent web are actually of this form.

\begin{theorem} \label{thm_generic_vertices}
Let $W$ be a generic consistent web of 5-branes with 7-branes. Then, the dual lattice polygon of every vertex of $W$ is either an elementary lattice triangle or an elementary lattice trapezoid.    
\end{theorem}

\begin{proof}
By Theorem \ref{thm: generic push in},  there exists a modified web of 5-branes with 7-branes $W'$, obtained from $W$ by a finite sequence of pushing operation of 7-branes, such that
the 7-branes of $W'$ are fully pushed, and all the vertices of $W'$ are 3-valent of multiplicity one. We claim that, once the 7-branes are fully pushed, there are no triple intersections of three monodromy cuts, no triple intersection of two monodromy cuts with an edge of $W'$, and no intersection of a monodromy cut with a vertex of $W'$.
Indeed, all these situations are non-generic with respect to small perturbations of 7-branes in directions transverse to their monodromy directions. If we had such a situation, perturbing the 7-branes and monodromy cuts in directions transverse to the monodromy directions, and then pushing back the 7-brane would produce a consistent non-trivial perturbation of $W$, which contradicts the assumption that $W$ is generic.

Let $v$ be a vertex of $W$. We have to prove that the dual polygon of $v$ is either an elementary triangle or an elementary trapezoid. If $v$ is still a vertex of $W'$, then, since no monodromy cut passes through $v$, the local web around $v$ in $W$ coincides with the local web around  of $W'$ around $v$.
Hence, the vertex $v$ in $W$ is 3-valent of multiplicity one, and so with dual given by a lattice triangle of size one, which is in particular elementary. 

If $v$ is no longer a vertex of $W'$, it means that $v$ disappeared after pushing operations of 7-branes. 
If $v$ is contained in an edge of $W'$, then, there is at most one monodromy cut passing through $v$, and so $v$
in $W$ is either 3-valent or 4-valent. Moreover, this vertex disappeared after a single 7-brane crossed $v$, and so it follows from the local analysis in \S\ref{sec_examples}
and from Remark \ref{remark_elementary_subdivision}
that the dual polygon is either an elementary 
triangle or an elementary trapezoid.
Similarly, if $v$ is not contained in $W'$, then there are at most two monodromy cuts passing through $v$, and it follows that $v$
in $W$ is either 3-valent or 4-valent. Moreover, this vertex disappeared after a single 7-brane crossed $v$, and so it follows as previously that the dual polygon is either an elementary 
triangle or an elementary trapezoid.
\end{proof}

\section{Consistent decorated toric polygons and Symington polygons}
\label{section_DTP}

In \S\ref{sec_consistent_DTP}, we introduce the concept of a consistent decorated toric polygon with a polyhedral decomposition and we prove that it is dual to the notion of a generic consistent web of 5-branes with 7-branes.
In \S\ref{section_symington}, we show that consistent decorated toric polygons can be turned by cutting and gluing operations into an integral affine manifold with singularities known as a Symington polygon. 
Finally,
we compare in \S\ref{section_comparison} consistent decorated toric polygons with Generalized Toric Polygons (GTPs) satisfying ``s-rules" previously introduced in the physics literature.

\subsection{Consistent decorated toric polygons}
\label{sec_consistent_DTP}

As in \S\ref{sec_dual_polygon_convex}, a polyhedral decomposition  $\overline{\mathcal{P}}$ of a lattice polytope $\overline{P}$ is said to be \emph{regular} if there exists a convex continuous piecewise linear function on $\overline{P}$ whose domains of linearity coincide with the faces of $\overline{\mathcal{P}}$.

\begin{definition}
\label{defn:dtp}
A decorated toric polygon $(\overline{P}, (L_i)_{1\leq i\leq Q})$ is a lattice polygon $\overline{P}$ in $\RR^2$, with a choice of non-overlapping integral line segments $L_i \subset \partial \overline{P}$ indexed by $1\leq i \leq Q$, where $Q \in \ZZ_{\geq 0}$. 
A \emph{regular polyhedral decomposition} $\overline{\mathcal{P}}$ of a decorated toric polygon $(\overline{P}, (L_i)_{1\leq i\leq Q})$ is a regular polyhedral decomposition of $\overline{P}$ such that the line segments $L_i$ are edges of $\overline{\mathcal{P}}$. We denote by $\overline{\mathcal{P}}^{[0]}$, $\overline{\mathcal{P}}^{[1]}$ and $\overline{\mathcal{P}}^{[2]}$ respectively the set of vertices, edges and faces of $\overline{\mathcal{P}}$.  
\end{definition}

Let $(\overline{P}, (L_i)_{1\leq i\leq Q})$ be a decorated toric polygon equipped with a regular polyhedral decomposition $\overline{\mathcal{P}}$. 
We define a web of 5-branes with 7-branes $W$ to be \emph{dual} to $(\overline{P}, (L_i)_{1\leq i\leq Q},\overline{\mathcal{P}})$ if 
there exists a convex continuous piecewise-linear function $\overline{\varphi}$ on $\overline{P}$ with domains of linearity $\overline{\mathcal{P}}$,
such that, denoting by  $\overline{W}$ the web of $5$-branes dual to 
$(\overline{\mathcal{P}},\overline{\varphi})$, 
$W$ is obtained from $\overline{W}$ by adding for every $1\leq i \leq Q$ a 7-brane ending the leg of $\overline{W}$ corresponding to $L_i$.

Recall from \S\ref{section_webs_log_CY} that any web of $5$-branes with $7$-branes determines a log Calabi surface with line bundle, denoted $(Y,D,L)$.
In what follows, the expression \emph{the log Calabi surface with line bundle associated to a decorated toric polygon with a regular polyhedral decomposition} refers to the corresponding data $(Y,D,L)$ associated with the dual web.

\begin{remark} \label{remark_GTP}
The notion of decorated toric polygon is closely related to the notion of generalized toric polygons (GTP) in the terminology of \cite{bourget2023generalized}. 
A GTP is a lattice polygon whose boundary integral points are colored into either black or white, with the condition that the vertices are colored in black. 
Every decorated toric polygon $(\overline{P}, (L_i)_{1\leq i\leq Q})$ defines a GTP by coloring in white the boundary integral points contained in the interior of one of the line segments $L_i$. 
Every GTP comes from a decorated toric polygon this way, but not in a unique way since line segments $L_i$ of integral length one do not give rise to any white dot. 
In \cite{bourget2023generalized}, GTPs are thought of as dual to webs with no unbounded 5-branes, whereas the notion of decorated toric polygon allows for unbounded 5-branes.
\end{remark}

To define the notion of ``consistency'' for a decorated toric polygon with a polyhedral decomposition, we first introduce the notion of a state.
\begin{definition}
\label{defn: state}
A \emph{state} 
 $\Psi=(\sim_{\Psi}, ({\mathcal{F}_{\Psi,i})}_{1\leq i \leq Q},  (A_{\Psi,i})_{1\leq i \leq Q} )$
on a a decorated toric polygon $(\overline{P}, (L_i)_i)$ with a polyhedral decomposition $\overline{\mathcal{P}}$ consists of:
\begin{itemize}
    \item[i)] An equivalence relation $\sim_{\Psi}$ on the set of edges $\overline{\mathcal{P}}^{[1]}$ of $\overline{\mathcal{P}}$, preserving the lattice length of edges. 
    \item[ii)] For every $1 \leq i \leq Q$, a set of faces
$\mathcal{F}_{\Psi,i} \subseteq \overline{\mathcal{P}}^{[2]}$. We refer to a face in $\mathcal{F}_{\Psi,i}$ as a ``$(\Psi,i)$-colored face'', and require that $\mathcal{F}_{\Psi,i} \cap \mathcal{F}_{\Psi,j} = \emptyset$, if $i\neq j$.
\item[iii)] For every $1 \leq i \leq Q$, an element $A_{\Psi,i} \in \overline{\mathcal{P}}^{[0]} \cup \overline{\mathcal{P}}^{[1]}$, referred to as the ``$(\Psi,i)$-active cell''.
\end{itemize}
\end{definition}

Now, we define the ``initial state'' $\Psi_{\mathrm{in}}=(\sim_{\Psi_{\mathrm{in}}}, ({\mathcal{F}_{\Psi_{\mathrm{in}},i})}_{1\leq i \leq Q},  (A_{\Psi_{\mathrm{in}},i})_{1\leq i \leq Q} )$ as follows:
\begin{itemize}
    \item[i)] $\sim_{\Psi_{\mathrm{in}}}$ is the trivial equivalence relation, that is, every edge is only equivalent to itself,
    \item[ii)] ${\mathcal{F}_{\Psi_{\mathrm{in}},i}} = \emptyset$, and;
    \item[iii)] The $(\Psi,i)$-active cell $A_{\Psi_{\mathrm{in}},i}$ is the edge $L_i$.
\end{itemize}

We define the following two moves on the set of states.

\begin{itemize}
    \item[1)] \emph{Triangle-coloring move}: Let $\Psi$ be a state, and let $1\leq i\leq Q$ be such that the $(\Psi,i)$-active cell $A_{\Psi,i}$ is an edge of $\mathcal{P}$.
Assume that there exists an edge $e_{\mathrm{in}}$ such that $e_{\mathrm{in}} \sim_{\Psi} A_{\Psi,i}$ and that there exists a face $\sigma$ of $\mathcal{P}$ such that:
\begin{itemize}
    \item[i)] $e_{\mathrm{in}}$ is an edge of $\sigma$.
    \item[ii)] $\sigma$ is an elementary triangle, of height $|e_{in}|$ with respect to the edge $e_{in}$, and with $|e_r|=|e_l|=1$, where $e_r$, $e_l$ are the two sides of $\sigma$ distinct of $e_{\mathrm{in}}$.
      \item[iii)] $\sigma$ is not $(\Psi,j)$-colored for any $1\leq j \leq Q$.
\end{itemize}
Then, we define a new state $\Psi'$, obtained from $\Psi$ by \emph{triangle-coloring} as follows:
\begin{itemize}
    \item[i)] We define $\sim_{\Psi'}$ as the equivalence relation generated by $\sim_{\Psi}$ and $e_l \sim_{\Psi'} e_r$.
    \item[ii)] We define $\mathcal{F}_{\Psi',i}
    =\mathcal{F}_{\Psi,i} \cup \{\sigma\}$ and $\mathcal{F}_{\Psi',j}=\mathcal{F}_{\Psi,j}$ for $j \neq i$, that is, $\sigma$ is $(\Psi',i)$-colored.
    \item[iii)] We define $A_{\Psi',i}$ as the vertex of $\sigma$ not contained in $e_{\mathrm{in}}$.
\end{itemize}

\item[2)] \emph{Trapezoid-coloring move}: Let $\Psi$ be a state, and let $1\leq i\leq Q$ be such that the $(\Psi,i)$-active cell $A_{\Psi,i}$ is an edge of $\mathcal{P}$.
Assume that there exists an edge $e_{\mathrm{in}}$ such that $e_{\mathrm{in}} \sim_{\Psi} A_{\Psi,i}$ and that there exists a face $\sigma$ of $\mathcal{P}$ such that:
\begin{itemize}
    \item[i)] $e_{\mathrm{in}}$ is an edge of $\sigma$.
    \item[ii)] $\sigma$ is an elementary trapezoid, such that, denoting by $e_{\mathrm{out}}$ the side of $\sigma$ parallel and opposite to $e_{\mathrm{in}}$, we have $|e_{\mathrm{out}}|=|e_{\mathrm{in}}|-h$, where $h$ is the lattice height of $\sigma$ with respect to $e_{\mathrm{in}}$, and such that $|e_r|=|e_l|=1$, where $e_r$, $e_l$ are the two other sides of $\sigma$.
      \item[iii)] $\sigma$ is not $(\Psi,j)$-colored for any $1\leq j \leq Q$.
\end{itemize}
Then, we define a new state $\Psi'$, obtained from $\Psi$ by \emph{trapezoid-coloring} as follows:
\begin{itemize}
    \item[i)] We define $\sim_{\Psi'}$ as the equivalence relation generated by $\sim_{\Psi}$ and $e_l \sim_{\Psi'} e_r$.
    \item[ii)] We define $\mathcal{F}_{\Psi',i}
    =\mathcal{F}_{\Psi,i} \cup \{\sigma\}$ and $\mathcal{F}_{\Psi',j}=\mathcal{F}_{\Psi,j}$ for $j \neq i$, that is, $\sigma$ is $(\Psi',i)$-colored.
    \item[iii)] We define $A_{\Psi',i}:=e_{\mathrm{out}}$.
\end{itemize}
\end{itemize}

Note that, for either a triangle-coloring or a trapezoid-coloring move, there exists a unique matrix $M \in SL(2,\ZZ)$ fixing the direction of $e_{\mathrm{in}}$ and sending the primitive integral direction of $e_l$ on the primitive integral direction of $e_r$. More generally, if $\Psi$ is a state obtained from the initial state by applying a finite sequence of triangle-coloring and trapezoid-coloring moves, and if $a$ and $b$ are two vertices of $\overline{\mathcal{P}}$ such that $a \sim_\Psi b$, then the composition of these matrices in $SL(2,\ZZ)$ or their inverse following a decomposition of $a \sim_\Psi b$ into ``elementary" equivalences induced by the moves defines a matrice $M_{ab}\in SL(2,\ZZ)$.

\begin{definition}
\label{Defn: DTP consistent}
A decorated toric polygon $(\overline{P}, (L_i)_{1\leq i\leq Q})$ with a regular polyhedral decomposition $\overline{\mathcal{P}}$ is \emph{consistent} if there exists a state $\Psi$ obtained from the initial state on $(\overline{P}, (L_i)_{1\leq i\leq Q})$ by a finite sequence of triangle-coloring and trapezoid-coloring moves, such that:
\begin{itemize}
    \item[i)] The $(\Psi,i)$-active cells $A_{\Psi, i}$ are vertices for all $1\leq i \leq Q$. Moreover, if not every face of $\mathcal{P}$ is $\Psi$-colored, we require that there exists a vertex $b$ of $\overline{\mathcal{P}}$ such that $A \sim_{\Psi} b$, and, denoting by $u$ the direction of the base $e_{\mathrm{in}}$ of the elementary triangle whose $\Psi$-coloring made $A_{\Psi,i}$ active, the intersection 
    \[ (b+\RR\, M_{A_{\Psi,i} b}\, u) \cap P\]
    is not entirely $\Psi$-colored.
    \item[ii)] All non-$\Psi$-colored faces (that is, elements of $\overline{\mathcal{P}}^{[2]} \setminus \cup_{1\leq i\leq Q} \mathcal{F}_{\Psi,i}$) are triangles of size one. 
    \end{itemize}
 We refer to such a state $\Psi$ as a \emph{final state}.   
\end{definition}

Given an initial state, there are finitely many possible sequences of triangle-coloring or trapezoid-coloring moves. Therefore, it follows from Definition \ref{Defn: DTP consistent} that we can check algorithmically if a decorated toric polygon with a regular polyhedral decomposition is consistent or not. 

Before proceeding, we discuss some properties of a final state. 
\begin{lemma} \label{lem_final_state}
Let $\Psi$ be a final state on a consistent decorated toric polygon $(\overline{P}, (L_i)_{1\leq i\leq Q})$ with a regular polyhedral decomposition $\overline{\mathcal{P}}$.
Then, for every $1\leq i \leq Q$, the union of $(\Psi,i)$-colored cells is of the form $\sigma_1 \cup \sigma_k$, where:
\begin{itemize}
\item[i)]
for all $1 \leq j \leq k-1$, $\sigma_j$ is an elementary trapezoid, with two parallel edges $e_{in,j}$ and $e_{out,j}$, which were once active cells, and with two other edges $e_{l,j}$ and $e_{r,j}$ such that $e_{l,j} \sim_{\Psi} e_{r,j}$.
\item[ii)] $\sigma_k$ is an elementary triangle, with an edge $e_{in, k}$, which was once an active cell, and with two other edges $e_{l,k}$ and $e_{r,k}$ such that 
$e_{l,k} \sim_{\Psi} e_{r,k}$.
\end{itemize}
Moreover, we have $|e_{in,1}|=|L_i|$, and, for every $1 \leq j \leq k-1$, we have 
$|e_{out,j}|=|e_{in, j+1}|$. 
\end{lemma}

\begin{proof}
For every $1\leq i\leq Q$, the union of $(\Psi,i)$-colored cells is necessarily obtained by first applying a sequence of trapezoid coloring moves, and then finally applying a triangle-coloring move. Hence, the result follows from the description of trapezoid-coloring and triangle-coloring moves.
\end{proof}

\begin{lemma} \label{lem_colored_area}
Let $\Psi$ be a final state on a consistent decorated toric polygon $(\overline{P}, (L_i)_{1\leq i\leq Q})$ with a regular polyhedral decomposition $\overline{\mathcal{P}}$. Then, the area of the $(\Psi,i)$-colored part of $\overline{P}$ is $\frac{1}{2}|L_i|^2$.
\end{lemma}

\begin{proof}
Using the notations introduced in Lemma \ref{lem_final_state}, 
for every $1\leq j \leq k$, we denote by $h_j$ the height of $\sigma_j$. Note that $\sum_{j=1}^k h_j=L_i$.
Then, the normalized area of the $(\Psi,i)$-colored part is
\begin{align*} \frac{1}{2}\left((2L_i-h_1)h_1+(2L_i-2h_1-h_2)h_2+\cdots \right)
&= \frac{1}{2}
\left(2L_i \sum_{j=1}^k h_j - \left(\sum_{j=1}^k h_j\right)^2 \right)\\
&= 
\frac{1}{2}(2L_i^2-L_i^2) = \frac{1}{2} L_i^2 \,. 
\end{align*}
\end{proof}

Finally, we show that a decorated toric polygon with a regular polyhedral decomposition is consistent as in Definition \ref{Defn: DTP consistent}, if and only if the associated dual web is consistent.

\begin{theorem} \label{thm_polytopes_webs}
    Let $(\overline{P}, (L_i)_{1\leq i\leq Q})$ be a decorated toric polygon with a regular polyhedral decomposition $\overline{\mathcal{P}}$, and let $W$ be a web of 5-branes with 7-branes dual to $(\overline{P}, (L_i)_{1\leq i\leq Q},\overline{\mathcal{P}})$. Then, 
    the decorated toric polygon 
    $(\overline{P}, (L_i)_{1\leq i\leq Q})$ with regular polyhedral decomposition  $\overline{\mathcal{P}}$ is consistent if and only if $W$ is a generic consistent web of 5-branes with 7-branes.
\end{theorem}

\begin{proof}
Recall by Theorem \ref{thm: generic push in}, a web $W$ of $5$-branes with $7$-branes is a generic consistent web if and only if there is a modified web $W'$ obtained from $W$ such that all $7$-branes in $W'$ are fully pushed as in Definition \ref{def_fully_pushed}, and furthermore all vertices of $W'$ are 3-valent. 

Let $W$ be a web of $5$-branes with $7$-branes dual to a decorated toric polygon with a regular polyhedral decomposition $(\overline{P}, (L_i)_{1\leq i\leq Q}, \overline{\mathcal{P}})$.  
By Theorem \ref{thm_generic_vertices}, the operation of pushing $7$ branes in $W$ amounts to triangle-coloring and trapezoid-coloring moves on the initial state of $(\overline{P}, (L_i)_{1\leq i\leq Q}, \overline{\mathcal{P}})$. Thus, the modified web $W'$ is dual to a decorated toric polygon with a final state obtained from the initial state by applying consecutively triangle and trapezoid coloring moves. Condition i) in Definition \ref{Defn: DTP consistent} of a final state is dual to the fact that the 7-branes in $W'$ are fully pushed. Finally, condition ii) in Definition \ref{Defn: DTP consistent} of a final state is dual to the fact that all vertices of $W'$ are 3-valent, and this concludes the proof. 
\end{proof}

\begin{example}
    Figure \ref{Fig37} depicts a consistent decorated toric polygon with a regular polyhedral decomposition, along with a final state. This final state is obtained by first applying a trapezoid-coloring move, and then a triangle-coloring move. The trapezoid-coloring move is dual to the 7-brane crossing a 4-valent vertex of the web, whereas the triangle-coloring move is dual to the 7-brane crossing a 3-valent vertex of the web.
\end{example}

\begin{figure}[hbt!]
\center{\scalebox{.9}{\includegraphics{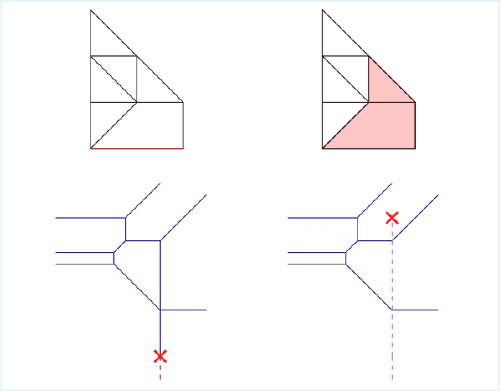}}}
\caption{On the left, a consistent decorated toric polygon with a regular polyhedral decomposition and its dual web. On the right, a corresponding coloring and the dual web with fully pushed 7-branes.}
\label{Fig37}
\end{figure}

\subsection{Symington polygons}
 \label{section_symington}
A ``Symington polygon'' is a polygon endowed with an integral affine structure with singularities, as introduced in \cite{symington} where such polygons are referred to as \emph{almost toric bases}. In this section, we provide a recipe to construct a Symington polygon from a consistent decorated toric polygon with a regular polyhedral decomposition. We also explain how to view such a Symington polygon as the base of an almost toric fibration for the associated log Calabi--Yau surface with line bundle $(Y,D,L)$.

We explain in a moment how to construct the Symington polygon by some cutting-gluing operations on a decorated toric polygon with a regular polyhedral decomposition. Let $(\overline{P}, (L_i)_{1\leq i\leq Q})$ be a consistent decorated toric polygon with a regular polyhedral decomposition $\mathcal{\overline{P}}$. Denote by $\Psi$ a final state on $(\overline{P}, (L_i)_{1\leq i\leq Q}, \mathcal{\overline{P}})$. 
Associated to the data of $(\overline{P}, (L_i)_{1\leq i\leq Q}, \mathcal{\overline{P}})$ and the final state $\Psi$, we define the Symington polygon, which we denote by $P$, as the integral affine manifold with singularities obtained as follows: 
First, cut out all the line segments $L_i$ and the interiors of the unions of the $(\Psi,i)$-colored faces. Then, glue the remaining edges which are equivalent according to the equivalence relation $\sim_{\Psi}$, that is;
\begin{equation}
\label{eq: identify}
    P:= \left(  \overline{P} \setminus \left( \bigcup_{i=1}^Q \mathrm{Int}(\mathcal{C}_i) \cup L_i \right) \right) \Big/ \sim_{\Psi} \,,
\end{equation}
where $\mathcal{C}_i = \bigcup_{\sigma \in \mathcal{F}_{\Psi,i}} \sigma$ is the union of all $(\Psi,i)$-colored faces, and $\mathrm{Int}(\mathcal{C}_i)$ denotes the interior of $\mathcal{C}_i$. 
The equivalence relation $\sim_{\Psi}$ is still well-defined after the cuts, by Lemma \ref{lem_final_state}. Moreover, the matrices $M_{ab}\in SL(2,\ZZ)$ introduced above Definition \ref{Defn: DTP consistent} for pairs of vertices $a$ and $b$ such that $a \sim_{\Psi} b$ are the matrices that identify the integral tangent space at $a$ with the integral tangent space at $b$.

The singularities of the integral affine structure on $P$ are located at the images of the active vertices $A_{\Psi, i}$ of $\mathcal{\overline{P}}$ in $P$, under the identification \eqref{eq: identify}. On the other hand, the polyhedral decomposition $\overline{\mathcal{P}}$ of $\overline{P}$ naturally induces
a polyhedral decomposition $\mathcal{P}$ on the Symington polygon $P$.
By construction, all singularities of the affine structure on $P$ are at vertices of $\mathcal{P}$. Furthermore, all faces of $\mathcal{P}$ are lattice triangles of size one, by the definition of consistency of a decorated toric polygon (see item ii) in Definition \ref{Defn: DTP consistent}). In particular, the number of faces of $\mathcal{P}$ is equal to $2 \mathrm{Area}(P)$, where the area $\mathrm{Area}(P)$ of $P$ is given by the following result:

\begin{lemma} \label{lemma_area}
Let $\Psi$ be a final state on a consistent decorated toric polygon $(\overline{P}, (L_i)_{1\leq i\leq Q})$ with a regular polyhedral decomposition $\overline{\mathcal{P}}$, and let $(P,\mathcal{P})$ be the corresponding Symington polygon. Then,
denoting by $(Y,D,L)$ the associated log Calabi--Yau surface with line bundle,
we have 
\begin{equation} \label{eq_area_symington}
\mathrm{Area}(P)=\frac{1}{2}L^2 \,.\end{equation}
In particular, $P$ is an interval if $L^2=0$. Moreover, if $L^2 >0$, then $P$ is homeomorphic to a disk if $L \cdot D>0$, and $P$ is homeomorphic to a sphere if $L\cdot D=0$.
\end{lemma}

\begin{proof}
Since $P$ is obtained from $\overline{P}$ by cutting the $(\Psi,i)$-colored cells for all $1\leq i \leq Q$, it follows from 
Lemma \ref{lem_colored_area} that $2\mathrm{Area}(P)=2\mathrm{Area}(\overline{P})-\sum_{i=1}^Q |L_i|^2$. Hence, Equation \eqref{eq_area_symington} follows from Lemma 
\ref{lem_L2}. Finally, if $L^2>0$, then, depending on whether $L\cdot D=0$ or $L\cdot D>0$, the boundary of $\overline{P}$ is either fully contracted to a point or remains an interval of positive length in $P$, and so $P$ is either homeomorphic to a disk or an interval.
\end{proof}

Let $(P,\mathcal{P})$ be the Symington polygon corresponding to a final state of a consistent decorated toric polygon with a regular polyhedral decomposition, with associated log Calabi--Yau surface $(Y,D,L)$. By Theorem \ref{thm_polytopes_webs}, the corresponding dual web of 5-branes with 7-branes $W$ is generic and consistent.
If $L^2>0$, then, as explained in \S\ref{section_fully_pushed}, 
the web $W$ defines a degeneration of polarized log Calabi--Yau surfaces $((\mathcal{Y}')^{\mathrm{pol}}, (\mathcal{D}')^{\mathrm{pol}}, (\mathcal{L}')^{\mathrm{pol}}) \rightarrow \CC$, whose dual intersection complex of the central fiber is the modified web $W'$, and whose general fiber $(Y^{\mathrm{pol}}, D^{\mathrm{pol}}, L^{\mathrm{pol}})$ is the polarized 
log Calabi--Yau surface obtained from $(Y,D,L)$
by contracting all the curves having intersection number zero with $L$. The Symington polygon $(P,\mathcal{P})$ can be viewed as the intersection complex of the central fiber of this degeneration
(see for instance \cite{GSreal}). In particular, the faces of $\mathcal{P}$ are lattice triangles of size one that are the momentum polytopes of the polarized irreducible toric components $(\PP^2, \partial \PP^2, \mathcal{O}(1))$ of the central fiber described in Lemma \ref{lem: generic_central_fiber}.

As in \cite[Lemma 5.6]{engel_friedman}, 
the integral affine structure with singularities on $P$ can be perturbed to obtain an integral affine manifold $\widetilde{P}$ 
with only focus-focus singularities. Then, as explained in \cite[\S5.4]{symington}, one can view $\widetilde{P}$ as the base of an \emph{almost toric fibration} on $(Y^{\mathrm{pol}},D^{\mathrm{pol}},L^{\mathrm{pol}})$, that is, a fibration whose general fibers are tori $T^2$, and singular fibers are pinched tori. Moreover, the fibers are Lagrangian submanifolds with respect to a symplectic form $\omega$ on $Y$ with class in $c_1(L^{\mathrm{pol}}) \in H^2(Y^{\mathrm{pol}},\RR)$. 
This fibration is obtained by surgeries introducing singular fibers in the toric momentum map fibration 
$(\overline{Y},\overline{D},\overline{L})\rightarrow\overline{P}$ on the polarized toric surface $(\overline{Y},\overline{D},\overline{L})$ with momentum polytope $\overline{P}$.
An alternative way to understand such almost toric fibrations is to consider the
degeneration 
$((\mathcal{Y}')^{\mathrm{pol}}, (\mathcal{D}')^{\mathrm{pol}}, (\mathcal{L}')^{\mathrm{pol}}) \rightarrow \CC$
of $(Y^{\mathrm{pol}},D^{\mathrm{pol}},L^{\mathrm{pol}})$, with central fiber given by a union of polarized toric surfaces $(\PP^2, \partial \PP^2, \mathcal{O}(1))$, and having $(P,\mathcal{P})$ as intersection complex. By appropriate scalings and refinements, one can ensure that after perturbing the decomposition $\widetilde{\mathcal{P}}$ on $\widetilde{P}$, all singularities still lie inside some edges. Then, as explained in \cite{arguz_KN}, the momentum maps on the irreducible components 
$(\PP^2, \partial \PP^2, \mathcal{O}(1))$ of the central fiber glue to a fibration with base the intersection complex $(\widetilde{P},\widetilde{\mathcal{P}})$. By \cite[Theorem 5.7]{arguz_KN}, the almost toric fibration on $(Y^{\mathrm{pol}},D^{\mathrm{pol}},L^{\mathrm{pol}
})$ can be obtained as a deformation of this fibration over $(\widetilde{P},\widetilde{\mathcal{P}})$.

\subsection{Comparison with existing s-rules}
\label{section_comparison}
Up to date, two proposals have been made
in the physics literature, providing so-called ``s-rules", for characterizing generalized toric polygons (GTPs) corresponding to 5d SCFTs \cite{BBT,van20215d}. Recall that decorated toric polygons are similar to such GTPs, though they encode slightly more further information -- see Remark \ref{remark_GTP} explaining how to translate between decorated toric polygons and GTP's. We provide below some concrete examples illustrating how each of the existing proposals in the literature differs from our notion of consistency.

The first proposed ``s-rule" can be found in \cite[\S 3.2]{BBT}. If a decorated lattice polygon $(\overline{P}, (L_i)_{1\leq i\leq Q})$ satisfies the s-rule of \cite[\S 3.2]{BBT}, then there exists in particular a polyhedral decomposition of $(\overline{P}, (L_i)_{1\leq i\leq Q})$ whose faces are either trapezoids or triangles whose three edges of the same lattice lengths. In particular, elementary triangles such as the one represented in the middle of Figure \ref{Fig:basic_triangles} are not allowed in the polyhedral decompositions considered in \cite[\S 3.2]{BBT}. 
It follows that the decorated polygon in Figure \ref{Fig29} does not satisfy the ``s-rule" of \cite[\S 3.2]{BBT}, whereas it admits a consistent regular polyhedral decomposition in the sense of Definition \ref{Defn: DTP consistent}. The M-theory dual Calabi--Yau 3-fold to the corresponding web of $5$-branes with $7$-branes for this example is described in Example \ref{Ex: fig 29}.

A modified version of the s-rule can be found in \cite[\S 3.3]{van2020symplectic}. In \cite{van2020symplectic} it is proposed
that a GTP defines a 5d SCFT if the s-rule in \cite[\S 3.3]{van2020symplectic} and the r-rule in 
\cite[\S 3.4]{van2020symplectic} are satisfied. 
Elementary triangles and elementary trapezoids are tiles in the sense of \cite[Definition 5]{van2020symplectic}, and so it follows that a consistent decorated toric polygon $(P, (L_i)_{1\leq i \leq Q})$ with a regular polyhedral decomposition $\overline{\mathcal{P}}$ satisfies the s-rule of \cite[\S 3.3]{van2020symplectic}. It also satisfies the 
r-rule $r \geq 0$ in 
\cite[\S 3.4]{van2020symplectic} since, by Lemma \ref{lemma_area}, the rank $r$ of the corresponding GTP is equal to $\mathrm{Area}(P)+1$, where $\mathrm{Area}(P)$ is the area of the associated Symington polygon. However, there are GTPs that satisfy both the s-rule and the r-rule but such that the corresponding 
decorated toric polygon is not consistent for any regular polyhedral decomposition. Moreover, even when the dual web is supersymmetric in the sense of Definition 
\ref{def_supersymmetric_asymp_web}, the rank of the corresponding 5d SCFT is not necessarily equal to the rank $r$ of the GTP, for the general reason presented in
Remark \ref{remark_r_rule} and contrary to the expectation of \cite{van2020symplectic} -- see example below. It is an interesting open question to determine if the web dual to a GTP that satisfies both the s-and the r-rule of \cite{van2020symplectic} is always supersymmetric in the sense of Definition \ref{def_supersymmetric_asymp_web}.

\begin{example}\label{example_gtp} Consider the GTP $\overline{P}$ on the right of Figure \ref{Fig16}. The polyhedral decomposition of $\overline{P}$
represented in Figure \ref{Fig16} satisfies the s-rule. The r-rule is also satisfied: the area of $\overline{P}$ is $\frac{37}{2}$ whereas the sum of squares of the parts of the partitions determined by the white dots is 
\[ 4^2+1^2+1^2+1^2+4^2+1^2+1^2=37\,,\]
 and so we have 
$r=\frac{37}{2}+1-\frac{37}{2}=1$ by \cite[Equation 3.47]{van2020symplectic}. Therefore, \cite{van2020symplectic} predicts that the GTP $\overline{P}$ defines a rank one 5d SCFT. However, we argue that it is not the case as can be seen from the web picture. 
Pushing 7-branes as on the right of Figure \ref{Fig17} shows that the general curve $C \in |L|$ in the corresponding log Calabi--Yau surface with line bundle $(Y,D,L)$ is the disjoint union of a connected genus two curve and of a genus zero $(-2)$-curve. 
The arithmetic genus of this disconnected curve $C$ is indeed equal to $r=1$, but $C$ is not a connected genus one curve. Since $C$ is smooth, the web is supersymmetric according to Definition \ref{def_supersymmetric_asymp_web}, and so we also expect the existence of a 5d SCFT.
However, this 5d SCFT is not of rank one, but is obtained by adding one free hypermultiplet to a rank two SCFT, as explained following Definition \ref{def: consistency}. In fact, since a component of $C$ has negative self-intersection, it follows from Lemma \ref{lem_L_nef}
that the corresponding decorated toric polygon is not consistent -- this can also be checked directly by inspecting the possible polyhedral decompositions.  
\end{example}

\begin{figure}[hbt!]
\center{\scalebox{0.9}{\includegraphics{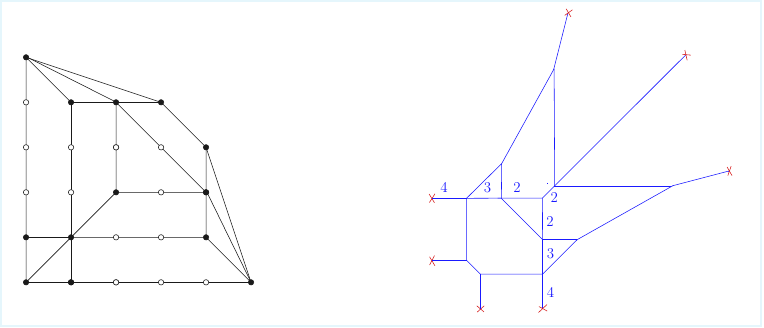}}}
\caption{A non-consistent GTP that satisfies both the s-rule and the r-rule of \cite{van2020symplectic}.}
\label{Fig16}
\end{figure}

\begin{figure}[hbt!]
\center{\scalebox{0.9}{\includegraphics{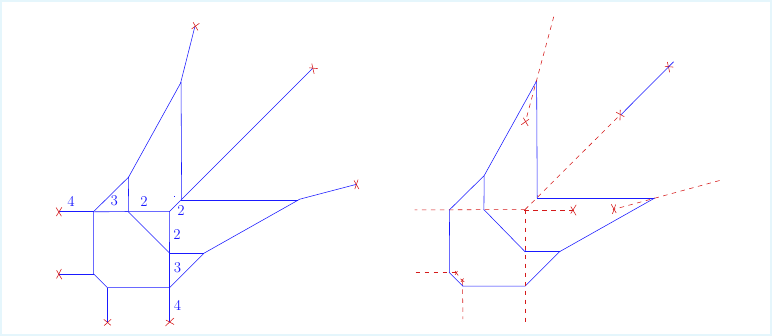}}}
\caption{The general curve $C \in |L|$ is a disjoint union of a genus two curve and of a genus zero $(-2)$-curve.}
\label{Fig17}
\end{figure}

\section{M-theory dual Calabi--Yau 3-folds to consistent webs of 5-branes with 7-branes}
\label{section_cy3}

In \S\ref{section_cy3_mirror}, we explain that the M-theory dual canonical 3-fold $\mathcal{X}^{\mathrm{can}}$ to a consistent web of 5-branes with 7-branes is the mirror to a Calabi--Yau 3-fold $Z$ obtained from the log Calabi--Yau surface with line bundle
$(Y,D,L)$ associated to the web. 
In \S\ref{section_cy3_construction}, we construct this mirror, by first constructing non-compact Calabi--Yau 3-folds $\mathcal{X}$ associated to generic consistent webs of 5-branes with 7-branes by non-toric deformations of toric Calabi--Yau 3-folds, and then proving the existence of a contraction $\mathcal{X}\rightarrow \mathcal{X}^{\mathrm{can}}$. 
In \S\ref{section_slc_sing}, we characterize the canonical 3-fold singularities
produced this way as total space of one-parameter smoothings of degenerate cusp, cusp, and simple elliptic surface singularities. In \S\ref{section_scattering}, we demonstrate how to obtain explicit equations for 
$\mathcal{X}^{\mathrm{can}}$ using a combinatorial objects called scattering diagrams. Finally, we explain in \S\ref{section_instanton_BPS} that scattering diagrams capture contributions of disk worldsheet instantons of the A-model, which can also be interpreted as contributions of BPS states of an auxiliary 4d $\mathcal{N}=2$
rank one quantum field theory.

\subsection{Calabi-Yau 3-folds and mirror symmetry for log Calabi--Yau surfaces}
\label{section_cy3_mirror}
As explained in \S\ref{section_consistent}, a consistent asymptotic web of 5-branes in 7-branes $W^{\mathrm{asym}}$ in Type IIB string theory is expected to define a 5d SCFT. In this section, we describe a general strategy to construct a canonical 3-fold singularity $\mathcal{X}^{\mathrm{can}}$ such that this 5d SCFT admits a dual description as M-theory on $\mathcal{X}^{\mathrm{can}}$. The detailed construction is presented in the following section \S\ref{section_cy3_construction}.

As in the toric case reviewed in \S \ref{section_toric_dual_CY3}, we will obtain $\mathcal{X}^{\mathrm{can}}$ as the result of a mirror construction. To explain this, we use as in the toric case the following chain of string dualities:

\begin{itemize}
    \item[i)] Denote by $(Y,D,L)$ the log Calabi--Yau surface with line bundle associated to $W^{\mathrm{asym}}$, $C$ a general curve in $|L|$, $U=Y \setminus D$ and $C^\circ= C \cap U$.
    Then, as discussed in \S \ref{section_consistent}, the duality between Type IIB string theory on $S^1$ and M-theory on $T^2$ implies that the $S^1$-compactification of the 5d theory defined by $W^{\mathrm{asym}}$ in Type IIB string theory has a dual description as the $4d$
    $\mathcal{N}=2$ theory defined by an M5-brane on $\RR^4 \times C^\circ$ in M-theory on $\RR^4 \times U \times \RR^3$. 
    \item[ii)] Compactifying a transverse direction to the M5-brane on a $S^1$, and using that M-theory on $S^1$ is equivalent to Type IIA string theory, the $4d$
    $\mathcal{N}=2$ theory can be described using an NS5-brane on $\RR^4 \times C^\circ$ in Type IIA string theory on $\RR^4 \times U \times \RR^2$. 
    \item[iii)] Compactifying another transverse direction on $S^1$, and applying T-duality between Type IIA and Type IIB string theory, the $4d$
    $\mathcal{N}=2$ theory can be described as Type IIB string theory compactified on the Calabi--Yau 3-fold $Z$, defined by 
    \[ Z = \{ uv=f\} \subset \mathrm{Tot}(\mathcal{O}_U \oplus L|_U)\,,\]
    where $\mathrm{Tot}(\mathcal{O}_U \oplus L|_U)$ is the total space of the rank $2$ vector bundle $\mathcal{O}_U \oplus L|_U$ over $U$, with fiber coordinates $u,v$, and where $f$ is the equation defining $C^\circ$.
    \item[iv)] If we could construct the mirror Calabi--Yau 3-fold $\mathcal{X}^{\mathrm{can}}$ to $Z$, then the 4d $\mathcal{N}=2$ theory would be described as Type IIA string theory on $\mathcal{X}^{\mathrm{can}}$, that is, M-theory on $\mathcal{X}^{\mathrm{can}} \times S^1$. 
\end{itemize}

Combining the previous steps i) to iv) shows that the $S^1$-compactifications of both 5d SCFTs defined by the web $W^{\mathrm{asym}}$ in Type IIB string theory and by M-theory on the Calabi--Yau 3-fold $\mathcal{X}^{\mathrm{can}}$ coincide. 
Decompactifying $S^1$ leads to the expectation that these 5d SCFTs are the same.

Therefore, our task is to construct the mirror $\mathcal{X}^{\mathrm{can}}$ to the Calabi--Yau 3-fold $Z$. Since $Z$ is determined by the log Calabi--Yau surface with line bundle $(Y,D,L)$, we will refer to $\mathcal{X}^{\mathrm{can}}$ as the mirror to $(Y,D,L)$.
In the following section \S\ref{section_cy3_construction}, we describe how to construct a canonical 3-fold singularity $\mathcal{X}^{\mathrm{can}}$ mirror to $(Y,D,L)$
following recent mathematical developments in mirror symmetry for log Calabi--Yau surfaces \cite{alexeev2024ksba, AAB2, engel, engel_friedman, GHK1}. We will first construct the smooth Calabi--Yau 3-folds $\mathcal{X}$ which are crepant resolutions of $\mathcal{X}^{\mathrm{can}}$ and which are
mirror to degenerations of log Calabi--Yau surfaces $(\mathcal{Y}',\mathcal{D}',\mathcal{L}') \rightarrow \CC$ associated to modified webs obtained from generic consistent perturbations $W$ of $W^{\mathrm{asym}}$ by fully pushing 7-branes, as in \S\ref{section_fully_pushed}. The smooth Calabi--Yau 3-fold $\mathcal{X}$ will be defined by smoothing the toric Calabi--Yau 3-fold $\overline{\mathcal{X}}$ combinatorially constructed as in 
\S\ref{section_toric_dual_CY3} from the web of 5-branes $\overline{W}$ obtained from $W$ by removing the 7-branes.

\subsection{Calabi--Yau 3-folds as non-toric deformations of toric Calabi--Yau 3-folds}
\label{section_cy3_construction}

Let $W$ be a generic consistent web of 5-branes with 7-branes, defined by a web of 5-branes $\overline{W}$ and a 7-brane data. As explained in \S\ref{sec: degeneration}, $\overline{W}$ defines a degeneration $(\overline{\mathcal{Y}}, \overline{\mathcal{D}},\overline{\mathcal{L}}) \rightarrow \CC$ of polarized log Calabi--Yau surfaces, and $W$ determines a degeneration $(\mathcal{Y}, \mathcal{D},\mathcal{L}) \rightarrow \CC$ of log Calabi--Yau surfaces obtained from $(\overline{\mathcal{Y}}, \overline{\mathcal{D}},\overline{\mathcal{L}}) \rightarrow \CC$ by blow-ups. 
By \S\ref{section_consistent_webs}, although the line bundle $\mathcal{L}$ is generally not nef, 
one can perform a finite sequence of M1-flop to produce a degeneration $(\overline{\mathcal{Y}}', \overline{\mathcal{D}}',\overline{\mathcal{L}}') \rightarrow \CC$ of log Calabi--Yau surface with nef $\mathcal{L}'$. 
According to \S\ref{section_fully_pushed}, this sequence of M1-flops is encoded into a modified web of 5-branes with 7-branes $W'$ obtained from $W$ by
``fully pushing" the 7-branes. As described in \S\ref{section_DTP}, the dual $W'$ is a Symington polygon $P$ with polyhedral decomposition $\mathcal{P}$, obtained from the lattice polygon $\overline{P}$ with polyhedral decomposition 
$\overline{\mathcal{P}}$ dual to $\overline{W}$ by a series of cutting operations.

Following \cite{alexeev2024ksba} (see also \cite{engel_friedman}), we define below a Calabi--Yau 3-fold $\mathcal{X}$ associated to $W$, with a morphism $\pi: \mathcal{X} \rightarrow \Delta$ where $\Delta \subset \CC$ is an analytic disk containing $0$, such that the central fiber $\mathcal{X}_0:= \pi^{-1}(0)$ has dual intersection complex $(P,\mathcal{P})$, that is, intersection complex given by the polyhedral decomposition of $\RR^2$ defined by the modified web $W'$. We proceed in three steps described below. 
First, for each vertex $v$ of $\mathcal{P}$,
we construct  a log Calabi--Yau surface $(X_0^v, \partial X_0^v)$.
Next, we glue these log Calabi--Yau surfaces together to form a normal crossing surface $\mathcal{X}_0$.
Finally, we construct a Calabi--Yau 3-fold $\mathcal{X}$ as the total space of a one-parameter smoothing of $\mathcal{X}_0$.

\begin{figure}[hbt!]
\center{\scalebox{0.9}{\includegraphics{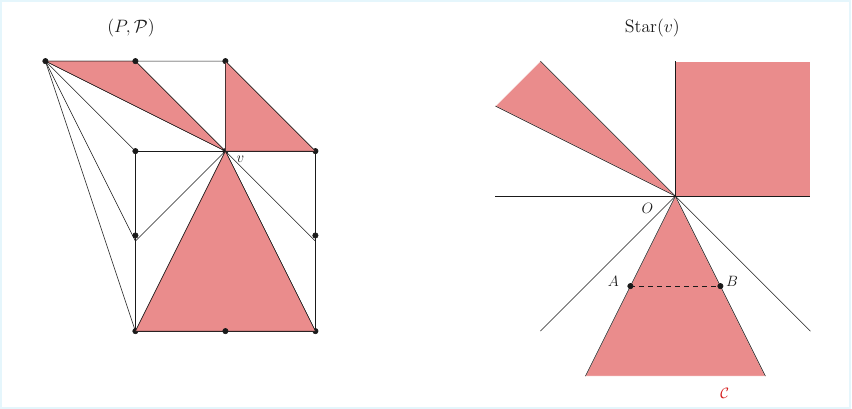}}}
\caption{On the left, the star around a vertex $v$ of the polyhedral decomposition $\mathcal{P}$ of a Symington polygon $P$ depicted on the right.}
\label{Fig:9}
\end{figure}

STEP I: \emph{Construction of the log Calabi--Yau surfaces $(X_0^v, \partial X_0^v)$.} Let $v$ be a vertex of $\mathcal{P}$. If $v$ is not a singular point of the integral affine structure on the Symington polygon $\mathcal{P}$, then the star of $\mathcal{P}$ at $v$ can be viewed as a toric fan, and we denote by  $(\overline{X}_0^v, \partial \overline{X}_0^v)$ the corresponding toric surface with its toric boundary divisor. If $v$ is a singular point of the integral affine structure on $\mathcal{P}$, then, the star of $\mathcal{P}$ at $v$ is obtained from the star of $v$ viewed as a vertex of $\overline{\mathcal{P}}$ by cutting out some cones, see Figure \ref{Fig:9}. Let $(\overline{X}_0^v, \partial \overline{X}_0^v)$
be the toric surface whose fan is the star of $\overline{\mathcal{P}}$ at $v$. Every cut cone $\mathcal{C}$ corresponds to a 0-dimensional stratum $x_{\mathcal{C}}$ of the toric boundary $\partial \overline{X}_0^v$. By construction of $(P,\mathcal{P})$ in 
\S \ref{section_symington}, every cut cone $\mathcal{C}$ is a cone over an elementary triangle. Explicitly, if $O$ is the vertex of $\mathcal{C}$, and $A$, $B$ are the primitive integral points of the two rays of $\mathcal{C}$, then the triangle $OAB$ is an elementary triangle: the lattice length of the line segment $AB$ is equal to the lattice distance of $AB$ from $O$, see Figure \ref{Fig:9}. 
Every 0-dimensional stratum $x_{\mathcal{C}}$ in $(\overline{X}_0^v, \partial \overline{X}_0^v)$ admits a neighborhood isomorphic to the affine toric surface with $(\overline{X}_\mathcal{C}, \partial{\overline{X}}_{\mathcal{C}})$
with cone $\mathcal{C}$. 
By \cite[\S 7.3]{Altmann},   $(\overline{X}_\mathcal{C}, \partial{\overline{X}}_{\mathcal{C}})$
is a quotient singularity of the form 
\[ \CC^2/(\ZZ/n^2 \ZZ) \,,\]
where the weights of the $\ZZ/n^2 \ZZ$-action are of the form $(1, an-1)$
with $\mathrm{gcd}(a,n)=1$. Moreover, the integer $n$ is the common lattice height and lattice base of the elementary triangle spanning the cone $\mathcal{C}$.

If $n=1$, then $(\overline{X}_\mathcal{C}, \partial{\overline{X}}_{\mathcal{C}})$ is simply $\CC^2_{x,y}$
with its toric boundary $xy=0$.
If $n>1$, then $(\overline{X}_\mathcal{C}, \partial{\overline{X}}_{\mathcal{C}})$ is an \emph{elementary T-singularity} in the terminology of \cite[Proposition 3.10]{KSB}, also known as a \emph{Wahl singularity}.
Using the change of variables $x=u^n$, $y=v^n$, $z=uv$, this singularity can be equivalently described as the quotient of the hypersurface 
\begin{equation}\label{eq_hyper}
xy=z^n\end{equation} in $\CC^3$ by the group $\ZZ/n \ZZ$ acting with weights $(1,-1,a)$ on $\CC^3$.
It follows from this description that $(\overline{X}_\mathcal{C}, \partial{\overline{X}}_{\mathcal{C}})$ admits a natural one-parameter deformation given by quotienting the hypersurface 
\begin{equation}\label{eq_hyper_smoothing}
xy=z^n+t\end{equation} by the action of the group $\ZZ/n \ZZ$ acting with weights $(1,-1,a)$ on $\CC^3$.
For $n=1$, this is just smoothing the normal crossing divisor $xy=0$ into $xy=t$. If $n>1$, then, both the singular toric surface $\overline{X}_\mathcal{C}$ and its toric boundary divisor $\partial \overline{X}_\mathcal{C}$  become smooth for $t \neq 0$.

According to \cite[Lemma 4.1]{alexeev2024ksba}, there is no obstruction to applying this deformation simultaneously to all neighborhoods $(\overline{X}_\mathcal{C}, \partial{\overline{X}}_{\mathcal{C}})$ of the 0-dimensional strata $x_{\mathcal{C}}$.
This yields a global deformation of the toric surface $(\overline{X}_0^v, \partial \overline{X}_0^v)$, and we denote the resulting log Calabi--Yau surface by $(X_0^v, 
\partial X_0^v)$. This surface is smooth since all the possibly singular 0-dimensional strata of $(\overline{X}_0^v, \partial \overline{X}_0^v)$ become smooth in the deformation.
Here, we adopt a slightly broader interpretation of ``log Calabi--Yau" surface than in \S\ref{section_log_open}. First, the surface  $(\overline{X}_0^v, \partial \overline{X}_0^v)$ may be non-compact, 
as occurs when $v$ lies on the boundary of $P$. 
Moreover, the boundary divisor $\partial \overline{X}_0^v$ may be smooth.
which occurs precisely when all cones in the star of $v$ have been cut out; this situation arises exactly when $L^2 = 0$ -- see Example \ref{Ex: fig 25}.

In the web picture, the vertex $v$ of $\mathcal{P}$ corresponds to a face of the polyhedral decomposition defined by $W'$. 
This face arises from a corresponding face of the polyhedral decomposition defined by $W$ through a series of  ``corner smoothings" induced by pushing-in 7-branes, as illustrated in Figure \ref{Fig:10}.

\begin{figure}[hbt!]
\center{\scalebox{.9}{\includegraphics{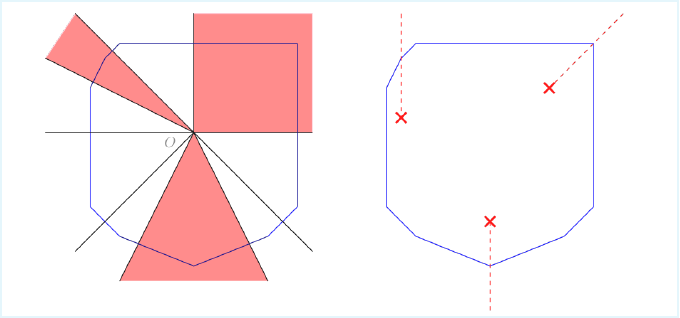}}}
\caption{The log Calabi--Yau surface 
$(X_0^v, \partial X_0^v)$ obtained by ``smoothing corners" of the toric surface $(\overline{X}_0^v, \partial \overline{X}_0^v)$.}
\label{Fig:10}
\end{figure}

When not all cones in the star of $v$ are cut out, the log Calabi--Yau surface $(X_0^v, \partial X_0^v)$ admits a more explicit description. 
The singularity of the integral affine manifold $P$ at $v$ can be viewed as several focus-focus singularities brought together, one for each cut triangle with tip at $v$. Since not all cones of the star of $v$ are cut out, these focus-focus singularities can be slightly displaced along their monodromy invariant directions, effectively smoothing the integral affine structure at $v$.
After this smoothing, the star of $\mathcal{P}$ at $v$ defines the fan of a toric variety. The log Calabi-Yau surface $(X_0^v, \partial X_0^v)$ is then obtained from this toric variety as follows:
\begin{itemize}
\item[i)] First, perform toric blow-ups so that every half-line emanating from $v$ and passing through a focus-focus singularity is incorporated into the fan.
\item[ii)] Then, for each focus-focus singularity, 
perform a non-toric blow-up at a point on the toric boundary component corresponding to the relevant half-line.
\item[iii)] Finally, contract all boundary components that do not correspond to edges of $\mathcal{P}$. 
\end{itemize}
In general, there are multiple ways to decompose the integral affine singularity at $v$ into focus-focus singularities, corresponding to distinct toric models of the log Calabi--Yau surface $(X_0^v, \partial X_0^v)$. 

When all cones in the star of $v$ are cut out, then we have $L^2=0$. As shown in Example \ref{example_L2_0}, the resulting surface $X_0^v$ is 
a smooth del Pezzo surface, and $\partial X_0^v$ is a smooth anticanonical divisor in $X_0^v$ -- see also Example \ref{Ex: fig 25}.

Step II: \emph{Construction of the normal crossing surface $\mathcal{X}_0$.} 
By Step I, a smooth log Calabi--Yau surface $(X_0^v, 
\partial X_0^v)$ has been assigned to every vertex of $\mathcal{P}$, with dual intersection given by the star of $\mathcal{P}$ at $v$. In particular, the edges of $\mathcal{P}$ adjacent to $v$ correspond to irreducible components of the boundary divisor $\partial X_0^v$. 
Gluing the surfaces $(X_0^v, 
\partial X_0^v)$ along boundary components associated with the same edges of $\mathcal{P}$ yields a singular surface $\mathcal{X}_0$, whose irreducible components are the various $(X_0^v, 
\partial X_0^v)$ and whose intersection complex is $\mathcal{P}$. 
Equivalently, the intersection complex of 
$\mathcal{X}_0$ coincides with the polyhedral decomposition determined by the web $W'$. Since each face of of $\mathcal{P}$ is a triangle, reflecting the 3-valency of $W'$, there are at most triple intersections among the irreducible components of $(X_0^v, 
\partial X_0^v)$. 
As each $X_0^v$ is smooth, it follows that $\mathcal{X}_0$ is a normal crossing surface.

Step III: \emph{Construction of the Calabi--Yau 3-fold $\mathcal{X}$.} We obtain a smooth non-compact Calabi--Yau 3-fold $\mathcal{X}$ as the total space of a one-parameter smoothing of the normal crossing surface.
Let $(Y,D,L)$ be the log Calabi--Yau surface associated to the asymptotic web of 5-branes with 7-branes $W^{\mathrm{asym}}$.
As established in \cite{engel} for the case $L \cdot D = 0$, and in \cite[Theorem 3.36]{alexeev2024ksba} for $L \cdot D > 0$, there exists a smooth non-compact Calabi--Yau 3-fold $\mathcal{X}$ and a holomorphic map 
\begin{equation}\pi: \mathcal{X} \longrightarrow \Delta\,,\end{equation} 
where $\Delta$ is an open disk in $\CC$ containing $0 \in \CC$, such that the central fiber $\pi^{-1}(0)$ is a reduced normal crossing divisor isomorphic to $\mathcal{X}_0$. 
Furthermore, it is shown in \cite[Proposition 3.14]{engel_friedman} for the case $L \cdot D = 0$, and in \cite{AAB2} for $L \cdot D > 0$, that for $t \neq 0$, the general fiber $\pi^{-1}(t)$ is diffeomorphic to $U = Y \setminus D$.
Finally, by \cite[Theorem 3.29]{alexeev2024ksba}, there exists a unique holomorphic map 
\begin{equation} f: \mathcal{X} \longrightarrow \mathcal{X}^{\mathrm{can}}\end{equation}
contracting all compact holomorphic curves of $\mathcal{X}$ to points.
The map $f$ is an isomorphism in the complement of the $\mathcal{X}_0$, and so there is an induced holomorphic map 
\begin{equation} \label{eq_pi_can}
\pi^{\mathrm{can}}: \mathcal{X}^{\mathrm{can}} \longrightarrow \Delta \,,\end{equation}
such that $\pi = \pi^{\mathrm{can}} \circ f$.
Moreover, $\mathcal{X}^{\mathrm{can}}$ is non-compact Calabi--Yau 3-fold with Gorenstein canonical singularities, and $f: \mathcal{X} \rightarrow \mathcal{X}^{\mathrm{can}}$ is a crepant resolution. 
By construction, $\mathcal{X}^{\mathrm{can}}$ is a generally non-toric deformation of the toric 3-fold singularity $\overline{\mathcal{X}}^{\mathrm{can}}$ associated 
as in \S \ref{section_toric_dual_CY3} to the 
asymptotic web of 5-branes $\overline{W}^{\mathrm{can}}$.
Moreover, the crepant resolution $f:\mathcal{X} \rightarrow \overline{X}^{\mathrm{can}}$ is a deformation of the crepant toric partial resolution $\overline{\mathcal{X}} \rightarrow \overline{\mathcal{X}}^{\mathrm{can}}$ associated to the web of 5-branes $\overline{W}$. The toric Calabi--Yau 3-fold $\overline{\mathcal{X}}$ is singular in general since $\overline{W}$ might contain 3-valent vertices of multiplicity $>1$ or 4-valent vertices. Thus, $\mathcal{X}$ is a smoothing of $\overline{\mathcal{X}}$.

It is proved in \cite{AAB2} that the 3-fold canonical singularity $\mathcal{X}^{\mathrm{can}}$ depends only, up to locally trivial deformations of its crepant resolutions, on the consistent asymptotic web $W^{\mathrm{asym}}$. Different generic consistent webs $W$ with the same asymptotic web $W^{\mathrm{asym}}$ produce different crepant resolutions $f: \mathcal{X} \rightarrow \mathcal{X}^{\mathrm{can}}$ of $\cX^{\mathrm{can}}$, which are related by flops. Actually, $\mathcal{X}^{\mathrm{can}}$ does not change when $W^{\mathrm{asym}}$ changes by Hanany--Witten moves, and only depends on the associated log Calabi--Yau surface with line bundle $(Y,D,L)$, that is, on the associated 5d SCFT. Different asymptotic webs related by Hanany--Witten moves correspond to different toric models of $(Y,D,L)$, which are mirrors to different ways to realize the same non-toric canonical 3-fold $\mathcal{X}^{\mathrm{can}}$ as deformations of different toric 3-fold singularities $\overline{\mathcal{X}}^{\mathrm{can}}$.
It follows from the reasoning in \S \ref{section_cy3_mirror} that the 5d SCFT defined by the consistent web of 5-branes with 7-branes $W^{\mathrm{asym}}$ in Type IIB string theory coincides with the 5d SCFT defined by M-theory on the canonical 3-fold singularity $\mathcal{X}^{\mathrm{can}}$.

\begin{example} \label{example_L2_0}
Assume that $L^2=0$, that is, as reviewed above Equation \eqref{eq FII0} that there exists an elliptic fibration $c: Y \rightarrow \PP^1$ with $D=c^{-1}(\infty)$ and $L=\mathcal{O}(kD)$.
In particular a general curve $C \in |L|$ consists of $k$ disjoint elliptic fibers. 
In this case, the web $W'$ obtained after fully pushing  the 7-branes does not contain any vertex, and is topologically the union of $k$ concentric circles, which can be viewed as the tropicalization of the $k$ elliptic connected components of $C$. If $k=1$, then the complement of $W'$ contains a unique bounded region, corresponding to a compact component of $\cX_0$, which is necessarily a del Pezzo surface $S$. It follows that $\mathcal{X}$ is the total space $K_S$ of the canonical divisor of $S$, that is, a ``local del Pezzos surface", and the map $\pi: \mathcal{X} \rightarrow \Delta$ is induced by a general anticanonical section, whose zero-locus is a smooth genus one curve $E$. 
If $(Y,D,L)$ corresponds to the $E_{Q-3}$ 5d SCFT, then $S$ is the corresponding del Pezzo surface obtained from $\PP^2$ by blowing up $Q-3$ points in general position. On the other hand, if $(Y,D,L)$ corresponds to the $\widetilde{E}_{1}$ 5d SCFT, then $S$ is the corresponding del Pezzo surface $\mathbb{F}_0=\PP^1 \times \PP^1$. In particular, we recover the description of $E_n$ and $\widetilde{E}_n$ 5d SCFTs as M-theory on the canonical 3-fold singularity obtained by contracting to a point the zero-section of a local del Pezzo surface \cite{morrison_seiberg}. If $k>1$, then $\mathcal{X}$ is a crepant resolution of the order $k$ base change $t \mapsto t^k$ of $K_S \rightarrow \Delta$. The base change creates a family of $A_{k-1}$-singularities parametrized by $E$, whose resolution produces $k-1$ additional irreducible components in $\mathcal{X}_0$, all isomorphic to $\PP^1$-bundles over $E$.
\end{example}

\begin{remark}
While the previous construction defines a M-theory dual canonical 3-fold singularity for every consistent asymptotic web of 5-branes with 7-branes, we do not expect the existence of a M-theory dual geometry for possibly non-consistent supersymmetric asymptotic webs. Indeed, our mirror construction crucially uses that the line bundle $L$ on the log Calabi--Yau surface $(Y,D)$ associated to the web is nef, and so defines a
large volume limit where mirror symmetry can be expected to be geometric. For a general supersymmetric web, the line bundle $L$ might fail to be nef, there is no natural large volume limit to consider, and so we do not expect to have a geometric mirror. From the point of view of engineering 5d SCFTs, this is not a restriction since, as described below Definition \ref{def: consistency}, any 5d SCFT defined by a supersymmetric web is obtained from a 5d SCFT defined by a consistent web by adding free hypermultiplets.
\end{remark}

\subsection{Calabi--Yau 3-folds and smoothings of slc singularities}
\label{section_slc_sing}

In this section, we present a characterization of canonical 3-fold singularities which are M-theory duals to webs of 5-branes with 7-branes, that is,
which appear as a result of the construction presented in \S\ref{section_cy3_construction}.

Let $W^{\mathrm{asym}}$ be a consistent asymptotic web of 5-branes with 7-branes and $\mathcal{X}^{\mathrm{can}}$. By \eqref{eq_pi_can}, the corresponding 
canonical 3-fold singularity comes with a holomorphic map $\pi^{\mathrm{can}}: \mathcal{X}^{\mathrm{can}} 
\rightarrow \CC$. The central fiber $\mathcal{X}^{\mathrm{can}}_0:= (\pi^{\mathrm{can}})^{-1}(0)$ is a surface, containing a particular point $x_0$, which is the image by $f: \cX \rightarrow \cX^{\mathrm{can}}$ of all compact surfaces and rigid compact curves in $\cX$. Using the explicit construction of $\mathcal{X}$ in \S\ref{section_cy3_construction}, we give below a complete description of the surface $\cX_0^{\mathrm{can}}$.

\begin{figure}[hbt!]
\center{\scalebox{.9}{\includegraphics{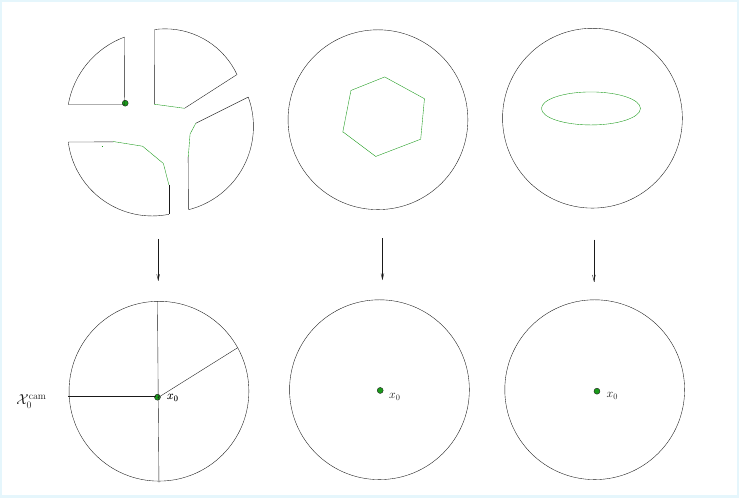}}}
\caption{The slc surface $\mathcal{X}_0^{\mathrm{can}}$ and its minimal resolution: degenerate cusp on the left, cusp in the middle, and simple elliptic singularity on the right.}
\label{Fig15}
\end{figure}

Let $(Y,D,L)$ be the log Calabi--Yau surface associated to $W^{\mathrm{asym}}$. If $L^2 >0$ and $L \cdot D>0$, that is $W^{\mathrm{asym}} \in \mathbf{CWebs}_I$, then we can consider the corresponding polarized log Calabi--Yau surface $(Y^{\mathrm{pol}}, D^{\mathrm{pol}}, L^{\mathrm{pol}})$ as in Equation \eqref{eq FI}.
The 0-dimensional strata $x_i$ of $D^{\mathrm{pol}}$ are in one-to-one correspondence with the vertices $v_i$ of a Symington polytope $P$ attached to $W^{\mathrm{asym}}$. Moreover, each 0-dimensional strata admits a neighborhood isomorphic to an affine toric surface $Y_i$, with fan given by the dual cone of the cone spanned by the edges of $P$ at the vertex $v_i$. Let $X_i$ be the dual affine toric surface, with fan given by the cone spanned by the edges of $P$ at the vertex $v_i$. Then, we have 
\[ \cX_0^{\mathrm{can}} = \bigcup_i X_i\,,\]
that is, the irreducible components of the surface
$\cX_0^{\mathrm{can}}$ are naturally indexed by the vertices of $P$, and are isomorphic to the toric surfaces $X_i$. Moreover, these surfaces are glued together along their toric divisors according to the cyclic order defined by the boundary of $P$, and all intersect together at the point $x_0$, see Figure \ref{Fig15}. Each irreducible component $X_i$ is a quotient cyclic singularity,
and so in this case $\mathcal{X}_0^{\mathrm{can}}$ is a \emph{degenerate cusp surface singularity} in the sense of \cite[\S 1]{SB2}.

If $L^2>0$ and $L \cdot D=0$, that is $W^{\mathrm{asym}} \in \mathbf{CWebs}_{II,+}$, then we can consider the corresponding Calabi--Yau surface $(Y^{\mathrm{pol}}, L^{\mathrm{pol}})$ with a cusp singularity obtained by contracting the cycle of rational curves $D$ to a point as in 
Equation \eqref{eq FII+}. 
Recall that a cusp singularity is a normal surface singularity whose exceptional divisor in the minimal resolution is a cycle of rational curves. Cusp singularities come in dual pairs \cite{nakamura} and we refer to \cite[\S 2]{looijenga}, \cite[\S 7.1]{GHK1}, \cite[\S 2]{engel} for the explicit description of this duality.
In this case, the surface $\cX_0^{\mathrm{can}}$ is normal, with an isolated singularity at the point $x_0$, which is the \emph{cusp surface singularity} dual to the cusp surface singularity of $Y^{\mathrm{pol}}$.

Finally, if $L^2=0$, that is, $W^{\mathrm{asym}}\in\mathbf{CWebs}_{II,0}$, then the surface $\mathcal{X}_0^{\mathrm{can}}$ is also normal, with an isolated singularity at the point $x_0$, which is a \emph{simple elliptic surface singularity}, that is, a normal surface singularity whose exceptional divisor in the minimal resolution is a smooth genus one curve.

Thus, we can summarize the previous discussion into the following result.

\begin{theorem} \label{thm_cusps}
Let $W^{\mathrm{asym}}$ be a consistent asymptotic web of 5-branes with 7-branes. Then, 
depending if $W^{\mathrm{asym}}$ is in
$\mathbf{CWebs}_{I}$, $\mathbf{CWebs}_{II,+}$, or $\mathbf{CWebs}_{II,0}$, the surface $\cX_0^{\mathrm{can}}$ is either a degenerate cusp, cusp, or simple elliptic singularity. Moreover, $\pi^{\mathrm{can}}: \cX^{\mathrm{can}} \rightarrow \Delta$ is a one-parameter smoothing of $\cX_0^{\mathrm{can}}=(\pi^{\mathrm{can}})^{-1}(0)$, that is, a flat family with smooth fibers
$(\pi^{\mathrm{can}})^{-1}(t)$ for $t\neq 0$, whose total space $\cX^{\mathrm{can}}$ has canonical singularities, which moreover admits semistable resolutions $f: \cX \rightarrow \cX^{\mathrm{can}}$, that is, with reduced central fiber $\cX_0:=\pi^{-1}(0)$, where $\pi=\pi^{\mathrm{can}} \circ f$.
\end{theorem}

Remarkably, the converse of Theorem \ref{thm_cusps} also holds, that is, we have the following characterization of canonical 3-fold singularities coming from webs of 5-branes with 7-branes.

\begin{theorem} \label{thm_slc}
Let $\cX^{\mathrm{can}}$ be a canonical 3-fold singularity. Then, $\cX^{\mathrm{can}}$ is the M-theory dual of a consistent asymptotic web of 5-branes with 7-branes if and only if there exists a holomorphic map $\pi^{\mathrm{can}}: \cX^{\mathrm{can}}
\rightarrow \Delta$ such that the following conditions hold:
\begin{itemize}
    \item[i)] The central fiber $\cX_0^{\mathrm{can}}:=(\pi^{\mathrm{can}})^{-1}(0)$ is either a degenerate cusp, or cusp, or simple elliptic surface singularity.
    \item[ii)] $\pi^{\mathrm{can}}: \cX^{\mathrm{can}}
\rightarrow \Delta$ is a smoothing of $\cX_0^{\mathrm{can}}$, that is, is flat and the fibers $(\pi^{\mathrm{can}})^{-1}(t)$ are smooth for $t\neq 0$.
    \item[iii)] There exists a semistable resolution $\cX \rightarrow \cX^{\mathrm{can}} \rightarrow \Delta$.
\end{itemize}
\end{theorem}

\begin{proof}
The ``only if" direction is Theorem \ref{thm_cusps}.
The ``if" direction when $\cX_0^{\mathrm{can}}$ is a degenerate cusp singularity follows from the construction of mirrors for ``open Kulikov degenerations" in \cite{alexeev2024ksba}. 
The result follows along similar lines when $\cX_0^{\mathrm{can}}$ is a cusp singularity -- see \cite{AAB2} for details. Finally, when $\cX_0^{\mathrm{can}}$ is a simple elliptic singularity, the result follows from the relation between smoothings of simple elliptic singularities and smooth anticanonical divisors on del Pezzo surfaces established in \cite{looijenga_wahl, pinkham}.
\end{proof}

Degenerate cusp, cusp and simple elliptic singularities are examples of \emph{slc} (\emph{semi-log-canonical}) surface singularities as introduced in \cite[Definition 4.17]{KSB}. The notion of slc surface singularity is closely related to the notion of canonical 3-fold singularity. Indeed, if $\mathcal{Z}_0$ is surface singularity with a smoothing $\mathcal{Z} \rightarrow \Delta$ that admits a semi-stable resolution, then, by Theorem \cite[Theorem 5.1]{KSB}, the surface $\mathcal{Z}_0$ is slc if and only if the 3-fold $\mathcal{Z}$ has canonical singularities. In other words, slc surface singularities are the surface singularities whose one-parameter smoothings produce canonical 3-fold singularities.  
According to \cite[Theorem 4.21]{KSB}, degenerate cusp, cusp, and simple elliptic singularities are the three main classes of Gorenstein slc singularities, that is, with an invertible canonical sheaf. General slc singularities are only $\QQ$-Gorenstein in general, that is, only a positive power of their canonical sheaf is a line bundle.
The non-Gorenstein slc singularities are classified in \cite[Theorem 4.24]{KSB}: they are either $\ZZ/2\ZZ$, $\ZZ/3\ZZ$, $\ZZ/4\ZZ$ or $\ZZ/6\ZZ$ quotients of simple elliptic singularities, or $\ZZ/2\ZZ$ quotients of cusp and degenerate cusp singularities. We expect the canonical 3-fold singularities defined by one-parameter smoothings 
of these non-Gorentsein slc singularities to be M-theory duals of webs of 5-branes with 7-branes and additional orientifolds and S-folds, as in \cite{O7_planes, O5_planes} and \cite{acharya,S_fold} respectively, but we leave the study of this correspondence to future work.

\subsection{Calabi--Yau 3-folds from scattering diagrams}
\label{section_scattering}

Let $W^{\mathrm{asym}}$ be a consistent web of 5-branes with 7-branes. We described in \S\ref{section_cy3_construction} a construction of the M-theory dual canonical 3-fold singularity $\mathcal{X}^{\mathrm{can}}$ through an explicit construction of the central fiber $\cX_0$ of a crepant resolution $\cX \rightarrow \cX^{\mathrm{can}}$. In particular, the 3-fold $\cX^{\mathrm{can}}$ is obtained indirectly, by first using deformation theory to produce $\cX$ from $\cX_0$, and then a general contraction result to produce $\cX^{\mathrm{can}}$ from $\cX$. In this section, we address the question to give a more concrete description of $\cX^{\mathrm{can}}$ by explicit algebraic equations.

Let $(Y,D,L)$ be the log Calabi--Yau surface with line bundle
associated to $W^{\mathrm{asym}}$. We show in \cite{AAB2} that $\cX^{\mathrm{can}}$ is the total space of a one-parameter subfamily of the \emph{intrinsic} mirror family of $(Y,D,L)$, as described
for log Calabi--Yau surfaces in \cite{GHK1} and for higher-dimensional log Calabi--Yau varieties in  \cite{gross2019intrinsic, GScanonical}.  
It follows that $\cX^{\mathrm{can}}$ can be calculated combinatorially in terms of algebraic structures called \emph{scattering diagrams}
and originally introduced in 
\cite{GSreal, KSaffine}. 
Below is a brief overview of this construction tailored for our purposes. We refer to \cite{arguz_equations, B_explicit, GHK1, GHKScubic} for detailed expositions and additional examples. 
We discuss the enumerative and physics meaning of scattering diagrams in \S\ref{section_instanton_BPS}.

\begin{definition}
An \emph{incoming wall} is a pair $(\fod, f_\fod)$, where:
\begin{itemize}
    \item[i)] $\fod$ is a ray of rational slope through the origin, that is, $\fod=\RR_{\geq 0}(p,q)$ for coprime $(p,q)\in \ZZ^2\setminus\{0\}$.   
    \item[ii)] $f_\fod \in \ZZ[x^p y^q][\![t]\!]$ is a power series in the monomial $x^p y^q$ with coefficients polynomials in $t$ with integral coefficients.
\end{itemize}
\end{definition}

\begin{definition}
An \emph{outcoming wall} is a pair $(\fod, f_\fod)$, where:
\begin{itemize}
    \item[i)] $\fod$ is a ray of rational slope through the origin, that is, $\fod=\RR_{\geq 0}(p,q)$ for coprime $(p,q)\in \ZZ^2\setminus\{0\}$.   
    \item[ii)] $f_\fod \in \ZZ[x^{-p} y^{-q}][\![t]\!]$ is a power series in the monomial $x^{-p} y^{-q}$ with coefficients given by polynomials in $t$ with integral coefficients.
\end{itemize}
\end{definition}

\begin{definition}
The \emph{wall-crossing automorphism} of a wall $(\fod, f_\fod)$, with $\fod=\RR_{\geq 0}(p,q)$
and coprime $(p,q)\in \ZZ^2\setminus\{0\}$, is the $\CC[\![t]\!]$-algebra automorphism
\begin{align*} \Phi_{(\fod, f_\fod)}:\ZZ[x^{\pm}, y^{\pm}][\![t]\!] \longrightarrow \ZZ[x^{\pm}, y^{\pm}][\![t]\!] \\
x^a y^b \longmapsto x^a y^b f_\fod^{pb-qa} \,.
\end{align*}
\end{definition}

\begin{definition} \label{def_scattering_diagram}
A \emph{scattering diagram} $\mathfrak{D}$ is a set of incoming or outgoing walls $(\fod, f_\fod)$ such that, for every $k \in \ZZ_{\geq 1}$, there are finitely many walls $(\fod, f_\fod)$ with $f_\fod \neq 1 \mod t^k$.
\end{definition}

We consider two scattering diagrams as equivalent if they are related by a series of walls mergings, that is, of identifications of pairs of walls $(\fod, f_1)$, $(\fod, f_2)$ supported on the same ray $\fod$ with the single wall $(\fod, f_1 f_2)$.

\begin{definition}
The \emph{total wall-crossing automorphism} $\Phi_{\mathfrak{D}}$ of a scattering diagram 
$\mathfrak{D}$ is the composition of all the wall-crossing automorphisms $\Phi_{(\fod, f_\fod)}$  of the walls $(\fod, f_\fod)$
in $\mathfrak{D}$, where the walls are ordered in the anticlockwise direction around the origin.
A scattering diagram $\mathfrak{D}$ is called \emph{consistent} if its total wall-crossing automorphism
is the identity: $\Phi_{\mathfrak{D}}=\mathrm{id}$.
\end{definition}

Given an asymptotic web of 5-branes with 7-branes 
$W^{\mathrm{asym}}$, defined by a web of 5-branes $\overline{W}$ and the 7-brane data $\mathbf{a} := ((a_{e,i})_{1\leq i \leq n_e})_{E \in L(\overline{W})}$, we define a corresponding \emph{initial scattering diagram} $\mathfrak{D}_{\mathrm{in}}$ as the set of incoming walls $(\fod_e, f_{\fod_e})_{e\in L(\overline{W})}$, where, denoting by $(p_e,q_e)\in \ZZ^2$ the primitive integral generator of the leg $e$ pointing away from the origin, we have 
\[ \fod_e=e=\RR_{\geq 0}(p_e,q_e)\] and 
\begin{equation} \label{eq_wall_function}
    f_{\fod_e}=t^{w_e} \prod_{i=1}^{n_e}
    (1+t^{-a_{e,i}} x^{p_e}y^{q_e}) 
    = t^{w_e-\sum_{i=1}^{n_e} a_{e,i}}\prod_{i=1}^{n_e}
    (t^{a_{e,i}}+ x^{p_e}y^{q_e}) \,.
\end{equation}
By Definition \ref{def_web_57_at_infinity} ii), we have $\sum_{i=1}^{n_e}a_{e,i} \leq w_e$, and so $f_{\fod_e}$ is indeed a polynomial in $t$.

The initial scattering diagram is not consistent in general. However, if the asymptotic web of 5-branes with 7-branes $W^{\mathrm{asym}}$ is consistent, it follows from \cite{GHK1} that there exists a unique consistent scattering diagram $\mathfrak{D}$
obtained from $\mathfrak{D}_{\mathrm{in}}$ by adding outgoing walls. Moreover, $\mathfrak{D}$
can be algorithmically computed-- see \cite[\S 3.4]{GHK1}, \cite{GPS}
and \cite{arguz_equations, HDTV} in higher dimensions.

\begin{remark}
    The main object of study in \cite{GHK1} is the ``canonical scattering diagram", which is contained in the integral affine manifold with singularity obtained by moving all the 7-branes at the origin of $\RR^2$. By contrast, the scattering diagram $\mathcal{D}$ is contained in $\RR^2$, with the 7-branes ``located at infinity", and is the ``heart" of the canonical scattering diagram in the terminology of \cite{arguz_equations}. The scattering diagram and its heart are related by moving the 7-branes to infinity, as described in \cite{arguz_equations, HDTV, GHK1}.
\end{remark}

We now review the notion of a broken line in a scattering diagram \cite{GHK1, GHS}.

\begin{definition}
Let $\mathfrak{D}$ be the consistent scattering diagram associated to a consistent asymptotic web of 5-branes with 7-branes $W^{\mathrm{asym}}$. 
A \emph{broken line} in $\mathfrak{D}$
is a piecewise linear continuous directed path
\[ \beta: (-\infty,0] \longrightarrow \RR^2 \setminus \{0\}\,,\]
with $\beta(0) \notin \cup_{(\fod,f_\fod)}\fod$ and whose image consists of finitely many line segments $L_1, L_2, \dots , L_N$, such that no $L_i$ is contained in a wall of $\fod$ and each $L_i$ is compact except $L_1$. Further, we require that each $L_i$ is contained in a 2-dimensional cone of the fan defined by $\overline{W}$. 
Moreover, to each $L_i$ is assigned a monomial $\mu_i = a_i t^{k_i} x^{m_i} y^{n_i}$, where $a_i\in \ZZ$, $k_i\in \ZZ_{\geq 0}$, $(m_i,n_i)\in \ZZ^2 \setminus \{0\}$ with
$\beta'(t)=-(m_i,n_i)$ on $L_i$. We require 
$a_1=1$ and $k_1=0$, and we refer to $(m_1,n_1)$ as the asymptotic direction of the broken line. 
Moreover, for every $1\leq i\leq n-1$, we require that $L_i\cap L_{i+1}$ is contained in a wall $\fod =\RR_{\geq 0}(p,q)$ of $\mathfrak{D}$, with coprime $(p,q)\in \ZZ^2 \setminus\{0\}$, and that $m_{i+1}$ is a monomial in $\mu_i f_\fod^{|pn_i-qm_i|}$. We refer to $\beta(0)$ as the end-point
of $\beta$, and $\mu_N$ at the final monomial of $\beta$.
\end{definition}

Broken lines define an associative commutative algebra $\CC[\![t]\!]$-algebra $\mathcal{A}$ as follows \cite{GHK1, GHS}. As a $\CC[\![t]\!]$-module, $\mathcal{A}$ has a basis $\{\vartheta_{m,n}\}_{(m,n)\in \ZZ^2}$ indexed by $\ZZ^2$, and the product is defined by 
\begin{equation}
\label{eq_product}
\vartheta_{m,n}\vartheta_{m',n'}=\sum_{\beta,\beta'} a_{\beta}a_{\beta'} t^{k_{\beta}+k_{\beta'}} \vartheta_{m'',n''}\,,\end{equation}
where the sum is over pairs of broken lines $\beta$, $\beta'$ with asymptotic directions $(m,n)$, $(m',n')$, common end-point $\beta(0)=\beta'(0)$ close enough to $(m'',n'')$, and final monomials $a_\beta t^{\beta} x^{m''}y^{n''}$ and $a_{\beta'} t^{\beta'} x^{m''}y^{n''}$ respectively. 
The basis elements $\vartheta_{m,n}$ are referred to as \emph{theta functions}, and the algebra $\mathcal{A}$ is called the \emph{algebra of theta functions} determined by the consistent scattering diagram $\mathfrak{D}$.
According to \cite{GHK1}, the algebra of theta functions $\mathcal{A}$ is the algebra of regular functions on the formal completion of $\mathcal{X}^{\mathrm{can}}$ along $\mathcal{X}^{\mathrm{can}}_0$. In particular,  the algebra of theta functions $\mathcal{A}$ gives an explicit description of the deformation of $\mathcal{X}^{\mathrm{can}}_0$ in $\mathcal{X}^{\mathrm{can}}$.

By \cite{GHK1, GHS}, for every point $p \in \RR^2 \setminus \cup_{(\fod,f_{\fod})\in \mathfrak{D}} \fod$, we have an embedding of $\CC[\![t]\!]$-algebras 
\begin{align} \label{eq_embedding}
    &\mathcal{A} \hooklongrightarrow \CC[x^{\pm}, y^{\pm}][\![t]\!] \\ \nonumber
    &\vartheta_{m,n} \longmapsto \sum_{\beta} a_\beta t^{\beta} x^{m_\beta} y^{n_\beta}\,,
\end{align} 
where the sum is over broken lines $\beta$ with asymptotic direction $(m,n)$, end-point $\beta(0)=p$, and final monomial $a_\beta t^{\beta} x^{m_\beta} y^{n_\beta}$. To describe the algebra of theta functions $\mathcal{A}$ by generators and relations, using the embeddings given by the Equation \eqref{eq_embedding} is often more convenient than calculating the structure constants in the basis $\{\vartheta_{m,n}\}$
using Equation \eqref{eq_product}.

\begin{example} \label{example_comparison}
Let $\overline{W}^{\mathrm{asym}}$ be the asymptotic web of 5-branes consisting of $w$ coincident parallel 5-branes on the horizontal axis $\RR(1,0)$.
Let $W^{\mathrm{asym}}$ be the asymptotic web of 5-branes with 7-branes obtained from $\overline{W}$ by ending $a_1, \dots, a_n$ 5-branes on 7-branes $x_1, \dots, x_n$ on $\RR_{>0}(1,0)$. 
By Equation \eqref{eq_wall_function}, the initial scattering diagram $\mathfrak{D}_{\mathrm{in}}$
consists of the incoming wall 
\[ (\RR_{\geq 0}(1,0), t^w \prod_{i=1}^{n}
    (1+t^{-a_i} x)) \,,\]
and $(\RR_{\geq 0}(-1,0), t^w)$.
The consistent scattering diagram $\mathfrak{D}$ is obtained from $\mathfrak{D}_{\mathrm{in}}$ by adding the outgoing wall 
$(\RR_{\geq 0}(-1,0), \prod_{i=1}^{n}
    (1+t^{-a_{e,i}} x))$, or equivalently by 
replacing $(\RR_{\geq 0}(-1,0), t^w)$ by
\[ (\RR_{\geq 0}(-1,0), t^{w-\sum_{i=1}^{n} a_{i}}\prod_{i=1}^{n}
    (t^{a_{i}}+ x)) \,.\]
Pick a point $p$ in the upper half-plane, and consider the corresponding embedding of the algebra of theta functions given by Equation \eqref{eq_embedding}.
We have $\vartheta_{1,0}=x$, $\vartheta_{0,1}=y$ and 
\[\vartheta_{0,-1}=y^{-1} t^{w-\sum_{i=1}^{n} a_{i}}\prod_{i=1}^{n}
    (t^{a_{i}}+ x)\,,\]
and so the defining equation of 
the algebra of theta functions is 
\begin{equation} \label{eq_theta_ex}
    \vartheta_{0,1}\vartheta_{0,-1}
    =t^{w-\sum_{i=1}^{n} a_{i}}\prod_{i=1}^{n}
    (t^{a_{i}}+ \vartheta_{1,0})\,.
\end{equation}
This recovers the description given in \cite[(4.4)]{bourget2023generalized} of the Calabi--Yau 3-fold M-theory dual to a web of 5-branes with 7-branes where all 7-branes have parallel monodromy-invariant directions, that is, when all the white dots are introduced on a single edge of the dual GTP. 
In \cite[(4.4)]{bourget2023generalized}, all 5-branes are assumed to end on a 7-brane, so that $n=\sum_{i=1}^n a_i$, and so the factor $t^{w-\sum_{i=1}^{n}a_i}$ in Equation \eqref{eq_theta_ex} does not appear explicitly. In \cite{bourget2023generalized}, the result is obtained from a D6-brane description in a Type IIA string duality frame, which is only possible in this particular situation where all 7-branes have parallel monodromy-invariant directions.
\end{example}

\begin{figure}[hbt!]
\center{\scalebox{0.9}{\includegraphics{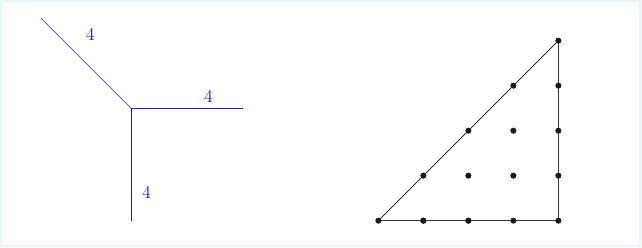}}}
\caption{Web of 5-branes and its dual polygon in Example \ref{example_triangle_toric}.}
\label{Fig11}
\end{figure}

\begin{example} \label{example_triangle_toric}
Let $\overline{W}^{\mathrm{asym}}$ be the asymptotic web of 5-branes given by the rays $\RR_{\geq 0}(1,0)$, $\RR_{\geq 0}(0,-1)$ and $\RR_{\geq 0}(-1,1)$ all endowed with weight $4$. The corresponding polarized toric log Calabi--Yau surface is $(\PP^2, \partial \PP^2, \mathcal{O}(4))$, with momentum polytope as in Figure \ref{Fig11}.
In this toric situation, the initial scattering diagram consisting of these three rays all endowed with the functions $t^4$ is already consistent. Choosing a point $p$ in the upper-right 2-dimensional cone of the polyhedral decomposition of $\RR^2$ defined by $\overline{W}^{\mathrm{asym}}$, the corresponding theta functions are 
$\vartheta_{1,0}=x$, $\vartheta_{-1,1}=x^{-1}y$, and $\vartheta_{0,-1}=y^{-1} t^4$. It follows that $\overline{\mathcal{X}}^{\mathrm{can}}$ is the toric Calabi--Yau 3-fold defined by the equation 
\begin{equation}\vartheta_{1,0}\,\vartheta_{0,-1}\,\vartheta_{-1,1}
=t^4 \,.\end{equation}
\end{example}

\begin{example} \label{example_A1}
    Let $\overline{W}^{\mathrm{asym}}$ be as in Example \ref{example_triangle_toric}, and let $W$ be the web of 5-branes with 7-branes obtained by ending two of the four $(1,0)$ 5-branes on a 7-brane on $\RR_{>0}(1,0)$. This web is consistent, see Figure \ref{Fig12} for a corresponding Symington polygon.
    The corresponding log Calabi--Yau surface $(Y,D,L)$ consists of the surface $Y$ obtained from $\PP^2$ by a non-toric blow-up of a point on the toric boundary divisor, the strict transform $D$ of the toric boundary divisor of $\PP^2$, and the ample line bundle $L=(\pi^{\star} \mathcal{O}(4))(-2E)$, where $\pi: Y \rightarrow \PP^2$ is the blow-up morphism and $E$ is the exceptional curve.
    By Equation \eqref{eq_wall_function}, the initial scattering diagram $\mathfrak{D}_{\mathrm{in}}$ consists of the three incoming walls $(\RR_{\geq 0}(1,0), t^4(1+t^{-2}x))$, $(\RR_{\geq 0}(0,-1), t^4)$, and $(\RR_{\geq 0}(-1,1), t^4)$. The consistent scattering diagram $\mathfrak{D}$ is obtained from $\mathfrak{D}_{\mathrm{in}}$ by adding the outgoing wall $(\RR_{\geq 0}(-1,0), 1+t^2 x)$.
    Choosing a point $p$ as in Figure \ref{Fig12}, the corresponding theta functions are $\vartheta_{1,0}=x$, $\vartheta_{-1,1}=x^{-1}y$, and $\vartheta_{0,-1}=y^{-1} t^4(1+t^{-2}x)$. Hence, $\mathcal{X}^{\mathrm{can}}$ is the  Calabi--Yau 3-fold defined by the equation 
\begin{equation}\vartheta_{1,0}\,\vartheta_{0,-1}\,\vartheta_{-1,1}
=t^2(t^2+\vartheta_1) \,.\end{equation}
This recovers the middle equation in \cite[Figure 5]{bourget2023generalized}.
\end{example}

\begin{figure}[hbt!]
\center{\scalebox{0.9}{\includegraphics{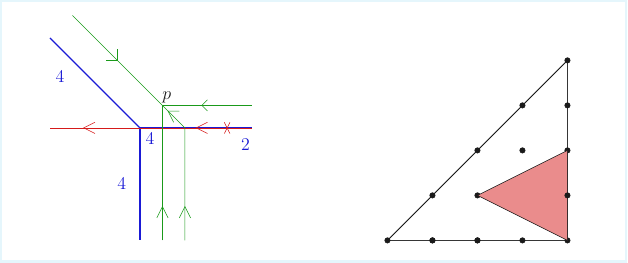}}}
\caption{Scattering diagram and Symington polygon in Example \ref{example_A1}}
\label{Fig12}
\end{figure}

\begin{figure}[hbt!]
\center{\scalebox{0.9}{\includegraphics{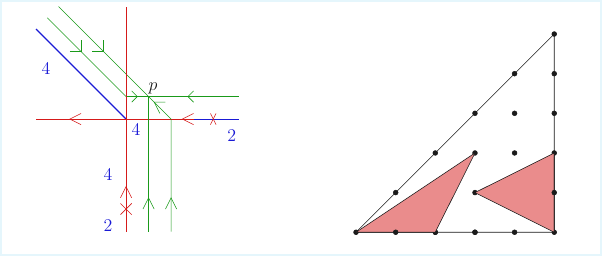}}}
\caption{Scattering diagram and Symington polygon in Example \ref{example_A2}.}
\label{Fig13}
\end{figure}

\begin{example} \label{example_A2}
Let $\overline{W}^{\mathrm{asym}}$ be the asymptotic web of 5-branes given by the rays $\RR_{\geq 0}(1,0)$, $\RR_{\geq 0}(0,-1)$ and $\RR_{\geq 0}(-1,1)$ all endowed with weight $5$.  
Let $W$ be the web of 5-branes with 7-branes obtained by ending two of the five $(1,0)$ 5-branes on a 7-brane on $\RR_{>0}(1,0)$, and two of the five $(0,-1)$ 5-branes on a 7-brane on $\RR_{>0}(0,-1)$. This web is consistent, see Figure \ref{Fig13} for a corresponding Symington polygon.
The corresponding log Calabi--Yau surface $(Y,D,L)$ consists of the surface $Y$ obtained from $\PP^2$ by the non-toric blow-up of two points on the toric boundary divisor, the strict transform $D$ of the toric boundary divisor of $\PP^2$, and the ample line bundle $L=(\pi^{\star} \mathcal{O}(5))(-2E_1-2E_2)$, where $\pi: Y \rightarrow \PP^2$ is the blow-up morphism and $E_1$, $E_2$ are the exceptional curves.
    
By Equation \eqref{eq_wall_function}, the initial scattering diagram $\mathfrak{D}_{\mathrm{in}}$ consists of the three incoming walls $(\RR_{\geq 0}(1,0), t^5(1+t^{-2}x))$, $(\RR_{\geq 0}(0,-1), t^5(1+t^{-2}y^{-1}))$, and $(\RR_{\geq 0}(-1,1), t^5)$. 
    By \cite[Figure 1.2]{GPS}, the consistent scattering diagram $\mathfrak{D}$ is obtained from $\mathfrak{D}_{\mathrm{in}}$ by adding the outgoing walls $(\RR_{\geq 0}(-1,0), 1+t^3 x)$,
    $(\RR_{\geq 0}(0,1), 1+t^3 y^{-1})$ and $(\RR_{\geq 0}(-1,1),1+tx^{-1}y )$.
    Choosing a point $p$ as in Figure \ref{Fig13}, the corresponding theta functions are $\vartheta_{1,0}=x$, $\vartheta_{-1,1}=x^{-1}y(1+t^3 y^{-1})$, and $\vartheta_{0,-1}=y^{-1} t^5(1+t^{-2}x)$. 
    We deduce that $y^{-1}=t^{-5}\frac{\vartheta_{0,1}}{1+t^{-2}\vartheta_{1,0}}$, and
    so that $\mathcal{X}^{\mathrm{can}}$ is the  Calabi--Yau 3-fold defined by the equation 
\begin{align*}\vartheta_{1,0}\,\vartheta_{0,-1}\,\vartheta_{-1,1}
&=t^5(1+t^3 y^{-1})(1+t^{-2}x)\\
&=t^5 \left(1+t^{-2} \frac{\vartheta_{0,-1}}{1+t^{-2}\vartheta_{1,0}}\right)(1+t^{-2}\vartheta_{1,0})\,,\end{align*}
that is,
\begin{equation}
\vartheta_{1,0}\,\vartheta_{0,-1}\,\vartheta_{-1,1} =t^5+t^3 \vartheta_{0,-1}+t^3 \vartheta_{1,0}\,.
\end{equation}
This result cannot be derived using the techniques of \cite{bourget2023generalized} since the two 7-branes have non-parallel monodromy directions (in the language of GTPs, there are white dots on two distinct edges).
\end{example}

\begin{example} \label{example_cubic}
    Let $\overline{W}^{\mathrm{asym}}$ be the asymptotic web of 5-branes given by the rays $\RR_{\geq 0}(1,0)$, $\RR_{\geq 0}(0,-1)$ and $\RR_{\geq 0}(-1,1)$ all endowed with weight $6$.  
Let $W$ be the web of 5-branes with 7-branes obtained by ending two pairs of $(1,0)$ 5-branes on two 7-branes on $\RR_{>0}(1,0)$, two pairs of $(0,-1)$ 5-branes on two 7-branes on $\RR_{>0}(0,-1)$, and two pairs of $(-1,1)$ 5-branes on two 7-branes on $\RR_{>0}(-1,1)$. This web is consistent, see Figure \ref{Fig14} for a corresponding Symington polygon.
Denote by $(Y,D,L)$ the corresponding log Calabi--Yau surface with line bundle. The surface $Y$
is obtained from $\PP^2$ by blowing up non-torically six points, two on each of the three toric lines, and so is a cubic surface. The divisor $D$ is the strict transform of the three toric lines in $\PP^2$. Moreover, we have $L=(\pi^\star \mathcal{O}(6))(-\sum_{i=1}^6 2E_i)$, where $\pi: Y \rightarrow \PP^2$ is the blow-up morphism, and $(E_i)_{1\leq i\leq 6}$ are the six exceptional curves.

By Equation \eqref{eq_wall_function}, the initial scattering diagram $\mathfrak{D}_{\mathrm{in}}$ consists of the three incoming walls $(\RR_{\geq 0}(1,0), t^6(1+t^{-2}x)^2)$, $(\RR_{\geq 0}(0,-1), t^6(1+t^{-2}y^{-1})^2)$, and $(\RR_{\geq 0}(-1,1), t^6(1+t^{-2}x^{-1}y)^2)$. 
By \cite{GHKScubic}, see also \cite{bousseau_skein}, the consistent scattering diagram $\mathfrak{D}$ is obtained from $\mathfrak{D}_{\mathrm{in}}$ by adding infinitely many rays of every possible rational slope. 
Moreover, the function attached to each ray can be explicitly determined.
For instance,
by \cite[Proposition 2.4]{GHKScubic}, the function $f_\fod$ attached to the added outgoing ray 
$\RR(1,0)$ is 
\[ f_\fod=\frac{(1+t^2 x^{-1})^8}{(1-t^4 x^{-2})^4}\,,\]
and the other added rays are determined by an $SL(2,\ZZ)$ symmetry described in 
\cite[Theorem 2.5]{GHKScubic}.
Even though the consistent scattering diagram contains infinitely many rays, generators and relations for the algebra of theta functions can still be determined, and it follows from \cite[Theorem 0.1]{GHKScubic}-\cite[Eq. (79)]{bousseau_skein} that $\mathcal{X}^{\mathrm{can}}$ is the Calabi--Yau 3-fold defined by the equation
\begin{equation} \vartheta_{1,0}\vartheta_{0,-1}\vartheta_{-1,1}
=t^2(\vartheta_{1,0}^2+\vartheta_{0,-1}^2
+\vartheta_{-1,1}^2)+8t^4(\vartheta_{1,0}+\vartheta_{0,-1}
+\vartheta_{-1,1})+28t^6 \,.\end{equation}
\end{example}

\begin{remark}
For a sufficiently complex web of 5-branes with 7-branes, the consistent scattering diagram $\mathfrak{D}$ can be extremely intricate. It will include every ray of rational slope, and finding a closed form for the function associated with each ray will usually be very challenging. In particular, it is in general difficult to calculate explicit equations for $\mathcal{X}^{\mathrm{can}}$ in such a situation. 
However, there are examples, such as the mirror of a degree two del Pezzo surface $dP_2$ in \cite{B_explicit}, where explicit equations can be written, even though the consistent scattering diagram is arbitrarily complicated and not fully known explicitly.
\end{remark}

\begin{figure}[hbt!]
\center{\scalebox{0.9}{\includegraphics{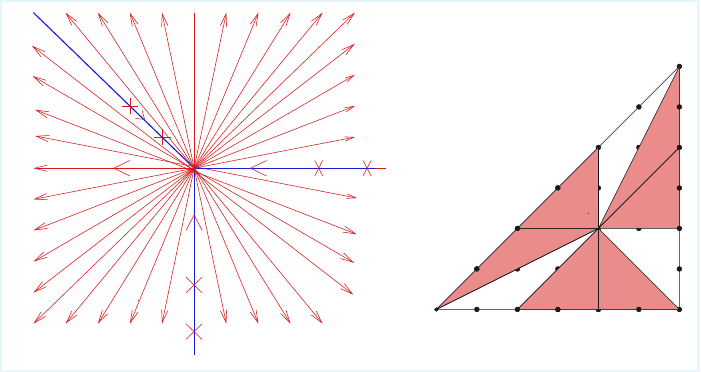}}}
\caption{Scattering diagram and Symington polygon in Example \ref{example_cubic}.}
\label{Fig14}
\end{figure}

\subsection{Disk worldsheet instantons and BPS states}
\label{section_instanton_BPS}
Let $W^{\mathrm{asym}}$ be a consistent asymptotic web of 5-branes with 7-branes and $\mathfrak{D}$ the 
corresponding consistent scattering diagram as in \S\ref{section_scattering}. Let $(Y,D,L)$
be the corresponding log Calabi--Yau surface with line bundle, obtained as in \S\ref{section_webs_log_CY} as non-toric blow-up of a polarized toric surface $(\overline{Y},\overline{D},\overline{L})$.
By \cite{GHK1}, the outgoing walls $(\fod,f_\fod)$ in $\mathfrak{D}$ added to the initial scattering diagram $\mathfrak{D}_{\mathfrak{in}}$ have an enumerative interpretation in terms of log Gromov--Witten counts of rational curves in the log Calabi--Yau surface $(Y,D)$. For every coprime $(p,q)\in \ZZ^2 \setminus \{0\}$, $k\in \ZZ_{\geq 1}$, and $\beta \in H_2(Y,\ZZ)$, one can define a log Gromov--Witten count \[N_{(kp,kq),\beta}
\in \QQ\]
of marked rational curves $f:C \rightarrow Y$
of class $\beta$, intersecting $D$ at a unique point and with contact order $(kp,kq)$ along $D$, that is, whose strict transform has contact order $k$ with the exceptional divisor in the surface obtained from $Y$ by the corner blow-up corresponding to adding the ray $\RR_{\geq 0}(p,q)$ in the fan of the toric surface $\overline{Y}$.
Then, the function $f_\fod$ associated to the outgoing ray $\fod=\RR_{\geq 0}(p,q)$
in $\mathfrak{D}$ is given by
\begin{equation}\label{eq_enum}
f_\fod = \exp \left( \sum_{k\geq 1}\sum_{\beta\in H_2(X,\ZZ)} k N_{(kp,kq),\beta} t^{L\cdot \beta} x^{-kp}  
y^{-kq}\right)\,.\end{equation}
Since the web $W^{\mathrm{asym}}$ is consistent, the line bundle $L$ is nef by Lemma \ref{lem_L_nef}, and so we have $L\cdot \beta\geq 0$
for every curve class $\beta$ represented by an algebraic curve, and so in particular for every $\beta$ such that $N_{(kp,kq),\beta} \neq 0$. In particular, all powers of $t$ in Equation \eqref{eq_enum} are indeed non-negative.

Recall from \S\ref{section_string_open_CY} that the open Calabi--Yau surface $U=Y \setminus D$ is a $T^2$-fibration over the integral affine manifold with singularity $B$ defined by the 7-branes, with a singular fiber over each $(p_i,q_i)$ 7-brane obtained by pinching a 1-cycle of class $(p_i,q_i)\in H_1(T^2,\ZZ)=\ZZ^2$ to a point. Then, $N_{(kp,kq),\beta}$ should be viewed as an algebro-geometric definition of a count of holomorphic disks in $U$ with boundary of class $(kp,kq)\in H_1(T^2,\ZZ)=\ZZ^2$ in a $T^2$-fiber. From the physics point of view, $N_{(kp,kq),\beta}$ is a count of worldsheet instantons in the 2-dimensional A-model with target $U$ and with boundary condition defined by a $T^2$-fiber. 
The walls of the initial scattering diagram $\mathfrak{D}_{\mathrm{in}}$
should be viewed as tropicalizations of elementary holomorphic disks created by the vanishing cycles of the singular fibers above the 7-branes. Moreover, the walls of the consistent scattering diagrams should be viewed as tropicalizations of 
more complicated disks obtained by gluing these elementary disks together \cite{bardwell2021scattering, GHK1, GPS}. 

Finally, these holomorphic disks have another interpretation as counts of M2-branes ending on an M5-brane wrapping a $T^2$-fiber in $U$, and so as BPS states in the 4d $\mathcal{N}=2$ theory on the $\RR^4$ part of the worldvolume of the M5-brane. 
This 4d $\mathcal{N}=2$ theory is in general distinct from the 4d $\mathcal{N}=2$ theory obtained by compactifying the 5d SCFT on $S^1$, which, as reviewed in \S\ref{section_consistent}, describes the $\RR^4$ worldvolume of an M5-brane wrapped on the curve $C^\circ$ in $U$. 
The 4d $\mathcal{N}=2$ theory defined by an M5-brane on $T^2 \subset U$ is of rank one, with Coulomb branch identified with the base $B$ of the $T^2$-fibration. Equivalently, it is the worldvolume theory on a D3-brane probing the configuration of 7-branes in Type IIB string theory, and the tropicalizations of the holomorphic disks are string junctions that realize the BPS states \cite{MNS}.

The precise relation between the consistent scattering diagram $\mathfrak{D}$ and the BPS spectrum of the 4d $\mathcal{N}=2$ theory defined by an M5-brane on $T^2 \subset U$ can be described as follows.
By Theorem 
\cite[Theorem 8.5]{bousseau_vertex}, the function $f_\fod$ associated to the outgoing ray $\fod=\RR_{\geq 0}(p,q)$ in $\mathfrak{D}$ can be uniquely written as 
\[ f_\fod=\prod_{k\geq 1}\prod_{\beta \in H_2(Y,\ZZ)} \left( 1-(-1)^k t^{L\cdot \beta} x^{-kp}  
y^{-kq}\right)^{(-1)^{k-1} k \Omega_{(kp,kq),\beta}}\,,\]
with $\Omega_{(kp,kq),\beta} \in \ZZ$. Then, the integers $\Omega_{(kp,kq),\beta}$ are the BPS indices
\cite[Eq. (1.1)]{GMN1}
near infinity of the Coulomb branch of the 4d $\mathcal{N}=2$ theory defined by an M5-brane on $T^2 \subset U$, where $(kp,kq)\in \ZZ^2$ is the electromagnetic charge with respect to the low-energy $U(1)$ gauge theory, and $\beta$ encodes the flavor charge. The consistency of 
the scattering diagram is a manifestation of the Kontsevich--Soibelman wall-crossing formula for BPS indices \cite{kontsevich2008stability}, as explained in \cite{arguz2024quivers, bousseau_scattering, bousseau_bps, bridgeland}. In particular, the fact that the consistent scattering diagram can be arbitrarily complicated is related to the inherent complexity of the BPS spectrum of 4d $\mathcal{N}=2$ theories, see \cite{galakhov2013wild} for instance.

\begin{example} \label{example_4d_rk1}
The 4d $\mathcal{N}=2$ theory defined by an M5-brane on $T^2 \subset U$ is given by:
\begin{itemize}
    \item[i)] the $\mathcal{N}=2$ $N_f=1$ $U(1)$ gauge theory in Example \ref{example_A1}.
    \item[ii)] the $A_2$ Argyres-Douglas theory in Example \ref{example_A2}.
    \item[iii)] the $\mathcal{N}=2$ $N_f=4$ $SU(2)$ gauge theory in Example \ref{example_cubic}. 
\end{itemize}
\end{example}

\begin{remark}
    The 4d $\mathcal{N}=2$ theory defined by an M5-brane on $T^2 \subset U$ is expected to define a UV complete 4d field theory when $U$ admits a complete hyperk\"ahler metric. It is the case, for example, in Examples \ref{example_4d_rk1} ii) and iii), where $U$ can be identified after hyperk\"ahler rotation with the total space of a Hitchin integrable system. As reviewed in \S \ref{section_string_open_CY}, $U$ admits only an incomplete hyperk\"ahler metric in general, and in such a case we have only a 4d low-energy effective field theory, as in \ref{example_4d_rk1} i). Nevertheless, we expect that the notion of BPS spectrum is still well defined in this situation, with a wall-crossing behavior still controlled by consistent scattering diagrams.
\end{remark}

\section{Further Examples}
\label{section_further_examples}

In this section, we provide several examples of consistent webs of 5-branes with 7-branes, along with their corresponding Symington polygons. We also describe the associated crepant resolution $\mathcal{X}\to \mathcal{X}^{\mathrm{can}}$ of the M-theory dual canonical 3-fold singularity $\mathcal{X}^{\mathrm{can}}$. We discuss examples engineering rank one 5d SCFTs in \S\ref{section_rank_one} and rank two 5d SCFTs in \S\ref{section_rank_two}. Additionally, we explain in \S\ref{section_fanosearch} how mirror symmetry for Fano orbifolds can be viewed as a particular instance of our construction.

\subsection{Rank one examples}
\label{section_rank_one}
The examples \ref{Ex:toric}-\ref{Ex: fig 23}-\ref{Ex: fig 24}-\ref{Ex: fig 25} below, as well as Example \ref{Ex: fig 29} in the following section, illustrate webs of $5$-branes with $7$-branes
engineering the rank one $E_0$ $5$d SCFT obtained by $M$-theory on local $\PP^2$.
\begin{example}
\label{Ex:toric}
Let $\overline{P}$ be the lattice polygon shown on the left of Figure \ref{Fig22}. The toric surface $(\overline{Y},\overline{D})$ with momentum polytope $\overline{P}$ is singular and has an $A_2$ singularity at each of the three 0-dimensional strata of $\overline{D}$. As discussed in \S\ref{section_degeneration_toric_surfaces}, the polyhedral decomposition $\mathcal{\overline{P}}$ of $\overline{P}$ in the middle of Figure \ref{Fig22} determines a one-parameter toric degeneration of $\overline{Y}$ into the union of three copies of $\PP^2$, corresponding to the three triangles forming $\overline{P}$.
According to \S\ref{section_toric_dual_CY3}, the mirror $\mathcal{\overline{X}}$ to this one-parameter degeneration of $\overline{Y}$ is the toric Calabi--Yau 3-fold whose fan is the cone over $(\overline{P},\mathcal{\overline{P}})$. In other words, the momentum polytope image of $\mathcal{\overline{X}}$ is the cone over the dual web illustrated in blue on the right Figure \ref{Fig22}. In this situation, we have $\mathcal{\overline{X}} = K_{\PP^2}$.
Moreover, the corresponding map $\mathcal{\overline{X}}=K_{\PP^2} \to \CC$ is determined by the section $s=xyz \in H^0(\PP^2, -K_{\PP^2})$, where $x,y,z$ denote the homogeneous coordinates of $\PP^2$. In particular, the intersection of $\PP^2$ with the other irreducible components of the central fiber $\overline{\mathcal{X}}_0$ is the 
union of the three lines forming the toric boundary of $\PP^2$ and defined by the equation $s=0$. Contracting the zero-section $\PP^2 \subset \mathcal{X}$ to a point, we obtain an affine toric canonical 3-fold singularity $\overline{\mathcal{X}}^{\mathrm{can}}$, which is the total space of a smoothing of a degenerate cusp singularity $\overline{\mathcal{X}}_0^{\mathrm{can}}$ with three irreducible components.
\begin{figure}[hbt!]
\center{\scalebox{0.9}{\includegraphics{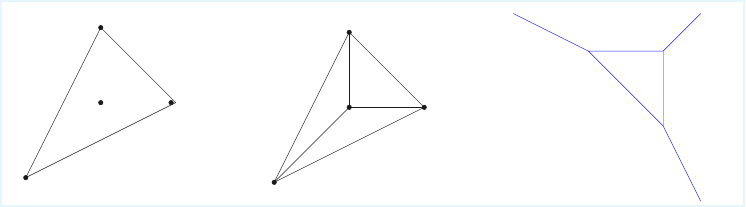}}}
\caption{On the left, the momentum polytope $\overline{P}$ for a polarized toric variety $\overline{Y}$. In the middle, a polyhedral decomposition $\overline{\mathcal{P}}$ of $\overline{P}$ defining a maximal degeneration of $\overline{Y}$.
On the right, the associated dual web of $5$-branes engineering the rank one $E_0$ 5d SCFT.}
\label{Fig22}
\end{figure}
\end{example}

\begin{example}
\label{Ex: fig 23}
Let $\overline{P}$ be the decorated lattice polygon shown on the left of Figure \ref{Fig23}. The corresponding log Calabi--Yau surface $(Y,D)$ is obtained by one interior blow-up from the toric surface $(\overline{Y},\overline{D})$ as in Example \ref{Ex:toric}.
An associated Symington polygon $P$ with a polyhedral decomposition is illustrated in the middle of Figure \ref{Fig23} and corresponds to a one-parameter degeneration of the corresponding polarized log Calabi--Yau surface $(Y^{\mathrm{pol}},D^{\mathrm{pol}})$ to a union of two copies of $\PP^2$. As described in \S\ref{section_cy3_construction}, the mirror $\mathcal{X}$ to this one-parameter degeneration is a non-toric deformation of $\overline{\mathcal{X}}=K_{\PP^2}$. In particular, the compact component of $\mathcal{X}_0$ corresponding to the bounded cell in the dual web on the right of Figure \ref{Fig23} is obtained from $\PP^2$ by smoothing one of the corners of its toric boundary. This compact component is still isomorphic to $\PP^2$, but its intersection with the other irreducible components of $\mathcal{X}_0$ is the non-toric divisor defined by the union of a line and a conic. It follows that we still have $\mathcal{X}=K_{\PP^2}$, but the map $\mathcal{X} \rightarrow\CC$ is a non-toric deformation of the toric morphism $\overline{\mathcal{X}} \rightarrow \CC$. The corresponding canonical 3-fold singularity $\overline{\mathcal{X}}^{\mathrm{can}}$ is the total space of a smoothing of a degenerate cusp singularity $\overline{\mathcal{X}}_0^{\mathrm{can}}$ with two irreducible components.

\begin{figure}[hbt!]
\center{\scalebox{0.9}{\includegraphics{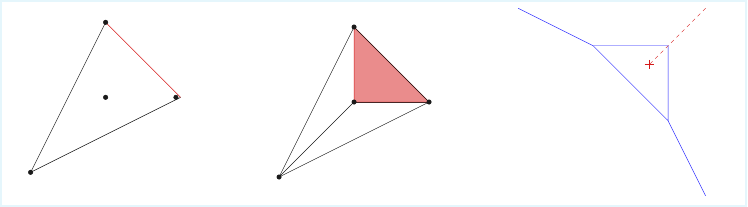}}}
\caption{On the left, the momentum polytope $\overline{P}$ for a polarized toric variety $\overline{Y}$. In the middle, a polyhedral decomposition $\overline{\mathcal{P}}$ of $\overline{P}$ defining a maximal degeneration of $\overline{Y}$.
On the right, the associated dual web of $5$-branes engineering the rank one $E_0$ 5d SCFT.}
\label{Fig23}
\end{figure}
\end{example}

\begin{example}
\label{Ex: fig 24}
Let $\overline{P}$ be the decorated lattice polygon shown on the left of Figure \ref{Fig24}. The corresponding log Calabi--Yau surface $(Y,D)$ is obtained by two interior blow-ups from the toric surface $(\overline{Y},\overline{D})$ as in Example \ref{Ex:toric}.
An associated Symington polygon $P$ with a polyhedral decomposition is illustrated in the middle of Figure \ref{Fig24} and corresponds to a one-parameter degeneration of the corresponding polarized log Calabi--Yau surface $(Y^{\mathrm{pol}},D^{\mathrm{pol}})$ to a non-normal surface whose normalization is isomorphic to $\PP^2$. As described in \S\ref{section_cy3_construction}, the mirror $\mathcal{X}$ to this one-parameter degeneration is a non-toric deformation of $\overline{\mathcal{X}}=K_{\PP^2}$. In particular, the compact component of $\mathcal{X}_0$ corresponding to the bounded cell in the dual web on the right of Figure \ref{Fig24} is obtained from $\PP^2$ by smoothing two of the corners of its toric boundary. This compact component is still isomorphic to $\PP^2$, but its intersection with the other irreducible components of $\mathcal{X}_0$ is the non-toric divisor defined by a nodal cubic curve. It follows that we still have $\mathcal{X}=K_{\PP^2}$, but the map $\mathcal{X} \rightarrow\CC$ is a non-toric deformation of the toric morphism $\overline{\mathcal{X}} \rightarrow \CC$. The corresponding canonical 3-fold singularity $\overline{\mathcal{X}}^{\mathrm{can}}$ is the total space of a smoothing of a degenerate cusp singularity $\overline{\mathcal{X}}_0^{\mathrm{can}}$ with a unique non-normal, self-intersecting, irreducible component.

\begin{figure}[hbt!]
\center{\scalebox{0.9}{\includegraphics{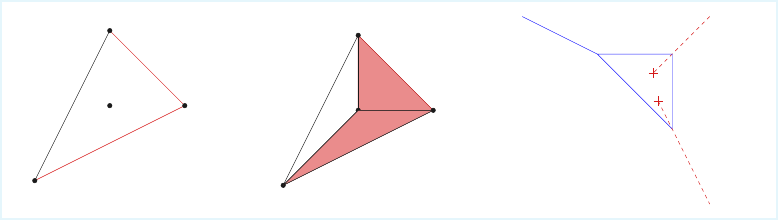}}}
\caption{On the left, the momentum polytope $\overline{P}$ for a polarized toric variety $\overline{Y}$. In the middle, a polyhedral decomposition $\overline{\mathcal{P}}$ of $\overline{P}$ defining a maximal degeneration of $\overline{Y}$.
On the right, the associated dual web of $5$-branes engineering the rank one $E_0$ 5d SCFT.}
\label{Fig24}
\end{figure}
\end{example}

\begin{example}
\label{Ex: fig 25}
Let $\overline{P}$ be the decorated lattice polygon shown on the left of Figure \ref{Fig25}. The corresponding log Calabi--Yau surface with line bundle $(Y,D, L)$ is obtained by three interior blow-ups from the toric surface $(\overline{Y},\overline{D})$ as in Example \ref{Ex:toric}.
An associated Symington polygon $P$ with a polyhedral decomposition is illustrated in the middle of Figure \ref{Fig25} and corresponds to a one-parameter degeneration of $(Y,D)$.
In this case, we have $L^2=0$, and contracting all curves having zero-intersection with $L$ produces a map $c: Y \rightarrow \CC$ which is an elliptic fibration.
As described in \S\ref{section_cy3_construction}, the mirror $\mathcal{X}$ is a non-toric deformation of $\overline{\mathcal{X}}=K_{\PP^2}$. In particular, the compact component of $\mathcal{X}_0$ corresponding to the bounded cell in the dual web on the right of Figure \ref{Fig25} is obtained from $\PP^2$ by smoothing 
the three corners of its toric boundary. This compact component is still isomorphic to $\PP^2$, but its intersection with the other irreducible components of $\mathcal{X}_0$ is the non-toric divisor defined by a smooth cubic curve. It follows that we still have $\mathcal{X}=K_{\PP^2}$, but the map $\mathcal{X} \rightarrow\CC$ is a non-toric deformation of the toric morphism $\overline{\mathcal{X}} \rightarrow \CC$. The corresponding canonical 3-fold singularity $\overline{\mathcal{X}}^{\mathrm{can}}$ is the total space of a simple elliptic singularity $\overline{\mathcal{X}}_0^{\mathrm{can}}$.

\begin{figure}[hbt!]
\center{\scalebox{.9}{\includegraphics{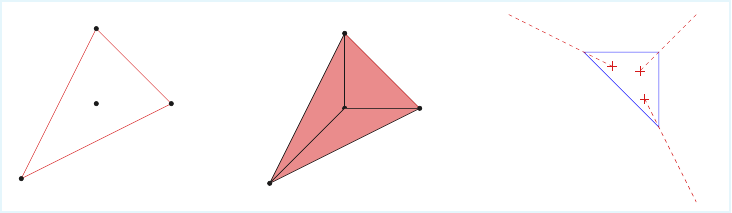}}}
\caption{On the left, the momentum polytope $\overline{P}$ for a polarized toric variety $\overline{Y}$. In the middle, a polyhedral decomposition $\overline{\mathcal{P}}$ of $\overline{P}$ defining a maximal degeneration of $\overline{Y}$.
On the right, the associated dual web of $5$-branes engineering the rank one $E_0$ 5d SCFT.}
\label{Fig25}
\end{figure}
\end{example}

The following example illustrates a web of $5$-branes with $7$-branes, describing the $\widetilde{E}_1$ $5$d SCFT obtained by $M$-theory on local $\mathbb{F}_0$.

\begin{example}
\label{Ex: fig 26}
Let $\overline{P}$ be the decorated lattice polygon shown on the left of Figure \ref{Fig26}, and $(\overline{Y},\overline{D})$ the toric surface with momentum polytope $\overline{P}$. The corresponding log Calabi--Yau surface $(Y,D)$ is obtained by a blow-up of $\overline{Y}$ along a smooth point on the component of the boundary divisor corresponding to the bottom edge of $\overline{P}$. 
An associated Symington polygon $P$ with a polyhedral decomposition is illustrated in the middle of Figure \ref{Fig26} and corresponds to a one-parameter degeneration of the corresponding polarized log Calabi--Yau surface $(Y^{\mathrm{pol}},D^{\mathrm{pol}})$ to the union of three copies of $\PP^2$. 
As described in \S\ref{section_cy3_construction}, the mirror $\mathcal{X}$ to this one-parameter degeneration is a non-toric deformation of the toric Calabi--Yau 3-fold $\overline{\mathcal{X}}=K_{\mathbb{F}_2}$. In particular, the compact component of $\mathcal{X}_0$ corresponding to the bounded cell in the dual web on the right of Figure \ref{Fig26} is obtained from $\mathbb{F}_2$ by smoothing one corner of its toric boundary. This compact component is isomorphic to $\mathbb{F}_0=\PP^1\times \PP^1$, with an intersection with the other irreducible components of $\mathcal{X}_0$ given by a non-toric divisor with three irreducible components. 
In particular, we have $\mathcal{X}=K_{\mathbb{F}_0}$, but with a non-toric morphism $\mathcal{X} \rightarrow \CC$. The corresponding canonical 3-fold singularity $\overline{\mathcal{X}}^{\mathrm{can}}$ is the total space of a smoothing of a degenerate cusp singularity $\overline{\mathcal{X}}_0^{\mathrm{can}}$ with three irreducible components.

\begin{figure}[hbt!]
\center{\scalebox{1.0}{\includegraphics{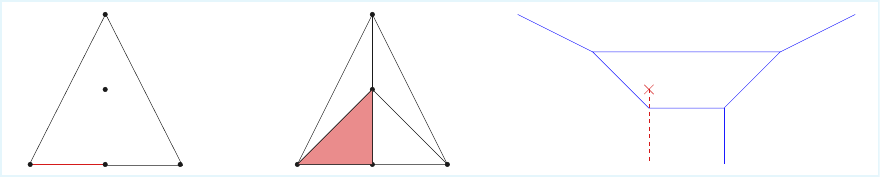}}}
\caption{The momentum polytope $\overline{P}$ for a polarized toric variety $\overline{Y}$ on the left, a polyhedral decomposition $\overline{\mathcal{P}}$ of $\overline{P}$ defining a maximal degeneration of $\overline{Y}$ in the middle, and the associated dual web of $5$-branes on the right, engineering the rank one $\widetilde{E}_1$ 5d SCFT.}
\label{Fig26}
\end{figure}

\end{example}

The following example illustrates a web of $5$-branes with $7$-branes, describing the $E_1$ $5$d SCFT obtained by $M$-theory on local $\mathbb{F}_1$.

\begin{example}
 Let $\overline{P}$ be the decorated lattice polygon shown on the left of Figure \ref{Fig36}, and $(\overline{Y},\overline{D})$ the toric surface with momentum polytope $\overline{P}$.
The corresponding log Calabi--Yau surface $(Y,D)$ is obtained by a blow-up of $\overline{Y}$ along two smooth points on two components of the boundary divisor corresponding to the most left, and bottom edges of $\overline{P}$.
 An associated Symington polygon $P$ with a polyhedral decomposition is illustrated in the middle of Figure \ref{Fig36} and corresponds to a one-parameter degeneration of the corresponding polarized log Calabi--Yau surface $(Y^{\mathrm{pol}},D^{\mathrm{pol}})$ to the union of two copies of $\PP^2$. As described in \S\ref{section_cy3_construction}, the mirror $\mathcal{X}$ to this one-parameter degeneration is a non-toric deformation of a toric Calabi--Yau 3-fold $\overline{\mathcal{X}}$. It follows from the dual web represented on the right of Figure \ref{Fig36} that the central fiber $\mathcal{X}_0$
consists of three irreducible component. One irreducible component is compact and isomorphic to $\PP^2$. The other two components are non-compact, one being toric and the other being non-toric and containing an interior $(-1)$-curve. Flopping this curve, we obtain the Calabi--Yau $3$-fold $K_{\mathbb{F}_1}$. The corresponding canonical 3-fold singularity $\overline{\mathcal{X}}^{\mathrm{can}}$ is the total space of a smoothing of a degenerate cusp singularity $\overline{\mathcal{X}}_0^{\mathrm{can}}$ with two irreducible components.

\begin{figure}[hbt!]
\center{\scalebox{1.0}{\includegraphics{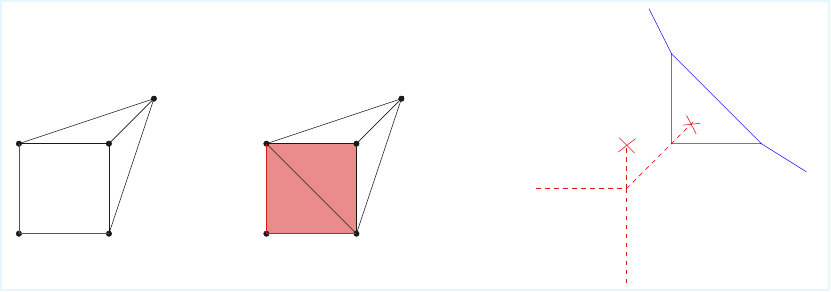}}}
\caption{The momentum polytope $\overline{P}$ for a polarized toric variety $\overline{Y}$ on the left, a polyhedral decomposition $\overline{\mathcal{P}}$ of $\overline{P}$ defining a maximal degeneration of $\overline{Y}$ in the middle, and the associated dual web of $5$-branes on the right, engineering the rank one $E_1$ 5d SCFT.}
\label{Fig36}
\end{figure}

\end{example}

\subsection{Rank two examples}
\label{section_rank_two}

\begin{example}
    Let $\overline{P}$ be the decorated lattice polygon shown on the left of Figure \ref{Fig27}, and $(\overline{Y},\overline{D})$ the toric surface with momentum polytope $\overline{P}$.
The corresponding log Calabi--Yau surface $(Y,D)$ is obtained by blowing up three smooth points on the three components of the toric boundary divisor $\overline{D}$ of $\overline{Y}$.
 An associated Symington polygon $P$ with a polyhedral decomposition is illustrated in the middle of Figure \ref{Fig27} and corresponds to a one-parameter degeneration of the corresponding polarized log Calabi--Yau surface $(Y^{\mathrm{pol}},D^{\mathrm{pol}})$ to the union of two copies of $\PP^2$. As described in \S\ref{section_cy3_construction}, the mirror $\mathcal{X}$ to this one-parameter degeneration is a non-toric deformation of a toric Calabi--Yau 3-fold $\overline{\mathcal{X}}$. It follows from the dual web represented on the right of Figure \ref{Fig27} that the central fiber $\mathcal{X}_0$ contains two compact irreducible components. 
 The compact component corresponding to the triangle in the web is isomorphic to $\PP^2$. On the other hand, the compact component corresponding to the other bounded cell of the web is isomorphic to the Hirzebruch surface $\mathbb{F}_6$.
The corresponding canonical 3-fold singularity $\overline{\mathcal{X}}^{\mathrm{can}}$ is the total space of a smoothing of a cusp singularity $\overline{\mathcal{X}}_0^{\mathrm{can}}$. 
The associated rank two 5d SCFT appears as $\PP^2 \cup \mathbb{F}_6$ in \cite[Figure 7]{classification}.
\end{example}

\begin{figure}[hbt!]
\center{\scalebox{0.9}{\includegraphics{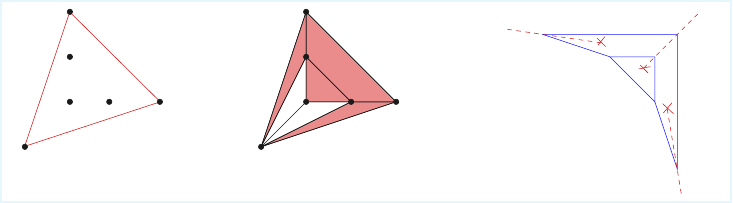}}}
\caption{The momentum polytope $\overline{P}$ for a polarized toric variety $\overline{Y}$ on the left, a polyhedral decomposition $\overline{\mathcal{P}}$ of $\overline{P}$ defining a maximal degeneration of $\overline{Y}$ in the middle, and the associated dual web of $5$-branes on the right, engineering a rank two 5d SCFT.}
\label{Fig27}
\end{figure}

\begin{example}
Let $\overline{P}$ be the decorated lattice polygon shown on the left of Figure \ref{Fig28}, and $(\overline{Y},\overline{D})$ the toric surface with momentum polytope $\overline{P}$.
The corresponding log Calabi--Yau surface $(Y,D)$ is obtained by blowing up four smooth points on the four components of the toric boundary divisor $\overline{D}$ of $\overline{Y}$.
 An associated Symington polygon $P$ with a polyhedral decomposition is illustrated in the middle of Figure \ref{Fig28} and corresponds to a one-parameter degeneration of the corresponding polarized log Calabi--Yau surface $(Y^{\mathrm{pol}},D^{\mathrm{pol}})$ to the union of two copies of $\PP^2$. As described in \S\ref{section_cy3_construction}, the mirror $\mathcal{X}$ to this one-parameter degeneration is a non-toric deformation of a toric Calabi--Yau 3-fold $\overline{\mathcal{X}}$. It follows from the dual web represented on the right of Figure \ref{Fig28} that the central fiber $\mathcal{X}_0$ contains two compact irreducible components. 
 The compact component corresponding to the quadrilatera in the web is isomorphic to $\mathbb{F}_1$. 
 On the other hand, the compact component corresponding to the other bounded cell of the web is isomorphic to $\mathbb{F}_6$.
The corresponding canonical 3-fold singularity $\overline{\mathcal{X}}^{\mathrm{can}}$ is the total space of a smoothing of a cusp singularity $\overline{\mathcal{X}}_0^{\mathrm{can}}$. 
The associated rank two 5d SCFT appears as $\mathbb{F}_1 \cup \mathbb{F}_6$ in \cite[Figure 8]{classification}, where it is also identified with the 5d $\mathcal{N}=1$ $\mathrm{Sp
}(2)_0$ gauge theory.
\end{example}

\begin{figure}[hbt!]
\center{\scalebox{0.9}{\includegraphics{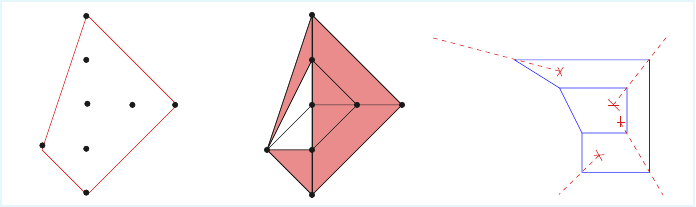}}}
\caption{The momentum polytope $\overline{P}$ for a polarized toric variety $\overline{Y}$ on the left, a polyhedral decomposition $\overline{\mathcal{P}}$ of $\overline{P}$ defining a maximal degeneration of $\overline{Y}$ in the middle, and the associated dual web of $5$-branes on the right, engineering a rank two 5d SCFT.}
\label{Fig28}
\end{figure}

\subsection{Relation to the Fanosearch program}
\label{section_fanosearch}
In \S\ref{section_cy3_construction}, we construct crepant resolutions $\mathcal{X} \rightarrow \mathcal{X}^{\mathrm{can}}$ of the M-theory dual 3-fold to a web of 5-branes with 7-branes as the mirror to degenerations of a log Calabi--Yau surface with line bundle $(Y,D,L)$, and explain that $\mathcal{X}^{\mathrm{can}}$ arises as a non-toric deformation of a toric Calabi--Yau $3$-fold $\overline{\mathcal{X}}^{\mathrm{can}}$. In addition, applying Hanany--Witten moves to the web of 5-branes gives rise to a different toric Calabi--Yau 3-folds $\overline{\mathcal{X}}^{\mathrm{can}}$ but to the same non-toric Calabi--Yau 3-fold $\mathcal{X}^{\mathrm{can}}$. In this section, we describe the relationship between this mirror construction, and mirror symmetry for 2-dimensional (possibly singular) Fano varieties, as developed in \cite{fanosearch, polygon_mutations, corti2023cluster}, as part of the Fanosearch program of Coates--Corti et al. Consequently, we explain the connection of our results with recent works in the physics literature \cite{arias2024geometry, Franco_twin}, where notions that appear in Fano mirror symmetry, such as mutations of Fano polygons, are applied to the study of 5d SCFTs defined by webs of 5-branes with 7-branes. 

We first briefly review mirror symmetry for 2-dimensional Fano varieties. 
A \emph{Fano polygon} is a lattice polygon $\overline{P}$ in $\ZZ^2$, which contains the origin $0 \in \ZZ^2$ in its strict interior, and the integral vectors $\overrightarrow{0 v}$ are primitive for all vertices $v$ of $\overline{P}$. Denote by $\overline{X}_{\overline{P}}$ the Fano toric surface whose fan is the spanning fan $\Sigma_{\overline{P}}$ of $\overline{P}$, that is, the fan in $\RR^2$ whose rays are span by the primitive vectors $\overrightarrow{0v}$, where $v$ are the vertices of $\overline{P}$. According to \cite{akhtar2014singularity}, every cone of $\Sigma_{\overline{P}}$ can be subdivided into elementary T-cones and R-cones. The elementary T-cones coincide with the cones over elementary triangles, with equal base and height as in Definition \ref{def_elementary_triangle}. On the other hand, the R-cones are the fans of toric cyclic singularities of class $R$, that is, without $\QQ$-Gorenstein smoothing.

It is shown in \cite[Lemma 6]{fanosearch} that there exists a Fano surface $X_{\overline{P}}$, whose singularities are the R-singularities corresponding to the R-cones of $\overline{P}$, and which is obtained as a $\QQ$-Gorenstein deformation of the toric surface $\overline{X}_{\overline{P}}$. The same Fano surface can be obtained from different Fano polygons. Indeed, if two Fano polygons $\overline{P}$ and $\overline{P'}$ are related by a combinatorial operation called \emph{mutation} \cite{polygon_mutations}, the resulting Fano surfaces $X_{\overline{P}}$ and $X_{\overline{P'}}$ are isomorphic, up to deformations, by \cite[Theorem 3]{fanosearch}. In other words, the same Fano surface admits two different degenerations to the two different 
Fano toric surfaces $\overline{X}_{\overline{P}}$ and $\overline{X}_{\overline{P}'}$. 

By \cite[Theorem 39]{corti2023cluster},
there exists actually a one-to-one correspondence between the set of $\QQ$-Gorenstein degenerations of the Fano surface $X_{\overline{P}}$ to Fano toric surfaces and the set of Fano polygons obtained from $\overline{P}$ by mutations. Finally, one can attach to a Fano polygon $\overline{P}$ a log Calabi--Yau surface $(Y,D)$. For this, denote by $(\overline{Y},\overline{D},\overline{L})$ the polarized toric variety with momentum polytope $\overline{P}$. Then, $(Y,D)$ is obtained from $(\overline{Y},\overline{D})$ by blowing up for each elementary T-cone of $\overline{P}$ a point on the corresponding toric divisor of $\overline{Y}$.
Finally, one can define a line bundle $L$ on $(Y,D)$ by $L=p^\star \overline{L} \otimes \mathcal{O}(-\sum_i a_i E_i)$, where $p: Y \rightarrow \overline{Y}$ is the toric morphism, $E_i$ are the exceptional curves, and $a_i$ are the common lattice height/base of the corresponding elementary T-cones. 
The set of cluster charts $(\CC^\star)^2 \subset U=Y \setminus D$ is in one-to-one correspondence with the set of mutations of $\overline{P}$, and the space of sections $H^0(Y,L)$ can be identified with the set of maximally mutable Laurent polynomials defined in \cite[Definition 4]{fanosearch}, which are expected to be Landau-Ginzburg mirror to the Fano surface $X_{\overline{P}}$.

The set-up of mirror symmetry for Fano surfaces briefly reviewed above can be viewed as a special case of the set-up considered in the present paper. Indeed, if $(Y,D,L)$ is a log Calabi--Yau surface with line bundle obtained from a Fano polygon $\overline{P}$, then, we have $L=\mathcal{O}(\sum_i a_i D_i)$, where $D_i$ are the irreducible components of $D$, and $a_i$ is the lattice distance between the origin $0$ and the side of $\overline{P}$
corresponding to $D_i$. In other words, we are in the particular case where the divisors defined by $L$ are linearly equivalent to a divisor supported on $D$. Conversely, if $(Y,D,L)$ is a log Calabi--Yau surface such that $L=\mathcal{O}(\sum_i a_i D_i)$ with $a_i \in \ZZ_{\geq 0}$, then, for every toric model given by a polarized toric surface $(\overline{Y}, \overline{D}, \overline{L})$ with momentum polytope $\overline{P}$, there exists a torus-invariant section of $\overline{L}$ with zero-divisor $\sum_i a_i \overline{D}_i$, which corresponds to a particular integral point in $\overline{P}$ which can be taken as the origin \cite{lai2022mirror}.
In this particular situation, the Fano polygon $\overline{P}$ naturally defines a Symington polygon $P$ by cutting out all elementary T-cones, so that the origin is the only singularity of the 
integral affine structure, see Figure \ref{Fig32}. The corresponding non-compact Calabi--Yau 3-fold is the total space of the canonical
line bundle $K_{X_{\overline{P}}}$
of the Fano surface $X_{\overline{P}}$. Indeed, starting from the toric Calabi--Yau 3-fold $K_{\overline{X}_{\overline{P}}}$, and applying the non-toric deformation described in \S\ref{section_cy3_construction}, the Fano toric surface $\overline{X}_{\overline{P}}$ is deformed non-torically into the Fano surface $X_{\overline{P}}$.
In particular, the corresponding canonical 3-fold singularity $\mathcal{X}^{\mathrm{can}}$ is obtained from $K_{X_{\overline{P}}}$ by contracting the zero section to a point. 
The 3-fold $K_{X_{\overline{P}}}$ is singular in general and so is only a partial crepant resolution of $\mathcal{X}^{\mathrm{can}}$. To obtain a crepant resolution, one should resolve the singularities by decomposing the Symington polygon $P$ into triangles of size one.

\begin{figure}[hbt!]
\center{\scalebox{0.9}{\includegraphics{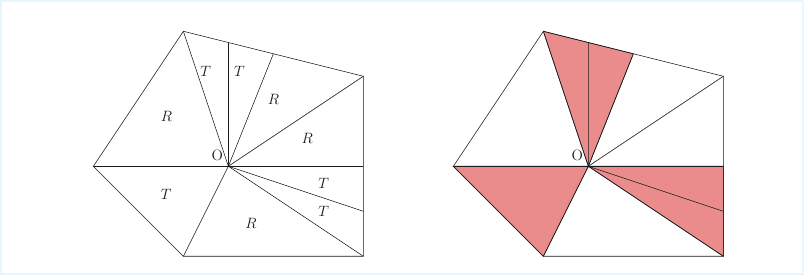}}}
\caption{On the left, a Fano polygon with a decomposition into elementary T-cones and R-cones. On the right, the corresponding Symington polygon obtained by cutting out the T-cones.}
\label{Fig32}
\end{figure}

Finally, mutations of Fano polygons correspond to Hanany--Witten moves on the web of 5-branes with 7-branes associated to the corresponding log Calabi--Yau surface $(Y,D,L)$, as observed in \cite{arias2024geometry, Franco_twin}. 
Different Fano polygons related by mutations correspond to two different ways to realize a Fano surface as a deformation of a toric Fano surface.
For general webs of 5-branes with 7-branes, corresponding to a general log Calabi--Yau surface $(Y,D,L)$, Hanany--Witten moves always have a dual description in terms of Symington polygons, which generalizes mutations of polygons in the particular case where $L$ is supported on $D$ and the Symington polygon comes from a Fano polygon. In general, different webs related by Hanany--Witten moves correspond to different ways to realize the M-theory dual Calabi--Yau 3-fold $\mathcal{X}^{\mathrm{can}}$ as a deformation of a toric Calabi--Yau 3-fold. 

In the following example, we describe how Fano polygons mirror to $\QQ$-Gorenstein degenerations of $\PP^2$ appear in the context of the present paper.

\begin{example}
\label{Ex: fig 29}
As shown in \cite{MR2581246}, all $\QQ$-Gorenstein degenerations of $\PP^2$ are weighted projective spaces $\PP(a^2,b^2,c^2)$, for Markov triples $(a,b,c)$, that is, triple of positive integers such that $a^2+b^2+c^2=3 abc$.
Correspondingly, all mutation equivalent Fano polygons to the Fano polygon $\overline{P}$ in Example \ref{Ex:toric}, can be characterized by a mutation sequence indexed by Markov triples as explained in \cite[Example 2.4]{MR3686766}. We illustrate in Figure \ref{Fig29} one of these mutation equivalent polygons $\overline{P}'$, along with a polyhedral decomposition which determines a degeneration of the toric variety $(\overline{Y}',\overline{D}')$ with momentum polytope image $\overline{P}'$. The mirror $\mathcal{X} \to \CC$ to this degeneration is a deformation of the toric Calabi--Yau 3-fold $K_{\PP(1,1,4)}$, and is still isomorphic to $K_{\PP^2}$, as in Example \ref{Ex:toric}. The compact component $\PP^2$ of $\mathcal{X}_0$ is obtained from the compact component $\PP(1,1,4)$ of $\overline{\mathcal{X}}_0$ by locally smoothing as in Equation \eqref{eq_hyper_smoothing} the elementary T-singularity $\CC^2/(\ZZ/4 \ZZ)$ corresponding to the elementary T-cone drawn in red in Figure \ref{Fig29}.

\begin{figure}[hbt!]
\center{\scalebox{0.9}{\includegraphics{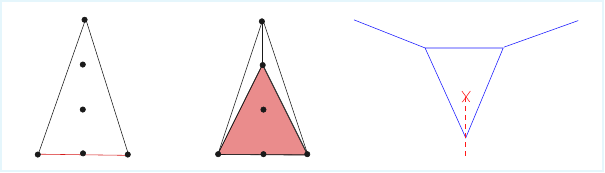}}}
\caption{The momentum polytope $\overline{P}'$ for a polarized toric variety $\overline{Y}'$ on the left, a polyhedral decomposition $\overline{\mathcal{P}}'$ of $\overline{P}'$ defining a maximal degeneration of $\overline{Y}'$ in the middle, and the associated dual web of $5$-branes on the right, engineering the rank one $E_0$ 5d SCFT.}
\label{Fig29}
\end{figure}

\end{example}

\textbf{Statements and declarations:}

\textbf{Conflict of interest statement:} The authors have no competing interest, directly or indirectly
related to this work, to declare.

\textbf{Data availability statement:} No data have been used or created for this work.

\bibliographystyle{plain}
\bibliography{bibliography}

\end{document}